\documentclass[a4paper,10pt]{amsart}

\usepackage[utf8]{inputenc}
\usepackage{amsthm,amssymb,dsfont,enumitem,mathtools}
\usepackage[hide links]{hyperref}
\usepackage[scr=rsfso]{mathalpha}
\usepackage[dvipsnames,svgnames]{xcolor}
\usepackage[margin=3.65cm]{geometry}

\usepackage{xpatch}
\xpatchcmd{\proof}{\itshape}{\bfseries}{}{}

\hypersetup{
    hidelinks,
    pdftitle={Renormalization flows for 1D mixed states and a quantum Goursat lemma},
    pdfauthor={L\'eo Le-Nestour, David P\'erez-Garc\'ia, Alberto Ruiz-de-Alarc\'on},
    pdfkeywords={mathematical physics, matrix product density operators, tensor networks, renormalization flows, quantum channels, Hopf algebras, quantum hypergroups}
}

\allowdisplaybreaks

\newcommand{\Fine}{\mathcal{F}}
\newcommand{\Coarse}{\mathcal{C}}
\newcommand{\Noisy}{\mathcal{N}}

\newcommand{\eval}[2]{\langle #1 , #2\rangle}
\newcommand{\innerprod}[2]{\langle #1 \,|\, #2\rangle_{\smash{\hat{h}}}}
\newcommand{\basisE}{\xi}
\newcommand{\placeholder}{\,\cdot\,}

\newcommand{\newterm}[1]{\emph{#1}}

\usepackage[noabbrev, capitalize, nameinlink]{cleveref}
\creflabelformat{equation}{#2\textup{#1}#3}
\newtheorem{theorem}{Theorem}[section]
\newtheorem{lemma}[theorem]{Lemma}
\newtheorem{corollary}[theorem]{Corollary}
\newtheorem{goursat-lemma-num}[theorem]{Goursat's Lemma}
\newtheorem{proposition}[theorem]{Proposition}
\newtheorem*{goursat-lemma}{Goursat's Lemma}
\newtheorem*{first-main-result}{First main result}
\newtheorem*{second-main-result}{Second main result}
\theoremstyle{definition}
\newtheorem{definition}[theorem]{Definition}
\newtheorem{definitions}[theorem]{Definitions}
\newtheorem{step-goursat}{Step}
\newtheorem{step-hypergroup}{Step}
\newtheorem{example}[theorem]{Example}
\newtheorem{examples}[theorem]{Examples}
\newtheorem{remark}[theorem]{Remark}
\newtheorem*{notation}{Notation}

\newtheorem*{theorem*}{Theorem}

\newtheorem*{acknowledgements}{Acknowledgements}
\newtheorem*{conflictofinterest}{Conflict of interest statement}
\newtheorem*{dataavailability}{Data availability statement}

\usepackage{tikz,tikz-cd}
\usetikzlibrary{decorations.pathmorphing, decorations.markings, arrows, arrows.meta, shapes,  patterns.meta, shapes}
\tikzset{baseline={([yshift=-.5ex]current bounding box.center)}}
\colorlet{Virtual}{RedOrange}
\tikzstyle{bevel} = [ preaction = { draw, white, line width=4pt,  line cap = round } ]
\tikzstyle{symb} = [ draw=black, fill=black, line width=0.4pt, inner sep=1.5pt ]
\tikzstyle{symb fdisk} = [ circle ]
\tikzstyle{symb disk} = [ circle ]
\tikzstyle{-mid} = [ decoration={ markings, mark = at position 0.50*\pgfdecoratedpathlength+0.6*3pt with \arrow{>[width=2pt]} }, postaction={decorate} ]
\tikzstyle{-third} = [ decoration={ markings, mark = at position 0.75*\pgfdecoratedpathlength+0.6*3pt with \arrow{>[width=2pt]} }, postaction={decorate} ]
\tikzstyle{mid-} = [ decoration={ markings, mark = at position 0.50*\pgfdecoratedpathlength+0.6*3pt with \arrow{<[width=2pt]} }, postaction={decorate} ]
\tikzstyle{third-} = [ decoration={ markings, mark = at position 0.75*\pgfdecoratedpathlength+0.6*3pt with \arrow{<[width=2pt]} }, postaction={decorate} ]

\tikzset{every path/.style={line cap=round}}
\tikzset{
   baseline = {([yshift=-.5ex]current bounding box.center)},
   every path/.style = {line cap=round},
   Vt/.style = {inner sep=0.5pt},
   Ti/.style = {inner sep=0.8pt},
   Sm/.style = {inner sep=1pt},
   No/.style = {inner sep=1.5pt},
   La/.style = {inner sep=1.75pt},
   Hu/.style = {inner sep=2pt},
   disk/.style   = {circle, draw=black, fill=black},
   fdisk/.style  = {circle, draw=black},
   star8/.style  = {star, star points=8, draw=black, fill=black, star point height=1.5pt, scale=0.65},
   fstar8/.style = {star, star points=8, draw=black, star point height=1.5pt, thin, scale=0.65},
   bevel/.style  = {preaction={draw, white, line width=#1, line cap=round}},
   bevel/.default = 3pt,
   thick bevel/.style = {bevel=5pt},
   thin bevel/.style = {bevel=2pt},
   -mid/.style = {decoration={markings, mark=at position 0.5*\pgfdecoratedpathlength+1.7pt with \arrow{>[width=3pt]}}, postaction={decorate}},
   mid-/.style = {decoration={markings, mark=at position 0.5*\pgfdecoratedpathlength+1.1pt with \arrow{<[width=3pt]}}, postaction={decorate}}
}
\tikzdeclarepattern{
  name = rising,
  bounding box = {(0,-0.4pt) and (1pt,0.4pt)},
  tile size = {(1pt,1pt)},
  tile transformation = {rotate=45},
  code = {
    \draw[line width=0.5pt] (0,0) -- (1pt,0);
  },
}
\tikzdeclarepattern{name = rising2,
   bounding box = {(0,-0.75pt) and (2pt,0.75pt)},
   tile size = {(2pt, 2pt)},
   tile transformation = {rotate=45},
   code = {\draw[line width=1.5pt] (0,0) -- (2pt,0);}
}
\tikzdeclarepattern{name = falling,
   bounding box = {(0,-0.75pt) and (2pt,0.75pt)},
   tile size = {(2pt, 2pt)},
   tile transformation = {rotate=-45},
   code = {\draw[line width=1.5pt] (0,0) -- (2pt,0);}
}
\usepackage[outline]{contour}

\usepackage[activate={true,nocompatibility},final,tracking=true,kerning=true,spacing=true,factor=1100,stretch=10,shrink=10]{microtype}
\SetTracking{encoding={*}, shape=sc}{15}
\microtypecontext{spacing=nonfrench}

\title[Renormalization flows for 1D mixed states and a quantum Goursat lemma]{Renormalization flows for 1D mixed states\\and a quantum Goursat lemma}
\author[L.~Le-Nestour]{L\'eo Le-Nestour$^1$}
\author[D.~P\'erez-Garc\'ia]{David P\'erez-Garc\'ia$^{2,3}$}
\author[A.~Ruiz-de-Alarc\'on]{Alberto Ruiz-de-Alarc\'on$^{4}$}
\thanks{\\[-1em]
$^1$~Fourier Institute, Université Grenoble Alpes, 38400 Saint-Martin-d'Hères, France.\\
$^2$~Faculty of Mathematics, Universidad Complutense de Madrid, 28040 Madrid, Spain.\\
$^3$~Instituto de Ciencias Matemáticas, CSIC-UAM-UC3M-UCM, 28049 Madrid, Spain.\\
$^4$~Department of Mathematics, CUNEF Universidad, 28040 Madrid, Spain.}


\begin{document}

\begin{abstract}
Renormalization provides a framework for relating microscopic models of physical systems to effective descriptions at larger length scales. This procedure is studied for the boundary states of non-chiral two-dimensional topologically ordered models. The initial data consist of renormalization fixed points built from representations of finite-dimensional $C^*$-Hopf algebras, which are then perturbed by uniform on-site noise quantum channels and repeatedly coarse-grained. The resulting flows admit an intrinsic algebraic description in terms of completely positive maps on the $C^*$-Hopf algebra or, equivalently, positive linear functionals on its enveloping $C^*$-Hopf algebra. Their iteration is governed by convolution powers, and convergent trajectories yield new matrix product density operator fixed points, described by finite $*$-quantum hypergroups. This provides a concrete physical interpretation of such structures. For finite group algebras and their duals, we provide explicit classifications via Goursat's lemma for groups. Finally, we formulate and prove a quantum generalization of Goursat's lemma for finite-dimensional $C^*$-Hopf algebras, a result of independent interest, which gives an explicit structural description of all convergent renormalization trajectories.
\end{abstract}

\maketitle

\section{Introduction}
\label{sec:intro}

The principles of renormalization have proven to be among the most fruitful in modern theoretical physics.
The fundamental underlying idea consists in replacing a microscopic description by effective descriptions at larger scales, keeping track of the degrees of freedom and structures that remain relevant at long distances.
In its modern form, this viewpoint goes back to the work of Gell-Mann and Low on scale-dependent couplings \cite{gell-mann_quantum_1954}, Kadanoff's real-space block-spin picture \cite{kadanoff_scaling_1966}, and Wilson's formulation of renormalization as a dynamical flow in a space of effective theories \cite{wilson_renormalization_1971,wilson_renormalization_1971-1}.
Wilson's numerical renormalization group gave a particularly successful implementation for impurity problems such as the Kondo model \cite{wilson_renormalization_1975}.
For strongly correlated quantum systems, real-space renormalization methods became systematically effective with White's density matrix renormalization group (DMRG), which identified the reduced density matrix as the right object for retaining the degrees of freedom most entangled with the environment \cite{white_density_1992}.

Tensor networks provide a natural and versatile framework for such settings, encoding the patterns of correlations of a state using contractions of tensors arranged according to the underlying geometry.
In particular, matrix product states, originally introduced in 1992 as finitely correlated states \cite{fannes_finitely_1992}, provide a variational manifold for one-dimensional quantum systems and give a conceptual explanation of the success of DMRG \cite{ostlund_thermodynamic_1995,schollwock_density-matrix_2005,schollwock_density-matrix_2011}.
Since then, tensor network renormalization methods have developed in several directions.
The multiscale entanglement renormalization ansatz (MERA), for instance, coarse-grains a state across scales using disentanglers and isometries \cite{vidal_entanglement_2007,vidal_class_2008}, while the tensor renormalization group (TRG) and tensor network renormalization (TNR) methods provide coarse-graining schemes for tensor networks associated with two-dimensional classical partition functions and related quantum systems \cite{levin_tensor_2007,evenbly_tensor_2015,evenbly_tensor_2015-1}.
More recently, rigorous renormalization schemes closely related to tensor network coarse-graining have been developed for classical lattice models, including convergence and stability results \cite{kennedy_tensor_2022,kennedy_tensor_2024,ebel_tensor_2025,ebel_rotations_2025}. These are remarkable achievements since obtaining mathematically rigorous results on renormalization processes is a notoriously difficult task. 

Furthermore, tensor networks have also become powerful analytical tools, providing a language for the structure and classification of quantum many-body states \cite{cirac_matrix_2021}. The underlying idea is that all global properties of a tensor network should be characterizable in terms of its defining local tensors. Since tensor networks provide accurate approximations to the ground and thermal states of locally interacting systems, local tensors emerge as central objects in the study and classification of quantum phases of matter. 
This idea was first exploited in the case of symmetry-protected topological phases in one dimension, where global symmetries in matrix product states were encoded in the local tensors as virtual projective representations, giving rise to the cohomological classification of symmetry-protected phases \cite{chen_classification_2011,schuch_classifying_2011,pollmann_symmetry_2012}. In two dimensions, intrinsic topological order in projected entangled pair states (PEPS) was shown to emerge from purely virtual symmetries in the local tensors \cite{schuch_peps_2010,bultinck_anyons_2017,sahinoglu_characterizing_2021,cirac_matrix_2021,molnar_matrix_2022, ogata2025haag}. To capture all known topological phases in 2D, those virtual symmetries must in general be associated to representations of algebraic structures beyond groups; in particular $C^*$-weak Hopf algebras \cite{molnar_matrix_2022, ogata2025haag}, whose representations correspond to unitary (multi)fusion categories and are naturally expressed as matrix product operators.

The fact that global properties of interest turn out to be associated with symmetries occurring at the virtual level is just one side of a more general holographic principle, which crystallizes in the notion of boundary state: a mixed state living at the virtual degrees of freedom of the tensor network which encodes all its relevant properties, including the topological content, but also the correlations or the spectral gap \cite{cirac_matrix_2017, cirac2011entanglement, kastoryano2019locality, ruiz_matrix_2024}. For the case of 2D topologically ordered tensor networks which are renormalization fixed points, such as Kitaev or Levin-Wen models \cite{kitaev_fault-tolerant_2003, levin_string-net_2005}, those boundary states are themselves renormalization fixed point matrix product density operators, encoding in their defining tensors all topological features of the bulk and, remarkably, giving rise to all known gapped mixed-state phases of matter in 1D \cite{sun_anomalous_2025}.

It is then natural to ask how these ideal one-dimensional fixed points sit within the larger space of local tensors and how that space is organized under coarse-graining.
Renormalization flows of matrix product density operators and the resulting classification of one-dimensional mixed-state phases have only recently begun to be explored \cite{cirac_matrix_2017,ruiz_matrix_2024,coser_classification_2019,sun_anomalous_2025,liu_trading_2025,liu_establishing_2026}. 
Here, a uniform on-site quantum channel acts as a local perturbation of the underlying tensor, and iterating perturbation and coarse-graining traces a trajectory through this space.
Two interrelated questions guide our analysis.
First, we ask whether these fixed points are stable under reasonable perturbations, since exact algebraic fixed-point identities are not expected to survive generic noise, and generic noise need not preserve the special correlations of a topological boundary fixed point.
Second, we ask which mathematical structures arise from the renormalization flows and corresponding fixed points generated by such perturbations.
A closely related renormalization flow for matrix product density operators was recently studied by Kato, who shows that convergence of such flows imposes a pre-bialgebra structure on the tensors \cite{kato_exact_2024}.
Our approach here is complementary: starting from boundary-state fixed points and perturbing them, we identify the resulting renormalization dynamics with convolution powers and classify all possible convergent limits.

In the present paper we focus on finite-dimensional $C^*$-Hopf algebras.
This class extends the group-algebra and dual-group-algebra cases, captures the basic mechanisms behind the renormalization flow and allows for a complete algebraic formulation. 
Our first main observation is that the physical noise admits an intrinsic algebraic description: an on-site quantum channel induces a completely positive map on the $C^*$-Hopf algebra, and the composition of independently perturbed blocks is governed by a convolution product on these maps.
Equivalently, the same data can be encoded by positive linear functionals on the enveloping $C^*$-Hopf algebra $A\otimes A^{\mathrm{op}}$.
In this dual formulation, the renormalization trajectory is simply described by powers of the corresponding functional.

Moreover, if a trajectory converges under repeated coarse-graining, its limit is described by an idempotent state on $A\otimes A^{\mathrm{op}}$, and every idempotent state satisfying the appropriate marginal condition defines a matrix product density operator renormalization fixed point of the kind considered here.
We then reframe the known correspondence between idempotent states, group-like projections and coideal $*$-subalgebras \cite{Franz_2009,franz_new_2009,salmi_idempotent_2012}.
In this way, the limiting fixed points acquire a natural algebraic description in terms of coideal $*$-subalgebras of the enveloping $C^*$-Hopf algebra or, equivalently, in terms of finite $*$-quantum hypergroups.
\begin{first-main-result}[\cref{th:NewFixedPoints}]
Renormalization fixed points arising from boundary theories of two-dimensional non-chiral topologically ordered quantum spin systems, based on $C^*$-Hopf algebras, after arbitrary on-site noise and coarse-graining, are effectively described by finite $*$-quantum hypergroups.
\end{first-main-result}
More precisely, the group-like projections associated with the limiting idempotent states determine such structures in the sense of Refs.~\cite{vandaele_fourier_2007,landstad_finite_2025}.
Thus, finite $*$-quantum hypergroups appear here not as an additional assumption, but as the algebraic structures naturally selected by the renormalization flows.
This provides a concrete physical interpretation of these objects within the tensor network description of mixed states and boundary theories.

In the finite-group cases, the above abstract description becomes very explicit. For group algebras and their duals, idempotent states correspond to subgroups \cite{kawada_probability_1940}. Since the relevant enveloping algebra involves a product of two copies of the group, the classification of attainable fixed points reduces to the classification of subgroups of a direct product, and Goursat's Lemma 
gives precisely such a classification \cite{goursat_sur_1889}:
\begin{goursat-lemma}
The subgroups of a direct product of two groups are uniquely described by the choice, in each factor, of a subgroup and a normal subgroup contained in it, together with an isomorphism between the two resulting quotient groups.
\end{goursat-lemma}
\noindent%
We use it to describe all limiting fixed points for finite group algebras and their duals.

In the general case, it is natural to look for an algebraic counterpart of Goursat's Lemma that classifies the finite $*$-quantum hypergroups contained in a tensor product of two $C^*$-Hopf algebras directly in terms of data attached to each factor. This is the content of our second main result, which parallels the classical statement closely.
\begin{second-main-result}[\cref{thm:QuantumGoursat}]
The finite $*$-quantum hypergroups contained in the tensor product of two $C^*$-Hopf algebras are uniquely described by, in each factor, a finite $*$-quantum hypergroup and an appropriate group-like projection contained in it, together with an anti-isomorphism between the corner $C^*$-Hopf algebras arising from such projections.
\end{second-main-result}
\noindent%
In particular, this result provides a structural description of the renormalization fixed points, extending the group-theoretic classification. To the best of our knowledge, this is the first extension of Goursat's Lemma to the context of $C^*$-Hopf algebras and, as such, we expect it to find further applications in the future.

The paper is organized as follows. In \cref{sec:preliminaries} we review quantum spin chains, matrix product operators and the theory of finite-dimensional $C^*$-Hopf algebras used throughout the paper. We also recall the construction of the matrix product density operators that serve as the starting fixed points, together with their fine-graining and coarse-graining channels. In \cref{sec:flows} we introduce uniform on-site perturbations by quantum channels and translate the resulting renormalization trajectories into maps on the $C^*$-Hopf algebra and positive functionals on the enveloping $C^*$-Hopf algebra. In \cref{sec:fixedpoints} we characterize the convergent limits of these trajectories through idempotent states, group-like projections, coideal $*$-subalgebras and finite $*$-quantum hypergroups. In \cref{sec:groups} we specialize the construction to finite group algebras and their duals, obtaining an explicit classification using Goursat's Lemma. In \cref{sec:quantumgoursat} we prove a quantum generalization of Goursat's Lemma and apply it to describe the renormalization fixed points in the general case.

\section{Preliminaries}
\label{sec:preliminaries}

\subsection{Quantum spin systems and tensor networks}
\label{subsec:quantumsys}

A finite quantum spin chain is a one-dimensional lattice with a finite-dimensional quantum system at each site. First, let us fix the corresponding notation.
\begin{notation}
For a chain of length $n$, we write $\{1,\ldots,n\}$ for the set of sites and use the same local Hilbert space $\mathscr{H}\cong\mathbb{C}^d$ throughout.
The Hilbert space associated with a finite region $\Lambda$ is therefore the tensor product Hilbert space $\mathscr{H}^{\otimes|\Lambda|}$.
We fix an orthonormal basis $\basisE_1,\ldots,\basisE_d\in\mathscr{H}$ and denote by $\basisE^1,\ldots,\basisE^d\in\mathscr{H}^*$ the corresponding dual basis, with respect to the canonical pairing $\eval{\placeholder}{\placeholder}:\mathscr{H}^*\times\mathscr{H}\to\mathbb{C}$.
\end{notation}

One of the fundamental difficulties in quantum many-body theory is that the dimension of this Hilbert space grows exponentially with the system size.
Tensor network states address this obstruction by replacing the coefficients of a single many-body tensor with the contraction of simpler local tensors.
In one dimension, this principle leads to matrix product states.
Their success is tied not only to this economical parametrization, but also to the entanglement structure of low-energy one-dimensional systems, since the ground states of gapped local Hamiltonians satisfy an area law and matrix product states are designed to represent states with such controlled bipartite entanglement \cite{hastings_area_2007,verstraete_matrix_2008}.

In order to define them, let us first introduce an auxiliary finite-dimensional Hilbert space $\mathscr{K}\cong \mathbb{C}^{\chi}$, called the \newterm{virtual space}.
Its dimension, $\chi$, is known as the \newterm{bond dimension} and it controls the size of the internal indices along which the network is contracted, hence the amount of correlation that can be transmitted from one site to the next.
We also consider a given tensor $A \in \operatorname{End}\mathscr{K} \otimes \mathscr{H}$, which may be equivalently read as a collection of matrices, one for each physical basis vector, by setting $ A_i \coloneqq (\mathrm{Id}_{\mathscr{K}}\otimes \basisE^i)(A) \in\operatorname{End}\mathscr{K}$, for all $i=1,\ldots,d$.
Then, for any given $B\in\operatorname{End}\mathscr{K}$, the vector on $n$ sites defined by
\[
    \psi_n(A,B)
    \coloneqq
    \sum_{i_1,\ldots,i_n=1}^{d}
        \operatorname{Tr}(B A_{i_1} \cdots A_{i_n})
        \basisE_{i_1}\otimes \cdots\otimes \basisE_{i_n}
    \in \mathscr{H}^{\otimes n}
\]
is called the \newterm{matrix product vector} generated by the tensor $A$ with \newterm{boundary condition} $B$.

The same philosophy applies to operators, with one physical index for the input and one for the output at each site. 
Matrix product operators provide local descriptions of Hamiltonians, observables, mixed states, transfer matrices and discrete time-evolution steps preserving the underlying tensor network structure \cite{verstraete_matrix_2008,pirvu_matrix_2010,cirac_unitaries_2017,cirac_matrix_2021}.
They are par\-ti\-cu\-lar\-ly relevant here because one-dimensional mixed states, including boundary theories of two-dimensional tensor network states, can often be described in this form \cite{zwolak_mixed-state_2004,cirac_matrix_2017,ruiz_matrix_2024}.
More concretely, the tensor $A$ is replaced by a given tensor $M \in \operatorname{End}\mathscr{K} \otimes \operatorname{End}\mathscr{H}$, which, as above, can be equivalently read as the collection of matrices $M_{ij} \coloneqq (\mathrm{Id}_{\mathscr{K}}\otimes \mathrm{Tr}\circ \basisE_i \basisE^j)(M) \in \operatorname{End}\mathscr{K}$.
Then, for any given $B\in \operatorname{End}\mathscr{K}$, the ope\-ra\-tor on $n$ sites defined by
\[
    \rho_{n}(M,B)
    \coloneqq
    \sum_{\substack{i_1,\ldots,i_n=1\\j_1,\ldots,j_n=1}}^{d}
        \operatorname{Tr}(B M_{i_1 j_1} \cdots M_{i_n j_n})
        \basisE_{i_1}\basisE^{j_1}\otimes \cdots\otimes \basisE_{i_n}\basisE^{j_n}
    \in \operatorname{End}\mathscr{H}^{\otimes n}
\]
is called the \newterm{matrix product operator} generated by $M$ with boundary condition $B$.
Furthermore, if $\rho_n(M,B)\geq 0$ and $\operatorname{Tr}\rho_n(M,B)=1$ for all system sizes $n\in\mathbb{N}$, these operators are called \newterm{matrix product density operators}.
The matrix product density operators studied in this work will be constructed from finite-dimensional $C^*$-Hopf algebras, which we introduce next.

\subsection{Algebraic framework}
\label{subsec:algebras}

This subsection provides a brief review of the basic notions of the theory of coalgebras and $C^*$-Hopf algebras that will be used throughout the paper.
For more background, we refer the reader to
Refs.~\cite{montgomery_representation_2001, montgomery_hopf_1993, bohm_weak_1999, etingof_fusion_2005, etingof_tensor_2015}.

\begin{definition}
\label{def:coalgebra}
A \newterm{coalgebra} is a vector space $A$ over the complex numbers equipped with a linear map $\Delta:A\to A\otimes A$ called the \newterm{comultiplication}, which is \newterm{coassociative}, i.e.,
\begin{equation}
    (\Delta\otimes\mathrm{Id})\circ\Delta =
    (\mathrm{Id}\otimes\Delta)\circ\Delta,
\end{equation}
and a linear map $\varepsilon:A\to \mathbb{C}$, called the \newterm{counit}, satisfying the compatibility identity
\begin{equation}\label{eq:axiom_eps_id}
    (\varepsilon\otimes\mathrm{Id})\circ\Delta
    =
    \mathrm{Id}
    =
    (\mathrm{Id}\otimes\varepsilon)\circ\Delta.
\end{equation}
\end{definition}

\begin{notation}
We define recursively the iterated comultiplications as follows:
\[
    \Delta^{(0)} \coloneqq
        \mathrm{Id}
    , \quad
    \Delta^{(1)} \coloneqq
        \Delta
    ,\quad\text{and}\quad
    \Delta^{(n+1)} \coloneqq
        (\Delta\otimes \mathrm{Id}^{\otimes n})\circ \Delta^{(n)}
        \text{ for }n\geq 1.
\]
We also use \newterm{Sweedler's sumless notation}: for all $x\in A$, we write
\[
    \Delta(x) \eqqcolon x_{(1)}\otimes x_{(2)},
\]
with the summation over tensor components suppressed, i.e., if $\Delta(x)=\sum_i x_{(1),i}\otimes x_{(2),i}$ explicitly, then in Sweedler's notation the summation symbol and the index $i$ are omitted.
For example, for any $x\in A$ and $\varphi\in A^*$, the expression $(\varphi\otimes\mathrm{Id})(\Delta(x))$ is written as $\eval{ \varphi}{x_{(1)}} x_{(2)}$.
More generally, for iterated coproducts we write
\[
    \Delta^{(n)}(x)
    \eqqcolon
    x_{(1)} \otimes x_{(2)} \otimes \cdots \otimes x_{(n+1)}.
\]
When two independent coproducts of the same element appear, we distinguish them with primed indices. We denote the \newterm{opposite coproduct} as
\[
    \Delta^{(n)}_{\mathrm{op}}(x)
    \coloneqq
    x_{(n+1)} \otimes x_{(n)} \otimes \cdots \otimes x_{(1)}.
\]
\end{notation}

\begin{definition}
    \label{def:HA}
    A \newterm{$C^*$-Hopf algebra} is a finite-dimensional $C^*$-algebra $A$ over the complex numbers equipped with the structure of a coalgebra, for which the comultiplication is a $*$-algebra homomorphism, i.e., for all $x,y\in A$,
	\begin{gather}
        \label{eq:DeltaHomo}
        (xy)_{(1)} \otimes (xy)_{(2)}
        =
        x_{(1)} y_{(1)} \otimes x_{(2)} y_{(2)},
        \quad 
        1_{(1)}\otimes 1_{(2)}
        =
        1\otimes 1,
        \\
        \text{and}
        \quad
        (x^*)_{(1)}\otimes (x^*)_{(2)}
        =
        (x_{(1)})^* \otimes (x_{(2)})^*,
    \end{gather}
	the counit, $\varepsilon\in A^*$, is a $*$-algebra homomorphism, i.e., for all $x,y\in A$,
	\begin{equation}
        \label{eq:axiom_eps}
        \eval{\varepsilon}{xy} = \eval{\varepsilon}{x}\eval{\varepsilon}{y},
        \quad
        \eval{\varepsilon}{1}= 1,
        \quad 
        \text{and}
        \quad
        \eval{\varepsilon}{x^*}
        =
        \overline{\eval{\varepsilon}{x}},
    \end{equation}
	and there exists a linear map $S:A\to A$, called the \newterm{antipode}, such that, for all $x\in A$,
	\begin{equation}
		\label{eq:axiomSeps}
        S(x_{(1)}) x_{(2)}
        =
        \eval{\varepsilon}{x} 1
        =
        x_{(1)} S(x_{(2)}).
	\end{equation}
\end{definition}

\begin{definition}\label{def:Aop}
Let $A$ be a $C^*$-Hopf algebra.
\begin{enumerate}
\item[(i)]
The \newterm{opposite $C^*$-Hopf algebra} of $A$ is the $C^*$-Hopf algebra with \newterm{opposite multiplication}, $x\cdot_{\mathrm{op}} y \coloneqq y x$, for all $x,y\in A$, and it is denoted $A^{\mathrm{op}}$.
\item[(ii)]
The \newterm{enveloping $C^*$-Hopf algebra} of $A$ is the tensor product $A\otimes A^{\mathrm{op}}$.
\end{enumerate}
\end{definition}

\begin{definition}\label{def:rep}
    A \newterm{$*$-representation} of a $C^*$-Hopf algebra $A$ is a pair consisting of a fi\-ni\-te-di\-men\-sio\-nal Hilbert space $\mathscr{H}$ over $\mathbb{C}$ and a $*$-algebra homomorphism $\Phi:A\to\operatorname{End}\mathscr{H}$, i.e., a linear map satisfying, for all $x,y\in A$,
    \begin{equation}
        \Phi(x y) = \Phi(x) \Phi(y),
        \quad
        \Phi(1) = \mathrm{Id},
        \quad
        \text{and}
        \quad
        \Phi(x^*) = \Phi(x)^\dagger.
    \end{equation}
    A $*$-representation is said to be \newterm{faithful} if the map $\Phi$ is injective.
\end{definition}

Let us now briefly recall several well-known facts.
\begin{remark}
    $S$ is an involutive $*$-algebra antihomomorphism, i.e., for all $x,y\in A$,
    \begin{equation}
    \label{eq:PropsAntipode}
    	S^2(x) = x,
        \quad
        S(xy) = S(y) S(x),
        \quad
        S(x^*) = S(x)^*,
        \quad
        \text{and}
        \quad
        S(1) = 1,
    \end{equation}
    and it is \newterm{anticomultiplicative}, i.e., for all $x\in A$,
    \begin{equation}
    \label{eq:Santicomult}
    	S(x)_{(1)} \otimes S(x)_{(2)}
        =
        S(x_{(2)}) \otimes S(x_{(1)}).
    \end{equation}
\end{remark}

\begin{remark}\label{rem:dualA}
The dual linear space $A^* = \operatorname{Hom}(A,\mathbb{C})$ is canonically endowed with the structure of a $C^*$-Hopf algebra defined for all $\varphi,\psi\in A^*$ by the maps
\begin{equation*}
	\varphi \psi \coloneqq (\varphi \otimes \psi) \circ \Delta,
    \quad
    \varphi^* \coloneqq \overline{\eval{\varphi}{S(\placeholder)^*}},
    \quad
    \hat\Delta(\varphi)(x\otimes y) \coloneqq  \varphi(xy),
    \quad
    \hat{S}(\varphi)  \coloneqq  \varphi \circ S,
\end{equation*}
the unit of $A^*$ is $\varepsilon\in A^*$, and the counit $\hat\varepsilon\in A^{**}$ is given by $\hat\varepsilon(\varphi) \coloneqq \eval{\varphi}{1}$ for all $\varphi\in A^*$.
\end{remark}

For all $x\in A$, we define the left- and right-multiplication maps $L_x,R_x:A\to A$ by
\begin{equation}
    \label{eq:DefLRelement}
    L_x(a) \coloneqq xa,
    \quad
    R_x(a) \coloneqq ax.
\end{equation}
Similarly, for all $\varphi\in A^*$, we define the left- and right-slice maps $L_\varphi,R_\varphi:A\to A$ by
\begin{equation}
    \label{eq:DefLRfunctional}
    L_\varphi(a) \coloneqq \eval{\varphi}{a_{(1)}} a_{(2)},
    \quad
    R_\varphi(a) \coloneqq \eval{\varphi}{a_{(2)}}a_{(1)}.
\end{equation}
Note that $L_x\circ R_x = R_x \circ L_x$ and $L_\varphi\circ R_\varphi = R_\varphi \circ L_\varphi$.
We also define $c_x,c_\varphi:A\to A$ by
\begin{align}\label{eq:DefConjugations}
    c_x(a)
        & \coloneqq
            (R_x \circ L_x)(a) = xax,
    \\
    c_\varphi(a)
        & \coloneqq
            (R_\varphi \circ L_\varphi)(a)
        =
            \eval{\varphi}{a_{(1)}} a_{(2)}
                \eval{\varphi}{a_{(3)}}.
\end{align}

\begin{proposition}\label{prop:existHaar}
There exists a unique non-degenerate element $h\in A$, which is self-adjoint, idempotent, invariant under the antipode, and cocentral, i.e.,
\begin{equation}
\label{eq:Haar}
	h^2 = h^* = h = S(h),
    \quad
        \text{and}
        \quad
    h_{(1)}\otimes h_{(2)} = h_{(2)}\otimes h_{(1)},
\end{equation}
and is a two-sided integral, i.e., it satisfies the following identity for all $x\in A$:
\begin{equation}
\label{eq:Haar2}
    xh = hx = \eval{\varepsilon}{x} h.
\end{equation}
See e.g.~Refs.~\cite{montgomery_representation_2001,bohm_weak_1999}. The element $h$ is known as the \newterm{Haar integral} of $A$.
\end{proposition}

\begin{remark}
\cref{eq:Haar2} is equivalent to the so-called pulling-through identities
\begin{equation}\label{eq:PT}
	S(x) h_{(1)} \otimes h_{(2)} = h_{(1)} \otimes x h_{(2)},
	\quad
	h_{(1)} S(y) \otimes h_{(2)}  = h_{(1)} \otimes h_{(2)} y,
\end{equation}
for all $x,y\in A$. It is easy to prove using \cref{eq:PropsAntipode} that these are equivalent to
\begin{equation}
\label{eq:PTmodified}
	S(h_{(1)}) x\otimes h_{(2)} = S(h_{(1)})\otimes xh_{(2)},
	\quad
	y S(h_{(1)})\otimes h_{(2)}  = S(h_{(1)})\otimes h_{(2)}y,
\end{equation}
for all $x,y\in A$. See e.g.~Ref.~\cite[Lemma~3.2]{bohm_weak_1999}.
\end{remark}

\begin{remark}
The Haar integral of $A^*$, $\hat{h}\in A^*$, satisfies, for all $x,y\in A$,
\begin{equation}
\label{eq:DualHaar}
	\hat{h} \circ S
    =
    \hat{h},
    \quad
    \eval{\hat{h}}{xy}
    =
    \eval{\hat{h}}{yx},
    \quad
        \text{and}
        \quad
    \eval{\hat{h}}{x_{(1)}} \eval{\hat{h}}{x_{(2)}}
    =
    \eval{\hat{h}}{x},
\end{equation}
by \cref{rem:dualA,prop:existHaar}. Moreover, \cref{eq:Haar2} becomes, for all $x\in A$,
\begin{equation}
\label{eq:DualHaar2}
    \eval{ \hat{h} }{ x_{(1)} } x_{(2)} 
    =
    \eval{ \hat{h} }{ x } 1
    = 
    \eval{ \hat{h} }{ x_{(2)} } x_{(1)}.
\end{equation}
In particular, $\hat{h}$ is trace-like and $\eval{\hat{h}}{1} = 1$. Furthermore, for all $v,w\in A$,
\begin{equation}\label{eq:propHaarExchXY}
    S(v_{(1)}) \eval{\hat{h}}{v_{(2)}w} = w_{(1)} \eval{\hat{h}}{v w_{(2)}},
\end{equation}
see e.g.~Ref.~\cite[Proposition~1.1]{vandaele_algebraic_2025} for a proof.
We also call $\hat{h}$ the \newterm{Haar measure} of $A$.
\end{remark}

The Haar measure endows $A$ with the structure of a Hilbert space, with inner product
\begin{equation}\label{eq:hhatinnerprod}
    \innerprod{x}{y} \coloneqq \eval{\hat{h}}{x^* y}.
\end{equation}

\begin{remark}\label{rem:DaggerLR}
    With respect to the inner product in \cref{eq:hhatinnerprod},
    \[
        (L_x)^\dagger = L_{x^*},\quad (R_x)^\dagger = R_{x^*},\quad (L_\varphi)^\dagger = L_{\varphi^*},\quad\text{and}\quad (R_\varphi)^\dagger = R_{\varphi^*}.
    \]
    See \cref{appendix:proofs2} for a proof.
\end{remark}

\begin{remark}
The Haar measure can be written as
\begin{equation*}
    \hat{h}
    =
    \frac{1}{D} (d_1 \chi_1 + \cdots + d_r \chi_r),
\end{equation*}
where $\chi_1,\ldots,\chi_r\in A^*$ stand for the characters of the irreducible $*$-re\-pre\-sen\-ta\-tions of $A$, $d_1,\ldots,d_r>0$ are their corresponding dimensions, and we let $D  \coloneqq  \dim A$ \cite{montgomery_representation_2001, bohm_weak_2000,etingof_tensor_2015,molnar_matrix_2022}.
\end{remark}

\begin{remark}
Both integrals $h\in A$ and $\hat{h}\in A^*$ satisfy
\begin{equation}\label{eq:DualInt}
	\eval{\hat{h}}{h_{(1)}}h_{(2)} = \frac{1}{D} 1
    ,\quad\text{and}\quad
    \eval{\hat{h}}{(\placeholder)h} = \frac{1}{D} \varepsilon.
\end{equation}
\end{remark}

\begin{examples}
\label{examples:CG}
	(i)
	For any finite group $G$, the algebra $\mathbb{C}G$, together with
	\begin{equation*}
        \Delta(g)  \coloneqq  g\otimes g,
        \quad
        \varepsilon(g)  \coloneqq 
        1,
        \quad
        \text{and}
        \quad
        S(g)  \coloneqq  g^*  \coloneqq  g^{-1},
    \end{equation*}
	for all $g\in G$
	(or, more concretely, their linear extensions), becomes a $C^*$-Hopf algebra.
    The Haar integral of $\mathbb{C}G$ is given by the element
	\begin{equation*}
        h  \coloneqq \frac{1}{|G|}\sum_{g\in G} g.
    \end{equation*}
	(ii)
	The dual space $(\mathbb{C}G)^* \cong \mathbb{C}^G$ becomes a $C^*$-algebra with
	\begin{equation*}
        \delta_g \delta_k = \delta_{g,k} \delta_g,
        \quad
        (\delta_g)^* = \delta_g,
        \quad
        \text{and}
        \quad
        1_{\mathbb{C}^G} = \varepsilon,
    \end{equation*}
    for all $g,k\in G$, where $\delta_g\in (\mathbb{C}G)^*$ stands for the corresponding linear functional of the dual basis and $\delta_{g,k}\in\{0,1\}$ is the Kronecker delta. In addition, it is endowed with
    \begin{equation*}
    \hat\Delta(\delta_g)
         \coloneqq  \sum_{g_1 g_2 = g} \delta_{g_1} \otimes \delta_{g_2},
    \quad
    \hat\varepsilon(\delta_g)
         \coloneqq  \eval{\delta_g}{e},
    \quad
        \text{and}
        \quad
    \hat{S}(\delta_g)
         \coloneqq  \delta_{\smash{g^{-1}}},
    \end{equation*}
    for all $g\in G$, where $e$ stands for the unit of $G$, becoming a $C^*$-Hopf algebra.
    The dual Haar integral is given by
	\begin{equation*}
        \hat{h} = \delta_e.
    \end{equation*}
    Note that the property in \cref{eq:DualInt} becomes 
	$\eval{\hat{h}}{h_{(1)}} h_{(2)} = |G|^{-1} 1$ trivially.
\end{examples}

The following example, known as the \newterm{Kac–Paljutkin $C^*$-Hopf algebra}, is the smallest that is neither a group algebra nor the dual of a group algebra \cite{kac_finite_1966}.

\begin{example}\label{example:H8}
Let $H_8$ be the $*$-algebra generated by $x$, $y$, and $z$, subject to
\begin{equation}
    x^2 = y^2 = 1,
    \quad
    xy = yx,
    \quad
    z^2 = \frac{1}{2}(1 + x + y - xy),
    \quad
    zx = yz,
    \quad\text{and}\quad
    zy = xz,
\end{equation}
where $x^* = x$, $y^* = y$, and $z^*=z$.
We endow it with the comultiplication defined by
\begin{equation}
    \Delta(x) = x \otimes x,
    \quad 
    \Delta(y) = y \otimes y,
    \quad\text{and}\quad
    \Delta(z) = \frac{1}{2} ( z \otimes (1+x) + yz \otimes (1 - x) ),
\end{equation}
the counit defined by $\varepsilon(x) = \varepsilon(y) = \varepsilon(z) = 1$, and the antipode defined by $S(x) = x$, $S(y) = y$, and $S(z) = z$.
The formula for $\Delta(z)$ makes it clear that $\Delta$ is not cocommutative. The representation theory of $H_8$ shows that there are four one-dimensional irreducible $*$-representations and a two-dimensional irreducible $*$-representation, hence
\[
    H_8
    \cong
    \mathbb{C} \oplus \mathbb{C} \oplus \mathbb{C} \oplus \mathbb{C} \oplus \mathrm{Mat}_2(\mathbb{C})
    \cong \mathbb{C}^4 \oplus \mathrm{Mat}_2(\mathbb{C}),
\]
as $*$-algebras; in this picture, it is usually convenient to use the basis elements 
\[
    e_k \coloneqq
    \delta_{1k} \oplus \delta_{2k} \oplus \delta_{3k} \oplus \delta_{4k}
    \oplus
        \big(\begin{smallmatrix}
            \delta_{5k} & \delta_{8k}\\
            \delta_{7k} & \delta_{6k}
        \end{smallmatrix}\big),
\]
for $k=1,\ldots,8$; we refer the reader to e.g.~Ref.~\cite{pal_counterexample_1996} for the explicit definitions of the comultiplication, the counit, and the antipode for $e_1,\ldots,e_8$.
\end{example}

The existence of a Haar integral yields a canonical pairing between the algebra and its dual. This pairing allows one to associate to each element a linear functional in a natural way, giving rise to what is commonly called the \newterm{Fourier transform} \cite{vandaele_fourier_2007}.

\begin{lemma}
\label{lemma:fourier}
    Let $\jmath: A\to A^*$ be the linear map given, for all $x\in A$, by
	\begin{equation}
		\jmath(x)
		 \coloneqq 
		\eval{\hat{h}}{(\placeholder) x}.
	\end{equation}
	Then, $\jmath$ is a one-to-one correspondence and its inverse is given, for all $\varphi\in A^*$, by
	\begin{equation}
		\jmath^{-1}(\varphi)
		  \coloneqq 
		D \eval{\varphi}{h_{(1)}} S(h_{(2)}).
	\end{equation}
\end{lemma}

\noindent%
See \cref{appendix:proofs2} for a proof.

\begin{lemma}
\label{lemma:sigma}
    Let $\sigma: A\otimes A \to A$ be the linear map defined, for all $x,y\in A$, by
	\begin{equation}\label{eq:def_sigma}
		\sigma(x\otimes y)
        \coloneqq 
        \jmath^{-1} (\jmath(x) \jmath(y))
        =
        D \eval{\hat{h}}{h_{(1)}x} \eval{\hat{h}}{h_{(2)}y} S(h_{(3)}).
	\end{equation}
    Then, $\sigma$ is completely positive and it satisfies,
	for all $x,y,z\in A$,
	\begin{equation}
    \label{eq:prophsigma}
		\eval{\hat{h}}{\sigma(x\otimes y)z}
		=
		\eval{\hat{h}}{x z_{(1)}} \eval{\hat{h}}{y z_{(2)}};
	\end{equation}
    moreover, it is associative, i.e., for all $x,y,z\in A$,
    \[
        \sigma(\sigma(x\otimes y)\otimes z)
        =
        \sigma(x \otimes \sigma(y\otimes z)).
    \]
\end{lemma}

\noindent%
See \cref{appendix:proofs2} for a proof. We remark that, equivalently, $\sigma$ is the adjoint of $\Delta$ relative to the inner product induced by $\hat{h}$ on $A$. Let us also note that the expression for $\sigma(x\otimes y)$ coincides with the convolution product $(x,y)\mapsto \eval{\hat{h}}{S(y_{(1)}) x} y_{(2)}$ defined in Ref.~\cite{vandaele_fourier_2007}.

\subsection{Construction of renormalization fixed points}
\label{subsec:rfps}
In this subsection, we recall the construction of matrix product operators and renormalization fixed points.
Let $A$ be a $C^*$-Hopf algebra, let $\Phi:A\to\operatorname{End}\mathscr{H}$ be a faithful $*$-representation of $A$, and choose any faithful $*$-representation $\Psi:A^*\to\operatorname{End}\mathscr{K}$ of $A^*$. Let $e_1,\ldots,e_{\dim A}$ be a basis of $A$ and let $e^{1},\ldots, e^{\dim A}\in A^*$ stand for its dual basis.
The tensor generating the renormalization fixed points is given by
\begin{equation}
\label{eq:DefTensorRFP}
    M \coloneqq  \sum_{i=1}^{\dim A}  \Psi(e^i) \otimes \hat{b}(\hat{h}) \Phi(e_i),
\end{equation}
where $b:A \to \operatorname{Im}\Psi \subset \operatorname{End}\mathscr{K}$ is the unique linear map satisfying, for all $x\in A$ and $\varphi\in A^*$,
\begin{equation}\label{eq:DefBdry}
	\operatorname{Tr}(b(x) \Psi(\varphi)) = \eval{\varphi}{x};
\end{equation}
here, $b(x)$ is called the boundary condition associated with $x$.
By duality, for each $\varphi \in A^*$ there exists a unique $\hat{b}(\varphi)\in \operatorname{Im}\Phi \subset \operatorname{End}\mathscr{H}$ such that, for all $x\in A$,
\begin{equation}\label{eq:DefBdry-Dual}
	\operatorname{Tr}(\hat{b}(\varphi) \Phi(x)) = \eval{\varphi}{x}.
\end{equation}
For simplicity, we define the weighted representation of $x\in A$ as
\begin{equation}
	\Phi_b(x) \coloneqq \hat{b}(\hat{h}) \Phi(x).
\end{equation}
In particular, this allows us to write $\operatorname{Tr}\Phi_b(x)=\eval{\hat{h}}{x}$ for all $x\in A$ and, furthermore, with the notation of \cref{subsec:quantumsys},
the matrix product operators take the form
\begin{equation}
    \rho_{n}(M,b(x)) = \Phi_b(x_{(1)}) \otimes \Phi_b(x_{(2)}) \otimes\cdots\otimes \Phi_b(x_{(n)})
\end{equation}
for all $x\in A$ and $n\in\mathbb{N}$; we refer the reader to Ref.~\cite{molnar_matrix_2022} for a more detailed explanation. Let us also remark that $\hat{b}(\hat{h})\in\operatorname{Center}(\operatorname{Im}\Phi)$ since $\hat{h}$ is cocentral, and that the operators $\rho_n(M,b(x))$ are mixed states for all $n\in\mathbb{N}$ if $x\geq 0$ and $\eval{\hat{h}}{x} = 1$.

\begin{example}
\label{example:RFPTensorCG}
	Let $G$ be a finite group and consider the $*$-representations defined by
    \begin{equation*}
    \begin{array}{ll}
        \Phi \coloneqq \mathcal{L}:\mathbb{C}G \to \operatorname{End} \mathbb{C}G, 
        &
        \mathcal{L}(g)(h) \coloneqq gh,
        \\[5pt]
        \Psi:(\mathbb{C}G)^*\cong \mathbb{C}^G\to\operatorname{End}\mathbb{C}G,
        &
        \Psi(\delta_g)(h)  \coloneqq \eval{\delta_g}{h} g,
    \end{array}
    \end{equation*}
    for all $g,h\in G$.
    Then, the associated matrix product ope\-ra\-tor tensor is given by
	\[
		M
        =
        \frac{1}{|G|}\sum_{g\in G}
        \eval{\delta_g}{\placeholder} g \otimes \mathcal{L}(g),
	\]
    with the notation of \cref{eq:DefTensorRFP}, and the associated matrix product density operators are thus given, for all $x\in\mathbb{C}G$ and all $n\in\mathbb{N}$, by the expression
    \[
        \rho_{n}(M,b(x))
        =
        \frac{1}{|G|^n}
        \sum_{g\in G}
        \eval{\delta_g}{x} \mathcal{L}(g)\otimes\cdots \otimes\mathcal{L}(g).
    \]
\end{example}
In general, with the previous notation, there exists a fine-graining quantum channel $\Fine:\operatorname{End}\mathscr{H}\to\operatorname{End}\mathscr{H}^{\otimes 2}$, given for all $X\in\operatorname{End}\mathscr{H}$ by the expression
\begin{equation}
    \Fine(X) \coloneqq  D \operatorname{Tr}((\Phi\circ S)(h_{(3)})X) \Phi_b(h_{(1)}) \otimes  \Phi_b(h_{(2)}),
\end{equation}
and satisfies, for all $B\in \operatorname{End}\mathscr{K}$, $n\in\mathbb{N}$, and $k\in\{0,\ldots,n-1\}$, the property
\begin{equation}\label{eq:PropertyFine}
    (\mathrm{Id}^{\otimes k} \otimes \Fine \otimes \mathrm{Id}^{\otimes (n-k-1)})(\rho_n(M,B)) = \rho_{n+1}(M,B),
\end{equation}
i.e., it increases the system size of the matrix product density operators by acting at any site. In addition, there exists a coarse-graining quantum channel $\Coarse:\operatorname{End}\mathscr{H}^{\otimes 2}\to\operatorname{End}\mathscr{H}$, defined for all $X,Y\in\operatorname{End}\mathscr{H}$ by the expression
\begin{equation}
    \Coarse(X\otimes Y) \coloneqq D
        \operatorname{Tr}((\Phi\circ S)(h_{(3)}) X)
        \operatorname{Tr}((\Phi\circ S)(h_{(2)}) Y)
        \Phi_b(h_{(1)})
\end{equation}
and satisfies, for all $B\in \operatorname{End}\mathscr{K}$, $n\in\mathbb{N}$, and $k \in \{0,\ldots,n-1\}$, the property
\begin{equation}\label{eq:PropertyCoarse}
    (\mathrm{Id}^{\otimes k} \otimes \Coarse \otimes \mathrm{Id}^{\otimes (n-k-1)})(\rho_{n+1}(M,B)) = \rho_{n}(M,B).
\end{equation}
We refer to Ref.~\cite[Theorem~3.2]{ruiz_matrix_2024} for a proof. See \cref{fig:FC} for a diagrammatic representation.
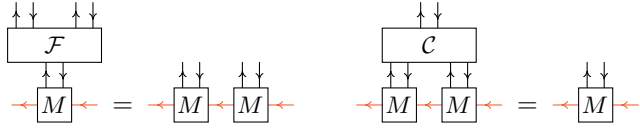
\begin{figure}
    \centering
\begin{tikzpicture}[scale=0.8]
    \draw[-mid] (-6.209, 0) -- (-6.209, 0.423);
    \draw[-mid] (-5.927, 0.423) -- (-5.927, 0);
    \draw[RedOrange, -mid] (-5.362, -0.282) -- (-5.786, -0.282);
    \draw[RedOrange, -mid] (-6.35, -0.282) -- (-6.773, -0.282);
    \draw[-mid] (-5.715, 0.988) -- (-5.715, 1.411);
    \draw[-mid] (-5.433, 1.411) -- (-5.433, 0.988);
    \draw[-mid] (-6.703, 0.988) -- (-6.703, 1.411);
    \draw[-mid] (-6.421, 1.411) -- (-6.421, 0.988);
    \filldraw[fill=white] (-6.844, 0.988) rectangle (-5.292, 0.423);
    \filldraw[fill=white] (-6.35, 0) rectangle (-5.786, -0.564);
    \node[anchor=center, font=\small] at (-6.068, -0.318) {$M^{}_{}$};
    \node[anchor=center, font=\small] at (-6.068, 0.706) {$\Fine$};
    \draw[-mid] (-2.963, 0) -- (-2.963, 0.423);
    \draw[-mid] (-2.681, 0.423) -- (-2.681, 0);
    \draw[-mid] (-3.951, 0) -- (-3.951, 0.423);
    \draw[-mid] (-3.669, 0.423) -- (-3.669, 0);
    \draw[RedOrange, -mid] (-3.104, -0.282) -- (-3.528, -0.282);
    \draw[RedOrange, -mid] (-4.092, -0.282) -- (-4.516, -0.282);
    \filldraw[fill=white] (-4.092, 0) rectangle (-3.528, -0.564);
    \node[anchor=center, font=\small] at (-3.81, -0.318) {$M^{}_{}$};
    \draw[RedOrange, -mid] (-2.117, -0.282) -- (-2.54, -0.282);
    \filldraw[fill=white] (-3.104, 0) rectangle (-2.54, -0.564);
    \node[anchor=center, font=\small] at (-2.822, -0.318) {$M^{}_{}$};
    \node[anchor=center] at (-4.939, -0.318) {$=$};
\end{tikzpicture}
\qquad
\begin{tikzpicture}[scale=0.8]
    \draw[RedOrange, -mid] (3.316, -0.282) -- (2.893, -0.282);
    \draw[RedOrange, -mid] (2.328, -0.282) -- (1.905, -0.282);
    \draw[RedOrange, -mid] (4.304, -0.282) -- (3.881, -0.282);
    \node[anchor=center] at (4.727, -0.318) {$=$};
    \draw[-mid] (5.715, 0) -- (5.715, 0.423);
    \draw[-mid] (5.997, 0.423) -- (5.997, 0);
    \draw[RedOrange, -mid] (6.562, -0.282) -- (6.138, -0.282);
    \draw[RedOrange, -mid] (5.574, -0.282) -- (5.151, -0.282);
    \filldraw[fill=white] (5.574, 0) rectangle (6.138, -0.564);
    \node[anchor=center, font=\small] at (5.856, -0.318) {$M^{}_{}$};
    \draw[-mid] (3.457, 0) -- (3.457, 0.423);
    \draw[-mid] (3.739, 0.423) -- (3.739, 0);
    \draw[-mid] (2.469, 0) -- (2.469, 0.423);
    \draw[-mid] (2.752, 0.423) -- (2.752, 0);
    \draw[-mid] (2.963, 0.988) -- (2.963, 1.411);
    \draw[-mid] (3.246, 1.411) -- (3.246, 0.988);
    \filldraw[fill=white] (2.328, 0) rectangle (2.893, -0.564);
    \node[anchor=center, font=\small] at (2.611, -0.318) {$M^{}_{}$};
    \filldraw[fill=white] (3.316, 0) rectangle (3.881, -0.564);
    \node[anchor=center, font=\small] at (3.598, -0.318) {$M^{}_{}$};
    \filldraw[fill=white] (2.328, 0.988) rectangle (3.881, 0.423);
    \node[anchor=center, font=\small] at (3.104, 0.706) {$\Coarse$};
\end{tikzpicture}
    \caption{
    Diagrammatic representation of the actions characterizing the fine- and coarse-graining quantum channels $\Fine$ and $\Coarse$ described in \cref{eq:PropertyFine,eq:PropertyCoarse}.
    Here, tensors are represented as boxes and their indices as legs. The central object is the tensor $M\in\operatorname{End}\mathscr{K}\otimes\operatorname{End}\mathscr{H}$, a four-index tensor drawn as a box with two horizontal legs, representing the first tensor factor, $\operatorname{End}\mathscr{K}\cong\mathscr{K}^*\otimes\mathscr{K}$, and two vertical legs, representing the second tensor factor, $\operatorname{End}\mathscr{H}\cong\mathscr{H}^*\otimes\mathscr{H}$, which carries the physical indices.
    Quantum channels (or, equivalently, their Choi representations) are depicted as maps acting on the corresponding physical legs.
    With these conventions, the figure shows the defining actions of $\Fine$ and $\Coarse$.}
    \label{fig:FC}
\end{figure}
The following result describes the action of both the fine- and coarse-graining quantum channels at the level of the algebra.

\begin{proposition}\label{prop:sigma_is_S}
	The comultiplication $\Delta :A\to A\otimes A$ implements the fine-graining quantum channel, i.e.,
	\begin{equation}
		\label{eq:PropertyT}
        \Fine \circ \Phi_b = (\Phi_b\otimes \Phi_b) \circ \Delta,
	\end{equation}
    and the linear map $\sigma:A\otimes A\to A$ implements the coarse-graining quantum channel, i.e.,
	\begin{equation}
		\label{eq:PropertyS}
        \Coarse \circ (\Phi_b\otimes \Phi_b) = \Phi_b \circ \sigma.
	\end{equation}
\end{proposition}

See \cref{appendix:proofs2} for a proof.

\section{Renormalization flows and their algebraic formulation}
\label{sec:flows}

In this section we address how renormalization flows can be defined close to a boundary state renormalization fixed point and how they can be described at an intrinsic algebraic level.
As introduced before, we focus on perturbations generated by uniform local noise, modeled by completely positive and trace-preserving quantum channels acting identically on each physical site and combined with the coarse-graining inherent to the fixed point.

\subsection{Renormalization trajectories} \label{sec:trajectories}

Having defined the properties characterizing renormalization fixed points and the corresponding coarse-graining and fine-graining quantum channels, we now consider the resulting renormalization trajectories in the vicinity of these fixed points.
To this end, we consider an arbitrary on-site perturbation obtained by applying the same completely positive trace-preserving linear map to every site of the mixed state.
\smallbreak\noindent%
Let us denote the on-site quantum channel by $\Noisy:\operatorname{End}\mathscr{H}\to\operatorname{End}\mathscr{H}$ and recall that 
the tensor generating the matrix product density operators is given by
\[
    M
    \coloneqq
    \sum_{i=1}^{\dim A} \Psi(e^{i}) \otimes \Phi_b(e_i).
\]
By applying the noise quantum channel $\Noisy$, we obtain matrix product density operators which are generated by the tensor
\begin{equation}
\label{eq:defM1}
    M_1
    \coloneqq
    \sum_{i=1}^{\dim A} \Psi(e^{i}) \otimes \Noisy(\Phi_b(e_i)).
\end{equation}
If one now coarse-grains the matrix product density operators obtained above, the generating tensor becomes
\begin{equation}
    \label{eq:defM2}
    M_{2}
    \coloneqq
    \sum_{i,j=1}^{\dim A}
        \Psi(e^{i}e^{j}) \otimes \Coarse(\Noisy(\Phi_b(e_i))\otimes \Noisy(\Phi_b(e_j))),
\end{equation}
or, more interestingly, if we employ the fine-graining quantum channel, one can rewrite
\begin{equation}
    \label{eq:charM2}
    M_{2}
    =
    \sum_{i=1}^{\dim A}
        \Psi(e^i) \otimes (\Coarse\circ \Noisy^{\otimes 2}\circ \Fine)(\Phi_b(e_i)).
\end{equation}
In general, for $n\geq 2$ we define
\begin{align}
    \label{eq:FineNsteps}
    \Fine^{\smash{(n)}} &\coloneqq (\Fine\otimes \mathrm{Id}^{\otimes (n-2)})\circ \cdots\circ  (\Fine\otimes \mathrm{Id}) \circ \Fine:\operatorname{End}\mathscr{H}\to \operatorname{End}\mathscr{H}^{\otimes n},
\\
    \label{eq:CoarseNsteps}
    \Coarse^{\smash{(n)}} &\coloneqq \Coarse \circ (\Coarse\otimes\mathrm{Id}) \circ \cdots \circ (\Coarse\otimes \mathrm{Id}^{\otimes (n-2)}):\operatorname{End}\mathscr{H}^{\otimes n}\to \operatorname{End}\mathscr{H}.
\end{align}
Note that by coassociativity of the comultiplication and associativity of $\sigma$, the order of the tensor factors in \cref{eq:FineNsteps,eq:CoarseNsteps} is not relevant.
Then, the tensor after $n$ steps is given by the expression
\begin{equation}
\label{eq:defMn}
    M_{n} \coloneqq
    \sum_{i=1}^{\dim A}
        \Psi(e^i) \otimes (\Coarse^{(n)}\circ \Noisy^{\otimes n}\circ \Fine^{(n)})(\Phi_b(e_i)).
\end{equation}
Equivalently, $M_n$ is characterized by the property
\begin{equation}
\label{eq:charMn}
    (\Coarse^{(n)}\circ \Noisy^{\otimes n} \circ \Fine^{(n)})(\rho_1(M,\placeholder))
    =
    \rho_{1}(M_{n},\placeholder).
\end{equation}
We will focus on trajectories where the limit
\begin{equation}
\label{eq:defMinfty}
M_{\smash{\infty}}
    \coloneqq
        \lim_{n\to \infty} M_{\smash{n}}
\end{equation}
exists in $\operatorname{End}\mathscr{K}\otimes\operatorname{End}\mathscr{H}$.
In this case, $M_{\smash{\infty}}$ is again a renormalization fixed point in the sense of \cref{subsec:rfps}.
See \cref{fig:tensornetworks} for diagrammatic representations of the previous constructions, and \cref{fig:WilsonMn} for a sketch of the renormalization flow.
Rather than focusing directly on the tensors $M_n$, we will often consider the induced flow at the level of quantum channels:
\begin{equation}
    \Noisy_n \coloneqq \Coarse^{(n)} \circ \Noisy^{\otimes n} \circ \Fine^{(n)} :
    \operatorname{End}\mathscr{H}\to\operatorname{End}\mathscr{H}.
\end{equation}

\begin{figure}[!ht]
    \centering
\begin{tikzpicture}[scale=0.85]
    \draw[RedOrange, -mid] (-6.209, 0.071) -- (-6.632, 0.071);
    \draw[-mid] (-6.068, 0.353) -- (-6.068, 0.776);
    \draw[-mid] (-5.786, 0.776) -- (-5.786, 0.353);
    \draw[RedOrange, -mid] (-5.221, 0.071) -- (-5.644, 0.071);
    \filldraw[fill=white] (-6.209, 0.353) rectangle (-5.644, -0.212);
    \node[anchor=center, font=\footnotesize, scale=0.9] at (-5.927, 0.071) {$M$};
    \draw[-mid] (-6.068, 1.341) -- (-6.068, 1.764);
    \draw[-mid] (-5.786, 1.764) -- (-5.786, 1.341);
    \filldraw[fill=white] (-6.209, 1.341) rectangle (-5.644, 0.776);
    \node[anchor=center, font=\footnotesize, scale=0.9] at (-5.927, 1.058) {$\Noisy$};
    \draw[RedOrange, -mid] (-3.951, 0.071) -- (-4.374, 0.071);
    \draw[-mid] (-3.81, 0.353) -- (-3.81, 0.776);
    \draw[-mid] (-3.528, 0.776) -- (-3.528, 0.353);
    \draw[RedOrange, -mid] (-2.963, 0.071) -- (-3.387, 0.071);
    \filldraw[fill=white] (-3.951, 0.353) rectangle (-3.387, -0.212);
    \node[anchor=center, font=\footnotesize, scale=0.9] at (-3.669, 0.035) {$M_1$};
    \draw[-mid] (-2.822, 0.353) -- (-2.822, 0.776);
    \draw[-mid] (-2.54, 0.776) -- (-2.54, 0.353);
    \draw[RedOrange, -mid] (-1.976, 0.071) -- (-2.399, 0.071);
    \filldraw[fill=white] (-2.963, 0.353) rectangle (-2.399, -0.212);
    \node[anchor=center, font=\footnotesize, scale=0.9] at (-2.681, 0.035) {$M_1$};
    \filldraw[fill=white] (-3.951, 1.341) rectangle (-2.399, 0.776);
    \draw[-mid] (-3.316, 1.341) -- (-3.316, 1.764);
    \draw[-mid] (-3.034, 1.764) -- (-3.034, 1.341);
    \node[anchor=center, font=\footnotesize, scale=0.9] at (-3.175, 1.058) {$\Coarse$};
    \draw[-mid] (4.727, -0.212) -- (4.727, -0.635);
    \draw[-mid] (5.433, -0.635) -- (5.433, -0.212);
    \draw[-mid] (5.715, -0.212) -- (5.715, -0.635);
    \filldraw[fill=white] (4.304, 0.353) rectangle (5.856, -0.212);
    \draw[-mid] (4.939, 0.353) -- (4.939, 0.776);
    \draw[-mid] (5.221, 0.776) -- (5.221, 0.353);
    \node[anchor=center, font=\footnotesize, scale=0.9] at (5.08, 0.071) {$\Coarse$};
    \filldraw[fill=white] (4.798, 1.341) rectangle (6.491, 0.776);
    \draw[-mid] (5.503, 1.341) -- (5.503, 1.764);
    \draw[-mid] (5.786, 1.764) -- (5.786, 1.341);
    \node[anchor=center, font=\footnotesize, scale=0.9] at (5.644, 1.058) {$\Coarse$};
    \draw[-mid] (6.068, 0.353) -- (6.068, 0.776);
    \draw[-mid] (6.35, 0.776) -- (6.35, 0.353);
    \node[anchor=center, font=\footnotesize, scale=0.9] at (6.209, 0.071) {$\ddots$};
    \node[anchor=center, font=\footnotesize, scale=0.9] at (6.526, 0.071) {$\ddots$};
    \filldraw[fill=white] (-1.129, 0.353) rectangle (0.423, -0.212);
    \node[anchor=center, font=\footnotesize, scale=0.9] at (-0.353, 0.071) {$\Fine$};
    \draw[-mid] (-0.494, -0.635) -- (-0.494, -0.212);
    \draw[-mid] (-0.212, -0.212) -- (-0.212, -0.635);
    \filldraw[fill=white] (-1.129, 2.328) rectangle (0.423, 1.764);
    \draw[-mid] (-0.494, 2.328) -- (-0.494, 2.752);
    \draw[-mid] (-0.212, 2.752) -- (-0.212, 2.328);
    \node[anchor=center, font=\footnotesize, scale=0.9] at (-0.353, 2.046) {$\Coarse$};
    \filldraw[fill=white] (-1.129, 1.341) rectangle (-0.564, 0.776);
    \node[anchor=center, font=\footnotesize, scale=0.9] at (-0.847, 1.058) {$\Noisy$};
    \filldraw[fill=white] (-0.141, 1.341) rectangle (0.423, 0.776);
    \node[anchor=center, font=\footnotesize, scale=0.9] at (0.141, 1.058) {$\Noisy$};
    \draw[-mid] (-0.988, 1.341) -- (-0.988, 1.764);
    \draw[-mid] (-0.706, 1.764) -- (-0.706, 1.341);
    \draw[-mid] (0, 1.341) -- (0, 1.764);
    \draw[-mid] (0.282, 1.764) -- (0.282, 1.341);
    \draw[-mid] (-0.988, 0.353) -- (-0.988, 0.776);
    \draw[-mid] (-0.706, 0.776) -- (-0.706, 0.353);
    \draw[-mid] (0, 0.353) -- (0, 0.776);
    \draw[-mid] (0.282, 0.776) -- (0.282, 0.353);
    \draw[shift={(1.411, 1.764)}, yscale=-1, mid-] (0, 0) -- (0, 0.423);
    \draw[shift={(1.693, 1.341)}, yscale=-1, mid-] (0, 0) -- (0, -0.423);
    \draw[shift={(2.399, 1.764)}, yscale=-1, mid-] (0, 0) -- (0, 0.423);
    \draw[shift={(2.681, 1.341)}, yscale=-1, mid-] (0, 0) -- (0, -0.423);
    \filldraw[shift={(1.27, 0.776)}, yscale=-1, fill=white] (0, 0) rectangle (1.552, -0.564);
    \draw[shift={(1.905, 0.776)}, yscale=-1, mid-] (0, 0) -- (0, 0.423);
    \draw[shift={(2.187, 0.353)}, yscale=-1, mid-] (0, 0) -- (0, -0.423);
    \node[anchor=center, font=\footnotesize, scale=0.9] at (2.046, 1.058) {$\Fine$};
    \filldraw[shift={(1.764, -0.212)}, yscale=-1, fill=white] (0, 0) rectangle (1.693, -0.564);
    \draw[shift={(2.469, -0.212)}, yscale=-1, mid-] (0, 0) -- (0, 0.423);
    \draw[shift={(2.752, -0.635)}, yscale=-1, mid-] (0, 0) -- (0, -0.423);
    \node[anchor=center, font=\footnotesize, scale=0.9] at (2.611, 0.071) {$\Fine$};
    \draw[shift={(3.034, 0.776)}, yscale=-1, mid-] (0, 0) -- (0, 0.423);
    \draw[shift={(3.316, 0.353)}, yscale=-1, mid-] (0, 0) -- (0, -0.423);
    \node[yscale=-1, anchor=center, font=\footnotesize, scale=0.9] at (3.175, 1.058) {$\ddots$};
    \node[yscale=-1, anchor=center, font=\footnotesize, scale=0.9] at (3.493, 1.058) {$\ddots$};
    \draw[-mid] (4.516, -0.635) -- (4.516, -0.212);
    \node[anchor=center] at (-5.927, -1.129) {(a)\vphantom{b}};
    \node[anchor=center] at (-3.175, -1.129) {(b)};
    \node[anchor=center] at (-0.353, -1.129) {(c)\vphantom{b}};
    \node[anchor=center] at (2.364, -1.129) {(d)\vphantom{b}};
    \node[anchor=center] at (5.398, -1.129) {(e)\vphantom{b}};
\end{tikzpicture}
\caption{
Diagrammatic representation, in the tensor network picture, of \cref{eq:defM1,eq:defM2,eq:charM2,eq:FineNsteps,eq:CoarseNsteps}, following the conventions of \cref{fig:FC}. The quantum channel $\Noisy$ (or, equivalently, its Choi representation) is depicted as an operation with two inputs and two outputs. In (a) we depict $M_1$, obtained by applying $\Noisy$ to the physical indices of $M$. In (b) we depict the tensor $M_2$, obtained by coarse-graining two tensors $M_1$. In (c) we represent the action on the physical indices from which one obtains $M_2$ from $M$. In (d) and (e) we depict the composition of two or more fine-graining and coarse-graining quantum channels, respectively, corresponding to \cref{eq:FineNsteps,eq:CoarseNsteps}.
}
    \label{fig:tensornetworks}
\end{figure}
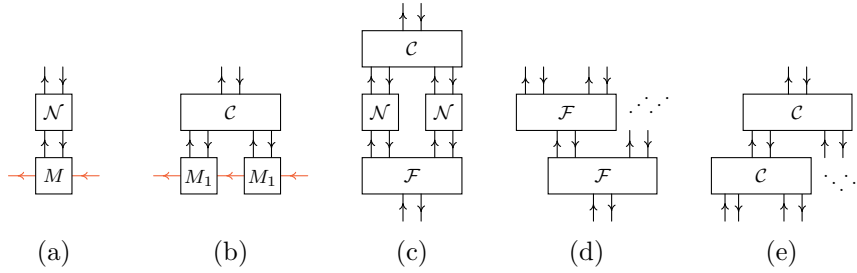

\begin{figure}[!ht]
    \centering
\begin{tikzpicture}[scale=1]
    \node[disk, No] at (-4.764, 1.232) {};
    \node[disk, No] at (-2.844, 0.138) {};
    \node[disk, No] at (-0.634, 0.152) {};
    \draw[densely dashed, ->] (-4.581, 1.286) .. controls (-3.157, 1.521) and (-1.886, 1.188) .. (-0.77, 0.287);
    \node[disk, No] at (1.273, 1.268) {};
    \draw[->] (-4.579, 1.183) .. controls (-3.52, 1.146) and (-4.208, -0.083) .. (-3.032, 0.1);
    \draw[->] (-2.665, 0.206) .. controls (-1.724, 0.686) and (-1.705, -0.722) .. (-0.771, 0.018);
    \node[anchor=center] at (-5.147, 1.241) {$M$};
    \node[anchor=center] at (1.801, 1.254) {$M_\infty$};
    \node[anchor=center] at (-0.331, -0.164) {$M_2$};
    \node[anchor=center] at (-2.875, -0.243) {$M_1$};
    \node[rotate=27.847, anchor=center] at (0.315, 0.871) {${\cdots}$};
    \draw (-0.509, 0.296) arc[start angle=-132.637, end angle=-122.661, x radius=4.077, y radius=-4.077];
    \draw[->] (0.574, 1.012) arc[start angle=-114.318, end angle=-106.536, x radius=4.077, y radius=-4.077];
    \node[anchor=center] at (-2.304, 1.674) {$\Coarse \circ \Noisy^{\otimes 2}\circ \Fine$};
    \node[anchor=center] at (-4.153, 0.269) {$\Noisy$};
    \node[anchor=center] at (-1.734, -0.44) {$\Coarse$};
\end{tikzpicture}
    \caption{Diagrammatic representation of the renormalization flow generated by the tensors in \cref{eq:defM1,eq:defM2,eq:defMinfty}.}
    \label{fig:WilsonMn}
\end{figure}
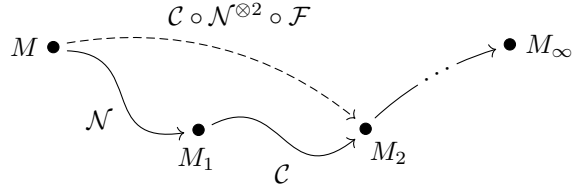

\begin{remark}
One can assume without loss of generality that $\Noisy:\operatorname{Im}\Phi \to \operatorname{Im}\Phi$. This is due to the fact that (i) the noise quantum channel acts on an already prepared global renormalization fixed-point matrix product density operator, and (ii) the action of the noise quantum channel is followed by the coarse-graining procedure, which destroys any non-diagonal block in the $C^*$-Hopf algebra case. Remarkably, (ii) may no longer be true in more general scenarios, such as $C^*$-weak Hopf algebras, where the coarse-graining channel exhibits more exotic behavior.
\end{remark}

\subsection{From noise quantum channels to algebra maps}

We will now show that the trajectories of quantum channels $\Noisy_n$ introduced in the previous subsection can be understood as trajectories of linear maps $F_{n}$ on the $C^*$-Hopf algebra $A$ satisfying the appropriate compatibility conditions, or equivalently as sequences of powers of certain positive linear functionals $\phi^{n}$ on the enveloping $C^*$-Hopf algebra $A \otimes A^{\mathrm{op}}$. 
In particular, these sequences can be characterized algebraically in terms of convolutions and, in subsequent sections, the fixed points of such trajectories will be more elegantly described using idempotency conditions.

The first step is to establish that one can, without loss of generality, consider the noise quantum channels equivalently as linear maps on the underlying symmetry algebra.

\begin{proposition}
\label{prop:NF}
	Let $A$ be a $C^*$-Hopf algebra, let $\Phi : A \to \operatorname{End}\mathscr{H}$ be a faithful $*$-re\-pre\-sen\-ta\-tion of $A$ and consider any completely positive linear map $F:A\to A$ such that $\hat{h}\circ F = \hat{h}$.
    Then, the linear map $\Noisy_F:\operatorname{End}\mathscr{H}\to\operatorname{End}\mathscr{H}$ defined by
	\begin{equation}\label{eq:DefinitionNF}
		\Noisy_F(X)
            \coloneqq
            D \operatorname{Tr}((\Phi\circ S)(h_{(1)})X) (\Phi_b\circ F)(h_{(2)})
	\end{equation}
    for all $X\in\operatorname{End}\mathscr{H}$
	is a completely positive and trace-preserving map satisfying
	\begin{equation}\label{eq:PropertyNF}
		\Noisy_F\circ \Phi_b
        =
        \Phi_b\circ F.
	\end{equation}
	Conversely, for any completely positive and trace-preserving linear map of the form $\Noisy:\operatorname{Im}\Phi\to\operatorname{Im}\Phi$, the linear map $F_\Noisy:A\to A$ defined by
	\begin{equation}\label{eq:NtoF}
		F_\Noisy(x) \coloneqq D \operatorname{Tr}((\Phi\circ S)(h_{(1)}) \Noisy(\Phi_b(x)))  h_{(2)}
	\end{equation}
	for all $x\in A$ is completely positive and satisfies $\hat{h}\circ F_{\Noisy} = \hat{h}$. Moreover, the maps $\Noisy\mapsto F_{\Noisy}$ and $F\mapsto \Noisy_F$ establish a one-to-one correspondence.
\end{proposition}

\begin{proof}
	Let $F:A\to A$ be as above. We first prove that $\Noisy_F$ is completely positive by building a Stinespring representation. Let
	\begin{equation}
		Q
        \coloneqq
        (\mathrm{Id}\otimes F)\circ \Delta(h)
        =
        h_{(1)}\otimes F(h_{(2)})
        \in A\otimes A,
	\end{equation}
	which is positive since $\Delta$ is a $*$-algebra homomorphism and $F$ is assumed to be completely positive. Hence, $Q = q q^*$ for some $q \coloneqq \sum_{i=1}^r q_{i,L}\otimes q_{i,R}\in A \otimes A$, and for all $X\in\operatorname{End}\mathscr{H}$,
	\begin{equation}
		\Noisy_F(X)
		=
        \sum_{i,j=1}^r D
		\operatorname{Tr}((\Phi\circ S)(q_{\smash{i,L}}^{} q_{\smash{j,L}}^*)X)
            \hat{b}(\hat{h}) 
        \Phi(q_{\smash{i,R}}^{} q_{\smash{j,R}}^*).
	\end{equation}
	Using that $\hat{b}(\hat{h})$ is central and positive, $\Phi$ is a $*$-al\-ge\-bra ho\-mo\-mor\-phism, $S$ is a $*$-algebra an\-ti\-ho\-mo\-morphism, and the trace is cyclic, one can rewrite the previous expression as
	\begin{equation}
        \Noisy_F(X)
        =
        \sum_{i,j=1}^r
            D
        \operatorname{Tr}((\Phi\circ S)(q_{i,L})X(\Phi\circ S)(q_{j,L})^\dagger )
            \hat{b}(\hat{h})^{1/2}
            \Phi(q_{i,R})
            \Phi(q_{j,R})^\dagger
            \hat{b}(\hat{h})^{1/2}.
	\end{equation}
	By defining $V \coloneqq \sqrt{D} \sum_{i=1}^r   (\Phi\circ S)(q_{i,L}) \otimes\hat{b}(\hat{h})^{1/2} \Phi(q_{i,R}) \in\operatorname{End}\mathscr{H}^{\otimes 2}$,
	\begin{equation}
    \Noisy_F(X)
    =
    (\mathrm{Tr} \otimes \mathrm{Id})(V (X\otimes \mathrm{Id})V^\dagger).
    \end{equation}
	Moreover, $\Noisy_F$ is trace-preserving since, for all $X\in\operatorname{End}\mathscr{H}$,
	\begin{align}
		\operatorname{Tr}\Noisy_F(X)
        & =
        D \operatorname{Tr}((\Phi\circ S)(h_{(1)})X)
        \operatorname{Tr}(\hat{b}(\hat{h})(\Phi\circ F)(h_{(2)}))
		\\ &
        =
        D \operatorname{Tr}((\Phi\circ S)(h_{(1)})X)
        \eval{\hat{h}}{F(h_{(2)})}
		\\ &
        =
        D
        \operatorname{Tr}((\Phi\circ S)(h_{(1)})X)
        \eval{\hat{h}}{h_{(2)}}
		\\ &
        =
        \operatorname{Tr}((\Phi\circ S)(1)X)
        =
        \operatorname{Tr}X,
        \vphantom{\hat{h}} 
	\end{align}
	where we employed \cref{eq:DefBdry-Dual} in the second equality, we used the hypothesis $\hat{h}\circ F = \hat{h}$ in the third, and we recalled \cref{eq:DualInt} in the fourth.
    Finally, it implements the action of $F$ since, for all $x\in A$,
	\begin{align}
		\Noisy_F(\hat{b}(\hat{h})\Phi(x))
		  & =
        D \hat{b}(\hat{h})
        \operatorname{Tr}((\Phi\circ S)(h_{(1)})
        \hat{b}(\hat{h})\Phi(x))
        (\Phi\circ F)(h_{(2)})
		\\ &
        =
        D \hat{b}(\hat{h})
        \operatorname{Tr}(\hat{b}(\hat{h})\Phi(S(h_{(1)})x))
        (\Phi\circ F)(h_{(2)})
		\\ &
        =
        D \hat{b}(\hat{h}) 
        \eval{\hat{h}}{S(h_{(1)}) x} 
        (\Phi\circ F)(h_{(2)})
		\\ &
        =
        D \hat{b}(\hat{h})
        \eval{\hat{h}}{S(h_{(1)})}
        (\Phi\circ F)(x h_{(2)})
		\\ &
        =
        \hat{b}(\hat{h}) (\Phi\circ F)(x)
        =
        \Phi_b(F(x)),
	\end{align}
	where we have used that $\hat{b}(\hat{h})$ is central and $\Phi$ is a $*$-re\-pre\-sen\-ta\-tion in the second equality, \cref{eq:DefBdry-Dual} in the third, \cref{eq:PTmodified} in the fourth, and \cref{eq:DualInt} in the last.
	\smallbreak\noindent%
	Conversely, let $\Noisy:\operatorname{Im}\Phi\to\operatorname{Im}\Phi$ be a quantum channel and define $F_{\smash{\Noisy}}:A\to A$ as
	\begin{equation}\label{eq:defFfromN}
		F_{\Noisy}
        \coloneqq
        D h_{(1)} \operatorname{Tr}(
            (\Phi\circ S)(h_{(2)})
            \Noisy(\hat{b}(\hat{h})\Phi(\placeholder))
        ).
	\end{equation}
	Let us first check that $F_{\Noisy}$ is completely positive and satisfies $\hat{h}\circ F_{\Noisy} = \hat{h}$.
	On the one hand, $h$ is a positive idempotent because $h = hh^*$, and hence
	\begin{align}
		F_{\Noisy}
        & =
        D h_{(1)} h_{(1')}^*
        \operatorname{Tr}((\Phi\circ S)(h_{(2)}h_{(2')}^*) (\Noisy\circ \Phi_b)(\placeholder))
		\\ &
        =
        D h_{(1)} h_{(1')}^*
        \operatorname{Tr}((\Phi\circ S)(h_{(2')})^\dagger(\Phi\circ S)(h_{(2)}) (\Noisy\circ \Phi_b)(\placeholder))
		\\ &
        =
        D h_{(1)} h_{(1')}^*
        \operatorname{Tr}((\Phi\circ S)(h_{(2)}) (\Noisy\circ \Phi_b)(\placeholder)(\Phi\circ S)(h_{(2')})^\dagger)
		\\ &
        =
        (\mathrm{Id}_A\otimes\mathrm{Tr})(W  (1_A \otimes (\Noisy\circ \Phi_b)(\placeholder))  W^\dagger),
	\end{align}
	where $W \coloneqq \sqrt{D} h_{(1)}\otimes (\Phi\circ S)(h_{(2)}) \in A \otimes \operatorname{End}\mathscr{H}$; note that the map $F_{\Noisy}$ is comple\-te\-ly po\-si\-tive since $\Noisy\circ \Phi_b$ is completely positive. On the other hand, for all $x\in A$, it holds that
	\begin{align}
		\eval{\hat{h}\circ F_{\Noisy}}{x}
        & =
        D
        \operatorname{Tr}((\Phi\circ S)(h_{(1)}) \Noisy(\hat{b}(\hat{h}) \Phi(x)))
        \eval{\hat{h}}{h_{(2)}}
		\\ &
        =
        D
        \operatorname{Tr}((\Phi\circ S)(\eval{\hat{h}}{h_{(2)}} h_{(1)})
        \Noisy(\hat{b}(\hat{h}) \Phi(x)))
		\\ &
        =
        \operatorname{Tr}((\Phi\circ S)(1) \Noisy(\hat{b}(\hat{h}) \Phi(x))) 
        \\ &
        =
        \operatorname{Tr}(\Noisy(\hat{b}(\hat{h}) \Phi(x)))
        =
        \operatorname{Tr}(\hat{b}(\hat{h}) \Phi(x))
        =
        \eval{\hat{h}}{x}.
	\end{align}
	Finally, we check that $F_{\Noisy}$ implements the action of $\Noisy$ on the algebra $A$ and, as an immediate consequence, $\Noisy_{F_{\Noisy}} = \Noisy$. For that purpose, note that as a consequence of the Artin--Wedderburn theorem, one can rewrite
	\begin{equation}
		\hat{b}(\hat{h})
        =
        \Phi(c_{\smash{\hat{h}}}),
        \quad\text{ for }\quad
        c_{\smash{\hat{h}}}
        \coloneqq
        \frac{1}{D} \sum_{a=1}^r E_a
        \in \operatorname{Center}A,
	\end{equation}
	where $E_1,\ldots,E_r\in A$ are the minimal central idempotents of $A$.
    Moreover, for every $x\in A$ there exists $\eta_x\in A$ such that
	\begin{equation}
		\Noisy(\Phi_b(x))
        =
        \Noisy(\hat{b}(\hat{h})\Phi(x))
        =
        \Noisy(\Phi(c_{\smash{\hat{h}}}x))
        =
        \Phi(c_{\smash{\hat{h}}} \eta_x)
        =
        \Phi_b(\eta_x),
	\end{equation}
   from the facts that $\Noisy(\operatorname{Im}\Phi)\subseteq \operatorname{Im}\Phi$ and that $\Phi$ is faithful.
    Then, for all $x\in A$,
	\begin{align}
		\hat{b}(\hat{h}) (\Phi\circ F_{\Noisy})(x)
        & =
        D
        \operatorname{Tr}(  (\Phi\circ S)(h_{(1)})
                            \Noisy(\hat{b}(\hat{h})\Phi(x))  )
            \hat{b}(\hat{h}) \Phi(h_{(2)})
		\\ &
        =
        D
        \operatorname{Tr}(  (\Phi\circ S)(h_{(1)})
                            \hat{b}(\hat{h}) \Phi(\eta_x)  )
            \Phi_b(h_{(2)})
        \\ &
        =
        D \eval{\hat{h}}{\eta_x S(h_{(1)})}  \Phi_b(h_{(2)})
        =
        \hat{b}(\hat{h})\Phi(\eta_x)
        =
        (\Noisy \circ \Phi_b)(x),
	\end{align}
	as we wanted to prove. Conversely, for all $x\in A$,
    \begin{align}
        F_{\Noisy_F}(x)
        &=
        D \operatorname{Tr}( (\Phi\circ S)(h_{(1)}) \Noisy_F(\Phi_b(x)) ) h_{(2)}
        \\ &=
        D \operatorname{Tr}( (\Phi\circ S)(h_{(1)}) \Phi_b(F(x)) ) h_{(2)}
        \\ &=
        D \eval{\hat{h}}{S(h_{(1)}) F(x)} h_{(2)}
        =
        D \eval{\hat{h}}{S(h_{(1)})} F(x) h_{(2)}
        =
        F(x),
    \end{align} 
    where the fourth equality holds by \cref{eq:PTmodified}, and the fifth by \cref{eq:DualHaar,eq:DualInt}.
\end{proof}

For all maps $F:A\to A$ under the assumptions above, let
\begin{equation}
M_F
	\coloneqq
    \sum_{i=1}^{\dim A}
        \Psi(e^i)
        \otimes
        \hat{b}(\hat{h})(\Phi\circ F)(e_i).
\end{equation}
The following result establishes the relation between the action of two on-site noise quantum channels after renormalization at the level of the algebra.

\begin{proposition}\label{prop:FGconv}
	Let $F,G:A\to A$ be two linear maps satisfying the conditions of \cref{prop:NF}. Then, the linear map $F\star G :A\to A$, defined by the expression
	\begin{equation} \label{eq:convMaps}
		F\star G \coloneqq \sigma \circ (F\otimes G) \circ \Delta,
	\end{equation}
    satisfies the aforementioned conditions, and it holds that
	\begin{equation}\label{eq:RenNoise}
		\Coarse \circ ({\Noisy}_{F}\otimes \Noisy_G) \circ \Fine = \Noisy_{F\star G}.
	\end{equation}
\end{proposition}
\noindent%
Let $H \coloneqq F\star G$ for simplicity. 
In the tensor network picture, this corresponds to the following sequence of mappings
\begin{equation}
\begin{tikzpicture}[scale=1]
    \draw[-mid] (-3.951, 1.129) -- (-3.951, 1.552);
    \draw[RedOrange, -mid] (-3.246, 0.847) -- (-3.669, 0.847);
    \draw[-mid] (-3.951, 0.141) -- (-3.951, 0.564);
    \draw[RedOrange, -mid] (-4.233, 0.847) -- (-4.657, 0.847);
    \filldraw[fill=white] (-4.233, 1.129) rectangle (-3.669, 0.564);
    \node[anchor=center, font=\footnotesize] at (-3.951, 0.811) {$M^{}_{}$};
    \draw[-mid] (-1.976, 1.129) -- (-1.976, 1.552);
    \draw[RedOrange, -mid] (-1.27, 0.847) -- (-1.693, 0.847);
    \draw[-mid] (-1.976, 0.141) -- (-1.976, 0.564);
    \draw[RedOrange, -mid] (-2.258, 0.847) -- (-2.681, 0.847);
    \filldraw[fill=white] (-2.258, 1.129) rectangle (-1.693, 0.564);
    \node[anchor=center, font=\footnotesize] at (-1.976, 0.811) {$M^{}_{}$};
    \draw[-mid] (-0.988, 1.129) -- (-0.988, 1.552);
    \draw[RedOrange, -mid] (-0.282, 0.847) -- (-0.706, 0.847);
    \draw[-mid] (-0.988, 0.141) -- (-0.988, 0.564);
    \filldraw[fill=white] (-1.27, 1.129) rectangle (-0.706, 0.564);
    \node[anchor=center, font=\footnotesize] at (-0.988, 0.811) {$M^{}_{}$};
    \draw[|->] (-3.246, 1.129) arc[start angle=-113.962, end angle=-66.038, x radius=0.695, y radius=-0.695] node[midway, above] {$\Fine$};
    \draw[|->] (-0.282, 1.129) arc[start angle=-113.962, end angle=-66.038, x radius=0.695, y radius=-0.695]  node[midway, above] {$\Noisy_F\otimes\Noisy_G$};
    \draw[-mid] (0.988, 1.129) -- (0.988, 1.552);
    \draw[RedOrange, -mid] (1.693, 0.847) -- (1.27, 0.847);
    \draw[-mid] (0.988, 0.141) -- (0.988, 0.564);
    \draw[RedOrange, -mid] (0.706, 0.847) -- (0.282, 0.847);
    \filldraw[fill=white] (0.706, 1.129) rectangle (1.27, 0.564);
    \node[anchor=center, font=\footnotesize] at (0.988, 0.811) {$M^{}_{F}$};
    \draw[-mid] (1.976, 1.129) -- (1.976, 1.552);
    \draw[RedOrange, -mid] (2.681, 0.847) -- (2.258, 0.847);
    \draw[-mid] (1.976, 0.141) -- (1.976, 0.564);
    \filldraw[fill=white] (1.693, 1.129) rectangle (2.258, 0.564);
    \node[anchor=center, font=\footnotesize] at (1.976, 0.811) {$M^{}_{G}$};
    \draw[|->] (2.681, 1.129) arc[start angle=-113.962, end angle=-66.038, x radius=0.695, y radius=-0.695] node[midway, above] {$\Coarse$};
    \draw[-mid] (3.951, 1.129) -- (3.951, 1.552);
    \draw[RedOrange, -mid] (4.657, 0.847) -- (4.233, 0.847);
    \draw[-mid] (3.951, 0.141) -- (3.951, 0.564);
    \draw[RedOrange, -mid] (3.669, 0.847) -- (3.246, 0.847);
    \filldraw[fill=white] (3.669, 1.129) rectangle (4.233, 0.564);
    \node[anchor=center, font=\footnotesize] at (3.951, 0.811) {$M^{}_{H}$};
\end{tikzpicture}
\end{equation}
where the quantum channels act at the physical level.

\begin{proof}
	As both terms in \cref{eq:RenNoise} are endomorphisms of $\operatorname{Im}\Phi$, it suffices to prove that the identity
	\[
		(\Coarse
            \circ ({\Noisy}_{F}\otimes \Noisy_G)
            \circ \Fine)(\Phi_b(x))
		=
        \Noisy_H(\Phi_b(x))
	\]
	holds for all $x\in A$.
	Indeed, the left-hand side takes the form
	\begin{align*}
		(\Coarse
            \circ ({\Noisy}_{F}\otimes \Noisy_G)
            \circ \Fine)
                (\hat{b}(\hat{h})\Phi(x))
		  & =
		(\Coarse \circ ({\Noisy}_{F}\otimes \Noisy_G))
            (\hat{b}(\hat{h})^{\otimes 2}(\Phi\otimes \Phi)(\Delta(x)))
		\\ &
        =
		\Coarse(\hat{b}(\hat{h})^{\otimes 2}((\Phi\circ F)\otimes (\Phi\circ G))(\Delta(x)))
		\\ &
        =
		\hat{b}(\hat{h}) \Phi( \sigma \circ (F \otimes G)\circ \Delta(x) ),
	\end{align*}
	for all $x\in A$. This follows directly from \cref{eq:PropertyT,eq:PropertyNF,eq:PropertyS}, and coincides with the expression on the right-hand side, as we wanted to prove. Observe that $F\star G$ is completely positive because $\Delta$ is a $*$-algebra homomorphism (and therefore completely positive) and $\sigma$ was shown to be completely positive in \cref{lemma:sigma}.
	Finally, $\hat{h}\circ (F\star G) = \hat{h}$ as a consequence of the following calculation, valid for all $x\in A$:
	\begin{align*}
		\eval{\hat{h}}{H(x)}
		  & =
		\eval{\hat{h}}{\sigma \circ (F\otimes G)\circ \Delta(x)}
		=
		\eval{\hat{h}}{\sigma(F(x_{(1)})\otimes G(x_{(2)}))}
		\\ &
        =
		\eval{\hat{h}}{F(x_{(1)})}
            \eval{\hat{h}}{G(x_{(2)})}
		=
		\eval{\hat{h}}{x_{(1)}}
            \eval{\hat{h}}{x_{(2)}} 
        =
        \eval{\hat{h}}{x},
	\end{align*}
	where we have recalled the definition of $H$ in the first equality, we have rewritten the expression in Sweedler's notation in the second, we have employed the identity in \cref{eq:prophsigma} with $z = 1$ in the third, that $\hat{h}\circ F = \hat{h}\circ G = \hat{h}$ by hypothesis in the fourth, and finally we have recalled that $\hat{h}$ is an idempotent in $A^*$, see~\cref{eq:DualHaar}, in the last.
\end{proof}

At the level of maps on the algebra, the trajectory is given, for all $n\geq 2$, by
\begin{equation}
    F_{1} \coloneqq F
    , \quad \text{and} \quad
    F_{n} \coloneqq F_{n-1}\star F
         = \sigma \circ ( F_{n-1} \otimes F)\circ \Delta.
\end{equation}
The following is an illustrative example, corresponding to the boundary state of the two-dimensional toric code model.

\begin{example}
    Let us consider the group $C^*$-algebra $A=\mathbb{C}\mathbb{Z}_2 \equiv \mathbb{C}\{e,g\}$.
    According to \cref{prop:NF}, we look at completely positive linear maps
    $F : \mathbb{C}\mathbb{Z}_2 \to \mathbb{C}\mathbb{Z}_2$
    compatible with the trace-preserving condition.
    Let us write the linear map as
    \[
        F(e) = a e + b g
        ,\quad\text{and}\quad
        F(g) = c e + d g,
    \]
    for some $a,b,c,d\in \mathbb{C}$.
    Since $\hat{h} = \delta_e$, the condition $\hat{h}\circ F = \hat{h}$ enforces $a = 1$ and $c = 0$, yielding
    a two-parameter family of maps
    \[
        F(e) = e + \alpha g
        , \quad \text{and} \quad
        F(g) = \beta g.
    \]
    Moreover, $\alpha,\beta\in\mathbb{R}$ necessarily since $F$ is, in particular, self-adjoint-preserving, and $e = e^*$, $g=g^*$.
    Let $p_{\pm} \coloneqq (e\pm g)/2$ be the two minimal central idempotents. Then, since $F$ is completely positive, 
    \[
        F(p_\pm) = \frac{1}{2} (e+(\alpha\pm\beta)g) \geq 0
    \]
    which holds if and only if $|\alpha\pm \beta|\leq 1$;
    equivalently, $|\alpha|+|\beta| \leq 1$.
    \smallbreak\noindent%
    Let us consider the $*$-representation defined by the Pauli-Z matrix $\Phi(g) \coloneqq Z = \operatorname{diag}(1,-1)$ and let us analyze the action of such noise on the renormalization fixed-point matrix product density operator 
    \[
        \rho_n \coloneqq \rho_n(M,b(h))
        = 
        \frac{1}{2^{n}} (\mathrm{Id}^{\otimes n} + Z^{\otimes n}),
    \]
    associated with the Haar integral $h = (e+g)/2$.
    \smallbreak\noindent%
    If $\alpha=0$, then $F$ acts diagonally in the group basis, and the state is
    transformed to
    \[
        \Noisy_F^{\otimes n}(\rho_n)
        =
        \frac{1}{2^{n}}(\mathrm{Id}^{\otimes n} + \beta^n Z^{\otimes n}).
    \]
    Since $\Delta(g)=g\otimes g$, the convolution in \cref{eq:convMaps} reduces to
    \[
        (F \star \overset{n}\cdots \star F)(e) = e,
        \quad
        \text{and}
        \quad
        (F \star \overset{n}\cdots \star F)(g) = \beta^n g.
    \]
    Thus, if in addition $\beta\in\{0,1\}$, the coarse-graining leaves the state unchanged; if $|\beta|<1$,
    the non-trivial block decays and the limit coincides with the case $\beta=0$; and if
    $\beta=-1$, the sign alternates with the system size and no limit exists.
	\smallbreak\noindent%
	On the other hand, if $\alpha\neq 0$,
    \[
        \Noisy_F^{\otimes n}(\rho_n)
        =
        \frac{1}{2^{n}} ( (\mathrm{Id}+\alpha Z)^{\otimes n} + \beta^n Z^{\otimes n} ),
    \]
    and the convolution satisfies
    \[
        (F \star \overset{n}\cdots \star F)(e) = e + \alpha^n g,
        \quad
    \text{and}
    \quad
        (F \star \overset{n}\cdots \star F)(g) = \beta^n g.
    \]
    If $|\alpha|<1$, the additional term $\alpha^n g$ disappears under coarse-graining and the behavior reduces to that of the case $\alpha=0$.
    If $\alpha=1$, the positivity constraint forces $\beta=0$,
    yielding a tensor-product fixed point.
    Finally, $\alpha=-1$ again leads to an alternating sign and therefore to a non-convergent flow.
	\smallbreak\noindent%
    Collecting the convergent cases, we find exactly three types of fixed-point states:
\begin{enumerate}[leftmargin=60pt,itemsep=2pt]
        \item[Case 1:] The unperturbed state
            $\tilde\rho_n = 2^{-n}(\mathrm{Id}^{\otimes n}+Z^{\otimes n})$,
            if $\alpha=0$ and $\beta=1$.
        \item[Case 2:] The maximally mixed state
            $\tilde\rho_n = 2^{-n}\mathrm{Id}^{\otimes n}$,
            if $|\alpha|,|\beta|<1$.
        \item[Case 3:] The tensor-product state
            $\tilde\rho_n = 2^{-n}(\mathrm{Id}+Z)^{\otimes n}$,
            if $\alpha=1$ and $\beta=0$.
    \end{enumerate}
Remarkably, each case corresponds to a subgroup of $\mathbb{Z}_2\times \mathbb{Z}_2$. The first case to the diagonal subgroup $\{(e,e),(g,g)\}$, the second case to the trivial subgroup $\{(e,e)\}$, and the third case to the subgroup $\{e\}\times\mathbb{Z}_2$. This correspondence will become transparent when we describe the general structure of the group case and its dual in \cref{sec:groups}.
\end{example}

\subsection{From algebra maps to linear functionals}

This subsection reframes the analysis of the noise quantum channel in a dual language.
Rather than working directly with the induced linear maps on the algebra, we pass to a formulation in terms of specific linear functionals on the tensor product $A\otimes A^{\mathrm{op}}$.
This perspective interfaces more directly with standard tools in the literature.

\begin{proposition}\label{prop:Fphi}
    Let $A$ be a $C^*$-Hopf algebra and let $F : A \to A$ be a completely positive linear map satisfying $\hat{h} \circ F = \hat{h}$. Then, the linear functional $\phi_F \in (A \otimes A^{\mathrm{op}})^*$ defined, for all $x, y \in A$, by the expression
	\begin{equation}\label{eq:FtophiF}
		\eval{\phi_F}{x\otimes y} \coloneqq \eval{\hat{h}}{F(x) y},
	\end{equation}
	is a unital positive linear functional on $A\otimes A^{\mathrm{op}}$ satisfying
	\begin{equation}\label{eq:TPphiF}
		\eval{\phi_F}{(\placeholder)\otimes 1} = \hat{h}.
	\end{equation}
	Conversely, for every positive linear functional $\phi\in (A\otimes A^{\mathrm{op}})^*$ satisfying
$\eval{\phi}{(\placeholder)\otimes 1}=\hat{h}$, the linear map $F_\phi:A\to A$ defined, for all $x\in A$, by the expression
	\begin{equation}\label{eq:phitoFphi}
		F_\phi(x) \coloneqq D \eval{\phi}{x\otimes S(h_{(1)})} h_{(2)},
	\end{equation}
	is a completely positive linear map satisfying $\hat{h}\circ F_\phi = \hat{h}$. Moreover, the maps $F\mapsto \phi_F$ and $\phi\mapsto F_\phi$ establish a one-to-one correspondence.
\end{proposition}

\begin{proof}
Consider a completely positive linear map $F : A \to A$ and let us prove that $\phi_F$ is a positive linear functional.
For that purpose, let $z \coloneqq \sum_{i=1}^s x_i\otimes y_i \in A\otimes A^{\mathrm{op}}$ be arbitrary. Then any positive element in $A\otimes A^{\mathrm{op}}$ is of the form
\begin{equation}\label{eq:positiveInAAop}
z^*z = \sum_{i,j=1}^s x_i^* x_j \otimes  y_jy_i^*.
\end{equation}
Note the order of the product $y_jy_i^*$ in $A$, since $y_i^*,y_j\in A^{\mathrm{op}}$.
Since $F$ is completely positive by assumption, for all $s\in\mathbb{N}$ and $x_1,\ldots,x_s\in A$, the matrix
    \begin{equation}\label{eq:CPmatrix}
    	M \coloneqq [F(x_i^* x_j)]_{i,j=1}^{s} \in \mathrm{Mat}_s(A)
    \end{equation}
    is positive semidefinite, and hence $M= N^\dagger N$ for some $N\in \mathrm{Mat}_s(A)$. Then,
    \begin{equation}\label{eq:yMy_pos}
    	{y}^\dagger M {y}
    	= \sum_{i,j=1}^s y_i^* F(x_i^* x_j) y_j
    	= (N{y})^\dagger(N {y})
    	\geq 0 \quad\text{in } A,
    \end{equation}
    for all $y \coloneqq (y_1,\ldots,y_s)^{\mathrm{t}}\in A^s$. Consequently,
    \begin{align*}
    	\eval{\phi_F}{z^*z}
    	&=
    	\sum_{i,j=1}^s\eval{\hat{h}}{F(x_i^* x_j) y_j y_i^*}
    	=
    	\sum_{i,j=1}^s \eval{\hat{h}}{y_i^* F(x_i^* x_j) y_j}
    	\geq 0,
    \end{align*}
    where we used the definition of $\phi_F$ together with \cref{eq:positiveInAAop} in the first equality, the trace-like property of $\hat{h}$ from \cref{eq:DualHaar} in the second, and the positivity of $\hat{h}$ in the last.
    \smallbreak\noindent%
    Moreover, for all $x\in A$,
	\begin{equation}
    \eval{\phi_F}{x \otimes 1}
    =
    \eval{\hat{h}}{F(x)1}
    =
    \eval{\hat{h}}{x},
    \end{equation}
    since $\hat{h} \circ F = \hat{h}$ by assumption. In particular, $\phi_F$ is unital.
    \smallbreak\noindent%
	Conversely, given $\phi \in (A\otimes A^{\mathrm{op}})^*$ as above, the complete positivity of $F_\phi$ follows since
    \[ F_\phi = D (\phi\otimes \mathrm{Id}) \circ (\mathrm{Id}\otimes S\otimes \mathrm{Id}) \circ (\mathrm{Id}\otimes \Delta) \circ ((\placeholder)\otimes h),
    \]
    which is a composition of completely positive maps. Indeed, $x\mapsto x\otimes h$ is completely positive since $h$ is positive, $\Delta$ is a $*$-algebra homomorphism from $A$ to $A\otimes A$, the antipode is a $*$-algebra homomorphism $S:A\to A^{\mathrm{op}}$, and finally, $\phi\otimes \mathrm{Id}$ is completely positive because $\phi$ is a positive linear functional.
    \smallbreak\noindent%
	Moreover, for all $x\in A$,
	\begin{equation}
		\eval{ \hat{h} }{ F_\phi(x) }
		=
		D
        \eval{ \phi }{ x\otimes S(h_{(1)}) }
        \eval{ \hat{h} }{ h_{(2)} }
		=
		\eval{ \phi }{ x\otimes S(1) }
		=
		\eval{ \hat{h} }{ x },
	\end{equation}
    where the first equality uses the definition of $F_\phi$, the second equality uses the first property in \cref{eq:DualInt} to contract $D \eval{\hat{h}}{h_{(2)}}h_{(1)}$, and the final equality uses $S(1)=1$ and the assumed property of $\phi$ with respect to the Haar functional $\hat{h}$.
    \smallbreak\noindent%
	To conclude the proof, we check that the two correspondences are indeed inverse to one another. The key point is that for any state $\phi\in (A\otimes A^{\mathrm{op}})^*$ and all $x,y\in A$,
    \begin{align*}
        \eval{ \hat{h} }{ F_\phi(x) y }
        &= \eval{ \hat{h} }{ D \eval{ \phi }{ x  \otimes S(h_{(1)}) } h_{(2)} y }\\
        &= \eval{ \hat{h} }{ D \eval{ \phi }{ x  \otimes S(h_{(1)})\cdot_{\mathrm{op}} y } h_{(2)}}\\
        &= \eval{ \phi }{ x  \otimes  S(1) \cdot_{\mathrm{op}}y }\\
        &= \eval{ \phi }{ x \otimes y },
    \end{align*}
    where the second equality uses the pulling-through property of \cref{eq:PTmodified}, while the third equality pulls the $\eval{\hat{h}}{h_{(2)}}$ inside and uses \cref{eq:DualInt}.
    As such, for all $x,y\in A$ and for every state $\psi\in (A \otimes A^{\mathrm{op}})^*$, it holds that
	\[
		\eval{ \phi_{F_\psi} }{ x \otimes y } =
		\eval{ \hat{h} }{ F_\psi(x) y } =
		\eval{ \psi }{ x \otimes y },
	\]
    and therefore, in particular,
	\[
		\eval{ \hat{h} }{ F_{\phi_G}(x) y }
        =
		\eval{ \phi_G }{ x \otimes  y }
		=
		\eval{ \hat{h} }{ G(x) y },
	\]
    which, because $\hat{h}$ is non-degenerate, proves that $F_{\phi_G} = G$.
\end{proof}

Having constructed the state $\phi$, one can show that the action of the coarse-graining quantum channel $\Coarse$ is precisely given by a convolution of states, by combining \cref{prop:sigma_is_S} with the following proposition.

\begin{proposition}\label{prop:sigma_is_convolution}
	Let $F,G,H:A\to A$ be three linear maps satisfying the assumptions of \cref{prop:NF} and let $\phi_F,\phi_G,\phi_H\in (A\otimes A^{\mathrm{op}})^*$ be their corresponding states as in \cref{prop:Fphi}.
	Then, $H = F \star G$ if and only if $ \phi_H = \phi_F \phi_G$.
\end{proposition}

\begin{proof}
	Let us first assume that $H = F\star G$. Then, for all $x,y\in A$,
	\begin{align}
		\eval{\phi_H}{x \otimes y}
		  & =
		\eval{ \hat{h} }{ H(x) y}
		\\
		  & =
		\eval{ \hat{h} }{\sigma(F(x_{(1)}) \otimes G(x_{(2)})) y }
		\\
		  & =
		\eval{ \hat{h} }{ F(x_{(1)}) y_{(1)} }
		\eval{ \hat{h} }{ G(x_{(2)}) y_{(2)} }
		\\
		  & =
		\eval{ \phi_F}{ x_{(1)} \otimes y_{(1)}}
            \eval{\phi_G }{x_{(2)} \otimes y_{(2)} }
                \vphantom{\hat{h}} 
		\\
		  & =
        \eval{ (\phi_F \otimes \phi_G) \circ \Delta_{A\otimes A^{\mathrm{op}}} }{ x \otimes y },
            \vphantom{\hat{h}} 
	\end{align}
    where the first equality uses the correspondence between $H$ and $\phi_H$ established in \cref{prop:Fphi}, the third equality follows from \cref{lemma:sigma}, and the last line rewrites the resulting expression in terms of $\phi_F$ and $\phi_G$ using \cref{prop:Fphi}.
    \smallbreak\noindent%
    Conversely, assume now that $\phi_H = \phi_F \phi_G$. Then, for all $x\in A$,
	\begin{align}
		H(x)
		  & =
    		D
            \eval{\phi_H}{x\otimes S(h_{(1)})}
            h_{(2)}
            \vphantom{\hat{h}}
    		\\
        & =
    		D
            \eval{\phi_F}{x_{(1)}\otimes S(h_{(1)})_{(1)}}
            \eval{\phi_G}{x_{(2)}\otimes S(h_{(1)})_{(2)}} 
            h_{(2)}
            \vphantom{\hat{h}}
    		\\
		  & =
    		D
            \eval{\hat{h}}{ F(x_{(1)})S(h_{(1)})_{(1)} }
            \eval{\hat{h}}{G(x_{(2)})S(h_{(1)})_{(2)}}
            h_{(2)}
    		\\
		  & =
    		D
            \eval{\hat{h}}{\sigma(F(x_{(1)}) \otimes G(x_{(2)})) S(h_{(1)})}
            h_{(2)}
    		\\
        & =
    		D
            \eval{\hat{h}}{(F\star G)(x) S(h_{(1)})}
            h_{(2)}
    		\\
        & =
    		D
            \eval{\hat{h}}{ S(h_{(1)})}
            h_{(2)} (F\star G)(x)
    		\\
		  & =
		      (F\star G)(x) \vphantom{\hat{h}},
	\end{align}
	where we used \cref{eq:phitoFphi} in the first equality,
    the hypothesis about $\phi_H$ in the second, the definition of $\phi_F$ and $\phi_G$ in the third, the property of \cref{eq:prophsigma} in \cref{lemma:sigma} in the fourth,
    the definition of $F\star G$ in the fifth, and the second pulling-through identity in \cref{eq:PTmodified} in the sixth. The conclusion follows from the first property of \cref{eq:DualInt} together with the first property of \cref{eq:DualHaar}.
\end{proof}

Therefore, the renormalization procedure induces an equivalent trajectory,
\begin{equation}
\phi_{n} \coloneqq \phi_{\smash{F_{n}}} = \phi^n \quad\text{ for }n\geq 1,
\end{equation}
at the level of linear functionals on $A\otimes A^{\mathrm{op}}$, where $\phi \coloneqq \phi_{F}$.

\section{Fixed points and their algebraic structure}
\label{sec:fixedpoints}

This section carries out the characterization of the new renormalization fixed points reachable by the trajectories of the previous section. Now equipped with a good understanding of those renormalization flows at the algebraic level, we can describe, in the case of convergence, the structure of those fixed points using the concept of idempotent states and other equivalent algebraic objects. See \cref{fig:diagramAll} for a diagram summarizing the results.

\begin{remark}\label{rem:fixedNFphi}
Let $A$ be a $C^*$-Hopf algebra. Consider a noise quantum channel $\Noisy$ with corresponding linear map $F:A\to A$ and linear functional $\phi\in (A\otimes A^{\mathrm{op}})^*$, as described in \cref{prop:NF,prop:Fphi}.
Then, the following are equivalent:
    \begin{enumerate}
        \item[(i)] $\Coarse \circ (\Noisy\otimes \Noisy) \circ \Fine = \Noisy$;
        \item[(ii)] $\sigma \circ (F\otimes F) \circ \Delta = F$;
        \item[(iii)] $(\phi\otimes \phi)\circ \Delta_{A\otimes A^{\mathrm{op}}} = \phi$.
    \end{enumerate}
This is an immediate consequence of \cref{prop:FGconv,prop:sigma_is_convolution}.
\end{remark}

\begin{figure}[h!]
    \centering
\begin{tikzpicture}[scale=1]
    \node[anchor=center] (nodeF) at (3.5, 4) {
        \begin{minipage}{3.951cm}\centering
                $F:A \to A$
                \\
                \footnotesize{}completely positive linear map satisfying
                $\hat{h}\circ F =\hat{h}$ such that $\sigma\circ (F\otimes F)\circ \Delta = F$
        \end{minipage}
         };
    \node[anchor=center] (nodeN) at (-3.5, 4) {
         \begin{minipage}{4.092cm}\centering
                $\Noisy:\operatorname{End}\mathscr{H}\to\operatorname{End}\mathscr{H}$
                \\
                \footnotesize{} completely positive and trace-preserving linear map
                such that $\Coarse\circ (\Noisy\otimes\Noisy)\circ \Fine=\Noisy$
         \end{minipage}
         };
    \node[anchor=center] (nodephi) at (3.5, 0.5) {
         \begin{minipage}{4.516cm}\centering
                 $\phi : A\otimes A^{\mathrm{op}}\to\mathbb{C}$
                 \\
                 \footnotesize{}unital positive linear functional satisfying $\langle \phi,(\placeholder)\otimes 1\rangle = \hat{h}$  such that $\phi^2 = (\phi\otimes\phi)\circ \Delta = \phi$
         \end{minipage}
         };
    \node[anchor=center] (nodeQ) at (-3.5, 0.5) {
         \begin{minipage}{5.2cm}\centering
                  $(Q,m,1_Q,*,\Delta_Q,\varepsilon_Q,\varphi) \subseteq A\otimes \smash{A^{\mathrm{op}}}$
                  \\
                  \footnotesize{}finite $*$-quantum hypergroup
         \end{minipage}
         };
    \node[anchor=center] (nodep) at (-3.5, -2.7) {
         \begin{minipage}{4.516cm}\centering
                  $p \in A\otimes A^{\mathrm{op}}$
                  \\
                  \footnotesize{}group-like projection such that $\eval{ \hat{h}\otimes \hat{h} }{ p((\placeholder)\otimes 1)} = \eval{ \hat{h}\otimes\hat{h} }{ p }\hat{h}$
         \end{minipage}
         };
    \node[anchor=center] (nodeC) at (3.5, -2.7) {
         \begin{minipage}{4cm}\centering
              $C \subseteq A\otimes A^{\mathrm{op}}$
              \\
              \footnotesize{} coideal $*$-subalgebra such that $(\mathrm{Id}\otimes\hat{h})(C) = \mathbb{C}1$
         \end{minipage}
         };
    \draw[draw=none,shorten <=-2pt, shorten >=2pt]   (nodeN)   -- node[sloped, below, font=\footnotesize] {Rem.~\ref{rem:fixedNFphi}} (nodeF);
    \draw[<->,black,shorten <=-2pt, shorten >=2pt]   (nodeN)   -- node[sloped, above, font=\footnotesize] {Prop.~\ref{prop:NF}} (nodeF);
    \draw[draw=none] (nodephi)   -- node[sloped, below, font=\footnotesize] {Prop.~\ref{prop:qHG_3}} (nodeQ);
    \draw[->,black] (nodephi)   -- node[sloped, above, font=\footnotesize] {Prop.~\ref{prop:qHG_1}} (nodeQ);
    \draw[draw=none] (nodeF)   -- node[sloped, below, font=\footnotesize] {Rem.~\ref{rem:fixedNFphi}} (nodephi);
    \draw[<->,black] (nodeF)   -- node[sloped, above, font=\footnotesize] {Prop.~\ref{prop:Fphi}} (nodephi);
    \draw[<->,black]  (nodephi) -- node[sloped, above, font=\footnotesize] {Prop.~\ref{prop:Cphi}} (nodeC);
    \draw[<->,black,shorten <=-1pt, shorten >=-10pt]    (nodephi) -- node[sloped, above, font=\footnotesize] {Prop.~\ref{prop:phip}}   (nodep);
    \draw[draw=none]   (nodep)   -- node[sloped, above, font=\footnotesize] {Prop.~\ref{prop:qHG_2}} (nodeQ);
    \draw[->,black]   (nodep)   -- node[sloped, below, font=\footnotesize] {Prop.~\ref{prop:qHG_3}} (nodeQ);
    \draw[densely dashed,<-] (nodeQ) -- (nodeN);
\end{tikzpicture}
    \caption{Summary of the relations between the different structures.}
    \label{fig:diagramAll}
\end{figure}
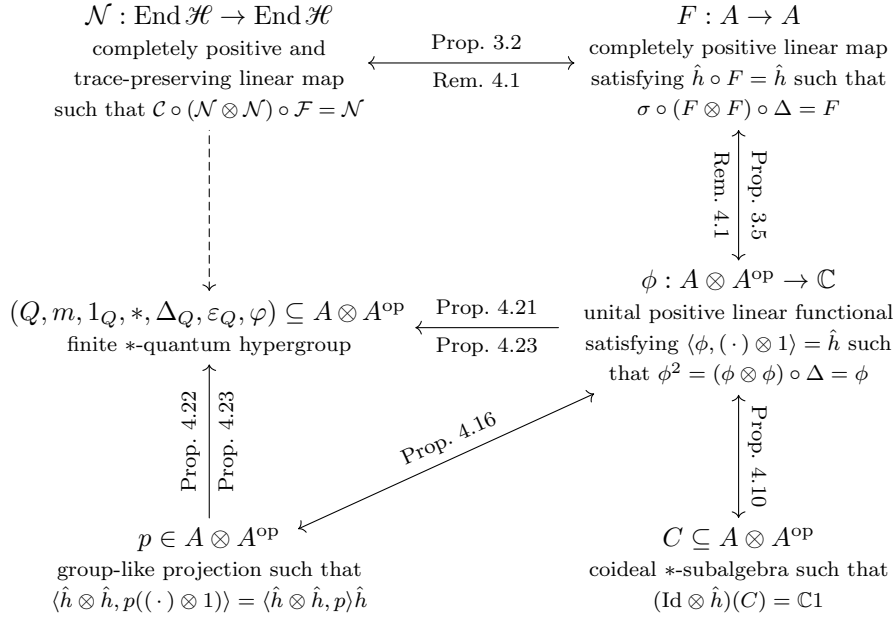
Let us now state a first version of our first main result.
\begin{proposition}
\label{th:NewFixedPoints}
	Let $A$ be a $C^*$-Hopf algebra and let $\Phi$ and $\Psi$ be faithful $*$-representations of $A$ and $A^*$, respectively. Recall that the associated renormalization fixed point matrix product density operators are generated by the tensor
	\begin{equation*}
		M = \sum_{i=1}^{\dim A} \Psi(e^i) \otimes \hat{b}(\hat{h}) \Phi(e_i).
	\end{equation*}
	After perturbation by a quantum channel $\Noisy$, we consider the coarse-graining trajectory $M_{n}$ and, assuming convergence, the new renormalization fixed point $M_{\infty}$ defined in \cref{sec:trajectories}. This new fixed point corresponds to an idempotent state $\phi$ on $A\otimes A^{\mathrm{op}}$ such that $\eval{\phi}{(\placeholder)\otimes 1} = \smash{\hat{h}}$ and we have
	\begin{align*}
		M_{\smash{\infty}}
		         & = \sum_{i=1}^{\dim A}  D  \eval{\phi}{e_i\otimes S(h_{(1)})}  \Psi(e^i) \otimes \hat{b}(\hat{h}) \Phi(h_{(2)}).
	\end{align*}
    Conversely, any idempotent state $\phi\in (A\otimes A^{\mathrm{op}})^*$ satisfying $\eval{\phi}{(\placeholder)\otimes 1} = \smash{\hat{h}}$ defines a valid renormalization fixed point $M_{\infty}$.
\end{proposition}

The limit tensors above satisfy the renormalization fixed point condition. Indeed, there exists a coarse-graining quantum channel, identical to $\Coarse$ by construction, together with a new fine-graining quantum channel $\tilde{\Fine}$; i.e.,
\begin{equation}
\begin{tikzpicture}[scale=0.7]
    \draw[RedOrange, -mid] (-3.246, 0.564) -- (-3.81, 0.564);
    \draw[RedOrange, -mid] (-1.834, 0.564) -- (-2.399, 0.564);
    \draw[-mid] (-2.822, 0.988) -- (-2.822, 1.552);
    \draw[-mid] (-2.822, -0.423) -- (-2.822, 0.141);
    \filldraw[thin, fill=white] (-2.399, 0.141) rectangle (-3.246, 0.988);
    \node[anchor=center,scale=0.8] at (-2.822, 0.564) {$M_\infty$};
    \draw[RedOrange, -mid] (-0.423, 0.564) -- (-0.988, 0.564);
    \draw[RedOrange, -mid] (0.988, 0.564) -- (0.423, 0.564);
    \draw[-mid] (0, 0.988) -- (0, 1.552);
    \draw[-mid] (0, -0.423) -- (0, 0.141);
    \draw[RedOrange, -mid] (2.399, 0.564) -- (1.834, 0.564);
    \draw[-mid] (1.411, 0.988) -- (1.411, 1.552);
    \draw[-mid] (1.411, -0.423) -- (1.411, 0.141);
    \filldraw[thin, fill=white] (0.423, 0.141) rectangle (-0.423, 0.988);
    \node[anchor=center,scale=0.8] at (0, 0.564) {$M_\infty$};
    \filldraw[thin, fill=white] (1.834, 0.141) rectangle (0.988, 0.988);
    \node[anchor=center,scale=0.8] at (1.411, 0.564) {$M_\infty$};
    \draw[|->] (-1.976, 1.129) arc[start angle=-135, end angle=-45, x radius=0.798, y radius=-0.798];
    \draw[|->] (-0.847, 0) arc[start angle=-135, end angle=-45, x radius=-0.798, y radius=0.798];
    \node[anchor=center] at (-1.411, 1.693) {$\tilde{\Fine}$};
    \node[anchor=center] at (-1.411, -0.564) {$\Coarse$};
\end{tikzpicture}\ .
\end{equation}
Before proving that this quantum channel is indeed the fine-graining quantum channel in \cref{prop:newFine}, we establish several key properties the proof relies upon in \cref{sec:idempotents}.
The reason this new fine-graining quantum channel works can be understood thanks to a deeper algebraic decomposition of the idempotent states and substructures of the enveloping $C^*$-Hopf algebra $A\otimes A^{\mathrm{op}}$. We develop this algebraic framework in \cref{sec:quantumgoursat}, culminating in more explicit description of the new renormalization fixed points in \cref{th:newGoursatFP}.

\subsection{Idempotent linear functionals}\label{sec:idempotents}

From now on, and for the sake of clarity, we consider an arbitrary $C^*$-Hopf algebra $B$, which will stand for $B \coloneqq  A\otimes A^{\mathrm{op}}$ as above.

\begin{definition}
	Let $B$ be a $C^*$-Hopf algebra. A linear functional $\phi\in B^*$ is \newterm{idempotent} if
	\[
        (\phi \otimes \phi) \circ \Delta_B = \phi.
	\]
    An \newterm{idempotent state} is a unital idempotent linear functional.
\end{definition}

Let us briefly state important properties that will be useful for this characterization.

\begin{proposition}\label{prop:IdempProp}
Let $\phi\in B^*$ be an idempotent state. Then, for all $v,w\in B$,
\begin{equation*}
    \eval{\phi}{v}
          \eval{\phi}{w}
    =
    \eval{\phi}{v_{(1)}}
          \eval{\phi}{v_{(2)}w}.
\end{equation*}
Moreover, $\phi$ is invariant under the antipode and a positive element, i.e., $\phi\circ S = \phi = \phi^*$.
\end{proposition}

\begin{proof}
    The proof adapts Lemma~3.1 of Ref.~\cite{franz_new_2009}.
	Let
    \[
        z
        \coloneqq 
        \eval{\phi}{v_{(1)}} v_{(2)}\in B
        ,\quad 
        Q
        \coloneqq 
        \Delta(z) - 1 \otimes z \in B \otimes B.
    \]
	First, by expanding the previous expressions and using that $\phi$ is idempotent,
	\begin{align*}
		\eval{\phi \otimes \phi}{Q^* Q}
		  & =
		2 \eval{\phi}{z^* z}
		      - \eval{\phi}{z^*_{(1)}}
		\eval{\phi}{z^*_{(2)}z}
		      - \eval{\phi}{z_{(1')}}
		          \eval{\phi}{z^* z_{(2')}}.
	\end{align*}
	In the case of the first term, let us note that
	\begin{align*}
		\eval{\phi}{z^*z}
		  & =
		\overline{\eval{\phi}{v_{(1)}}}
		      \eval{\phi}{ v_{(1')} }
		      \eval{\phi}{ v_{(2)}^* v_{(2')}^{} },
	\intertext{while in the case of the second term, this becomes}
		\eval{\phi}{z_{(1)}^*}
		      \eval{\phi}{z_{(2)}^* z}
		  & =
		\overline{\eval{\phi}{v_{(1)}}}
            \eval{\phi}{v_{(2)}^*}
                \eval{\phi}{v_{(3)}^* z}
		\\
		  & =
		\overline{\eval{\phi}{v_{(1)}}}
    		\eval{\phi}{v_{(1')}}
    		      \eval{\phi}{v_{(2)}^* v_{(2')}}
        =
        \eval{\phi}{z^*z},
        \\
    \intertext{and finally, by virtue of the previous equalities and the fact that $\phi$ is self-adjoint,}
		\eval{\phi}{z_{(1')}}
		      \eval{\phi}{z^* z_{(2')}}
		&=
		\overline{\eval{\phi}{z_{(1')}^*}}
            \overline{\eval{\phi}{z_{(2')}^*z}}
		=
        \overline{\eval{\phi}{z^*z}}
        =
        \eval{\phi}{z^*z}.
	\end{align*}
	Consequently, $\eval{\phi \otimes \phi}{Q^* Q} = 0$.
	By the Cauchy--Schwarz inequality for $\phi \otimes \phi$,
	\[
		|\eval{\phi \otimes \phi}{ Q (w \otimes h) }|^2
		\leq
        \eval{\phi \otimes \phi}{Q^* Q} 
		      \eval{\phi \otimes \phi}{w^* w\otimes h^* h}
        =
        0,
	\]
	and hence $\eval{\phi \otimes \phi}{Q(w\otimes h)} = 0$, i.e.,
	\[
		\eval{\phi \otimes \phi}{w \otimes z h}
		=
        \eval{\phi \otimes \phi}{\Delta(z)(w \otimes h)}.
	\]
	Recalling the definition of $z$, we obtain that
	\[
		\eval{\phi}{v_{(1)}} \eval{\phi}{v_{(2)} h} \eval{\phi}{w}
		=
        \eval{\phi}{v_{(1)}} \eval{\phi}{v_{(2)}w} \eval{\phi}{v_{(3)} h},
	\]
	since $v_{(1)} \otimes v_{(2)}h = v_{(1)} \otimes \eval{\varepsilon}{v_{(2)}}h = v \otimes h$,
	\[
		\eval{\phi}{v}\eval{\phi}{h}\eval{\phi}{w} 
		=
        \eval{\phi}{v_{(1)}} \eval{\phi}{v_{(2)}w} \eval{\phi}{h},
	\]
    and therefore
    \[
		\eval{\phi}{v}\eval{\phi}{w} 
		=
        \eval{\phi}{v_{(1)}} \eval{\phi}{v_{(2)}w}.
	\]
    since $\eval{\phi}{h} \neq 0$. 
    Indeed, let $\phi = u (\phi^* \phi)^{1/2}$ be the polar decomposition of $\phi$, where $u\in B^*$ is a partial isometry. Then, it is well known that the range projection $q \coloneqq u u^*\in B^*$ satisfies
    $q \phi =\phi$ and $ \phi q = q$, and hence
    \[
    \eval{\phi}{h}=\eval{q \phi }{h}=\eval{q}{h_{(1)}} \eval{\phi}{h_{(2)}} = \eval{q}{h_{(2)}} \eval{\phi}{h_{(1)}}=\eval{\phi q }{h} = \eval{q}{h}.
    \]
    Since $q$ is a positive projection and $\hat{h}q=q\hat{h}=\hat{h}$, $q\geq \hat{h}$ and therefore
    \[
        \eval{\phi}{h}=\eval{q}{h}\geq \eval{\hat{h}}{h} = D^{-1}>0.
    \]
	We can now prove the invariance of $\phi$ under the antipode:
	\begin{align*}
		\eval{\phi}{h}\eval{\phi}{w} 
        & =
        \eval{\phi}{h_{(1)}} \eval{\phi}{h_{(2)}w}
        =
        \eval{\phi}{h_{(1)}S(w)} \eval{\phi}{h_{(2)}}
		\\
        & =
        \eval{\phi}{h_{(1)}} \eval{\phi}{h_{(2)}S(w)}
        =
        \eval{\phi}{h} \eval{\phi}{S(w)},
	\end{align*}
	for all $w\in B$, where we used the property above for $v\equiv h$, the pulling-through of \cref{eq:PT} in the second equality, the fact that $h_{(1)} \otimes h_{(2)}= h_{(2)}\otimes h_{(1)}$ of \cref{eq:Haar} in the third, and the property for $v\equiv h$ again in the last.
\end{proof}

\begin{corollary}\label{cor:IdempPropExch}
    For any idempotent state $\phi\in B^*$ and all $v,w\in B$, it holds that
    \begin{equation*}
        \eval{\phi}{v}
            \eval{\phi}{w}
        =
        \eval{\phi}{w v_{(1)}}
            \eval{\phi}{v_{(2)}}.
    \end{equation*}
\end{corollary}

\begin{proof}
    Indeed, for all $v,w\in B$, it holds that
    \begin{align*}
        \eval{\phi}{v} & \eval{\phi}{w}
        =
        \eval{\phi}{S(v)}
            \eval{\phi}{S(w)}
        =
        \eval{\phi}{S(v)_{(1)}}
            \eval{\phi}{S(v)_{(2)} S(w)}
        \\
        &=
        \eval{\phi}{S(v_{(2)})}
            \eval{\phi}{S(v_{(1)}) S(w)}
        =
        \eval{\phi}{S(v_{(2)})}
            \eval{\phi}{S(w v_{(1)})}
        =
        \eval{\phi}{v_{(2)}}
            \eval{\phi}{w v_{(1)}},
    \end{align*}
    where we have recalled that $\phi\circ S = \phi$ by \cref{prop:IdempProp} in the first equality, we have applied the identity in \cref{prop:IdempProp} in the second, the anticomultiplicativity of the antipode in \cref{eq:Santicomult} in the third, its antimultiplicativity in \cref{eq:PropsAntipode} in the fourth, and we have used that $\phi \circ S = \phi$ again in the last.
\end{proof}

Let us now prove existence of the new fine-graining quantum channel at the fixed point.

\begin{proposition}\label{prop:newFine}
    Let $A$ be a $C^*$-Hopf algebra and consider a renormalization fixed point of the form exhibited in \cref{th:NewFixedPoints}, given by an idempotent state $\phi \in (A\otimes A^{\mathrm{op}})^*$, and let $F \coloneqq F_\phi$. Then, the linear map $\Delta_F : A \to A\otimes A$ defined, for all $x\in A$, by 
    \begin{equation}
    \Delta_F(x)=(1 \otimes F(1))\Delta(x),
    \end{equation}
    satisfies the intertwining property
    \begin{equation}\label{eq:new_FG}
        \Delta_F \circ F = (F \otimes F) \circ \Delta,
    \end{equation}
    and corresponds to the fine-graining quantum channel of the renormalization limit tensor,
    \begin{equation}
        \tilde\Fine \circ \Phi_b
        =
        (\Phi_b\otimes \Phi_b) \circ \Delta_F.
    \end{equation}
\end{proposition}

\begin{proof}
    We have $F(1)= F_\phi(1)= D \eval{\phi}{1 \otimes S(h_{(1)})}h_{(2)}$. For all $x\in A$,
    \begin{align*}
        (F\otimes F)(\Delta(x))
        & =
        D^2 \eval{\phi}{x_{(1)} \otimes S(h_{(1)})}
            \eval{\phi}{x_{(2)} \otimes S(h_{(1')})}
                h_{(2)} \otimes h_{(2')}
        \\ &
        =
        D^2 \eval{\phi}{x_{(1)} \otimes h_{(1)}}
            \eval{\phi}{S(x_{(2)}) \otimes h_{(1')}}
                S(h_{(2)}) \otimes h_{(2')}
        \\ &
        =
        D^2 \eval{\phi}{x_{(1)} \otimes h_{(1)}}
            \eval{\phi}{x_{(2)}
                S(x_{(3)})
                \otimes h_{(2)} \cdot_{\mathrm{op}} h_{(1')}} S(h_{(3)})
                \otimes h_{(2')}
        \\ &
        =
        D^2 \eval{\phi}{x \otimes h_{(1)}}
            \eval{\phi}{1 \otimes h_{(1')}}
                S(h_{(3)}) \otimes h_{(2')} S(h_{(2)})
        \\ &
        =
        D^2 \eval{\phi}{x \otimes h_{(1)}}
            \eval{\phi}{1 \otimes h_{(1')}}
                S(h_{(2)})_{(1)} \otimes h_{(2')} S(h_{(2)})_{(2)}
        \\
        &=
        D \eval{\phi}{x \otimes h_{(1)}}
            S(h_{(2)})_{(1)} \otimes F(1) S(h_{(2)})_{(2)} 
        \\ &
        =
        F(x)_{(1)} \otimes F(1) F(x)_{(2)},
    \end{align*}
    where we used the definition of $F$ in the first equality, the invariance of $\phi$ under the antipode and the fact that $S(h_{(1)}) \otimes h_{(2)}=h_{(1)} \otimes S(h_{(2)})$ in the second, \cref{prop:IdempProp} with $v\equiv x_{(1)} \otimes h_{(1)}$ and $w\equiv S(x_{(2)}) \otimes h_{(1')}$ in the third, the first pulling-through property of \cref{eq:PT} in the fourth, the comultiplicativity of the antipode in the fourth and the definition of $F$ in the fifth and the sixth.
\end{proof}

In order to get a better understanding and characterization of those fixed points, we recall in the next section the correspondence between the idempotent states defined above and coideal $*$-subalgebras, with their associated conditional expectation.

\subsection{Conditional expectations and coideal \texorpdfstring{$\boldsymbol{*}$}{*}-subalgebras}
\label{subsec:expect}

\begin{definitions}
    \label{def:coideal}
    Let $B$ be a $C^*$-Hopf algebra. A subset $C\subseteq B$ is a \newterm{(right) coideal} of $B$ if
    \[
        \Delta(C)\subseteq C\otimes B.
    \]
    A \newterm{coideal $*$-subalgebra} of $B$ is simultaneously a coideal and a $*$-subalgebra of $B$.
\end{definitions}

\begin{remark}
One can also consider the notion of ``left'' coideal $C\subseteq B$, satisfying the condition $\Delta(C)\subseteq B\otimes C$. However, both notions are in one-to-one correspondence. Indeed, if $C\subseteq B$ is a coideal, then 
$
    \Delta(S(C)) \subseteq
    (S\otimes S)(B \otimes C) = B\otimes S(C)
$,
and hence $S(C)$ is a left coideal, since $S$ is one-to-one. Thus, we will simply speak of coideal $*$-subalgebras.
\end{remark}

\begin{example}
    Coideal $*$-subalgebras of $\mathbb{C}G$ are of the form $\mathbb{C}H$, where $H$ is a subgroup of $G$. Similarly, coideal $*$-subalgebras of $\mathbb{C}^G$ are of the form $\mathbb{C}^{G/H}$, the algebra of functions over the left cosets $G/H$ where $H$ is a (not necessarily normal) subgroup of $G$.
\end{example}

In a general $C^*$-Hopf algebra, idempotent states are in one-to-one correspondence with coideal $*$-subalgebras \cite{franz_new_2009}. 
\begin{proposition}
\label{prop:Cphi}
    Let $B$ be a $C^*$-Hopf algebra and let $\phi\in B^*$ be an idempotent state. Then, the linear map $E_\phi : B \to B$ defined by
    \begin{equation}
        E_\phi
            \coloneqq
                (\phi\otimes\mathrm{Id})\circ \Delta
    \end{equation}
    is a positive conditional expectation onto a coideal $*$-subalgebra $C_\phi \coloneqq E_\phi(B)$ sa\-tis\-fying $\hat{h} \circ E_\phi = \hat{h}$.
    Conversely, for every coideal $*$-subalgebra $C$ of $B$ with con\-di\-tio\-nal ex\-pecta\-tion $E_C : B \to C$ such that $\hat{h} \circ E_C= \hat{h}$, the linear functional
    \begin{equation}
        \phi_C
            \coloneqq
                \varepsilon\circ E_C \in B^*
    \end{equation}
    is an idempotent state. Moreover, the maps $\phi\mapsto C_\phi$ and $C\mapsto \phi_C$ establish a one-to-one correspondence between such idempotent states and coideal $*$-subalgebras.
\end{proposition}

For the sake of completeness, we include an adapted proof of this result.

\begin{proof}
    Let $\phi \in B^*$ be an idempotent state. First, $E_\phi$ is a projection since, for all $x\in B$,
    \[
    E_\phi(E_\phi(x))=\eval{\phi}{x_{(1)}}\eval{\phi}{x_{(2)}} x_{(3)}=\eval{\phi}{x_{(1)}} x_{(2)} = E_\phi(x).
    \]
    The map $E_\phi$ is completely positive since it is the composition of $\Delta$, a $*$-algebra homomorphism, and $\phi \otimes \mathrm{Id}$, both of which are completely positive.
    In addition, $E_\phi$ is trivially unital since $\phi$ is unital.
    Furthermore, for all $x, y \in B$, we have
    \begin{align*}
        E_\phi(E_\phi(x)y)
        &=
        \eval{\phi}{x_{(1)}} E_\phi(x_{(2)} y)
        =
        \eval{\phi}{x_{(1)}}
            \eval{\phi}{x_{(2)} y_{(1)}}
                x_{(3)} y_{(2)}
        \\&=
        \eval{\phi}{x_{(1)}}
            \eval{\phi}{y_{(1)}} x_{(2)}
                y_{(2)}
        =
        E_\phi(x) E_\phi(y),
    \end{align*}
    where we used \cref{prop:IdempProp} in the third equality with $v \equiv x_{(1)}$ and $w \equiv y_{(1)}$. This also shows that $C_\phi$ is closed under multiplication. 
    Moreover, for all $x\in B$,
    \[
        E_\phi(x^*)
        =
        \eval{\phi}{x_{(1)}^*}x_{(2)}^*
        =
        (\overline{\eval{\phi}{x_{(1)}^*}}x_{(2)})^*
        =
        (\eval{\phi}{x_{(1)}}x_{(2)})^*
        =
        (E_\phi(x))^*,
    \]
    since $\phi$ is self-adjoint, making $C_\phi$ a $*$-subalgebra of $B$.
    For all $x \in B$,
    \begin{align*}
        \Delta(E_\phi(x))
        &=
        \eval{\phi}{x_{(1)}} x_{(2)} \otimes x_{(3)}
        =
        E_\phi(x_{(1)}) \otimes x_{(2)}
    \end{align*}
    and therefore $\Delta(C_\phi) \subseteq C_\phi \otimes B$, i.e., $C_\phi$ is a coideal.
    Finally, we check that for all $x \in B$,
    \[
    \eval{\hat{h}}{E_\phi(x)}
    =
    \eval{\hat{h}}{\eval{\phi}{x_{(1)}} x_{(2)}}
    =
    \eval{\phi}{\eval{\hat{h}}{x_{(2)}}x_{(1)}}
    =
    \eval{\hat{h}}{x}\eval{\phi}{1}
    =
    \eval{\hat{h}}{x},
    \]
    and hence $\hat{h}\circ E_\phi = \hat{h}$.
    \smallbreak\noindent%
    Conversely, let $C$ be a coideal $*$-subalgebra of $B$ and let $E : B \to C$ be the positive conditional expectation such that $\hat{h} \circ E = \hat{h}$, and define $\phi_C \coloneqq \varepsilon \circ E$. Since $E$ is completely positive and $\varepsilon$ is a $*$-algebra homomorphism, $\phi_C$ is trivially a positive linear functional, and $\eval{\phi_C}{1} = \eval{\varepsilon}{E(1)} = \eval{\varepsilon}{1} = 1$.
    \smallbreak\noindent
    Now, we claim that $E$ has the following property, valid for all $x\in B$:
    \begin{equation} \label{eq:colinearity}
        E(x_{(1)})\otimes x_{(2)}
        =
        E(x)_{(1)}\otimes E(x)_{(2)}.
    \end{equation}
    For that purpose, let us prove that $\ker E$ is the $\hat{h}$-orthogonal complement of $C$:
    \begin{equation}
        \label{eq:kerEcomplHaar}
        \ker E = \{ x\in B : \eval{\hat{h}}{c^* x} = 0 \text{ for all }c\in C\}.
    \end{equation}
    For the inclusion from left to right, note that for all $x\in B$ and $c\in C$,
    \[
        \eval{\hat{h}}{c^*E(x)}
        =
        \eval{\hat{h}}{E(c^*)E(x)}
        =
        \eval{\hat{h}}{c^*x},
    \]
    since $E$ is a conditional expectation with $\hat{h}\circ E = \hat{h}$. Conversely, for all $x\in B$ let $c \coloneqq E(x)$,
    \[
        0=\eval{\hat{h}}{E(x)^*x} = \eval{\hat{h}}{E(E(x)^*x)}=\eval{\hat{h}}{E(x)^*E(x)},
    \]
    and faithfulness of $\hat{h}$ implies $E(x)=0$, as we wanted to prove.
    Now, let us show that
    \begin{equation}
    \label{eq:kerECoideal}
        \Delta(\ker E) \subseteq \ker E \otimes B,
    \end{equation}
    or, equivalently, that, for all $x\in \ker E$,
    \begin{equation}
    \label{eq:kerECoidealWithPhi}
        (\mathrm{Id}\otimes \varphi)(\Delta(x)) \in\ker E
            \quad \text{ for all } \varphi\in B^*.
    \end{equation}
    Indeed, $x\in \ker E$ if and only if 
    \[
        \eval{\hat{h}}{c^*x} = 0
            \quad \text{ for all } c\in C,
    \]
    as proven in \cref{eq:kerEcomplHaar}; with the notation of \cref{eq:DefLRfunctional}, this is equivalent to:
    \[
        \eval{\hat{h}}{R_{\varphi^*}(c)^* x} = 0
            \quad \text{ for all }c\in C \text{ and }\varphi\in B^*,
    \]
    and since $(R_\varphi)^\dagger = R_{\varphi^*}$ by \cref{rem:DaggerLR}, the previous statement is equivalent to
    \[
        \eval{\hat{h}}{c^* R_{\varphi}(x)} = 0
            \quad \text{ for all }c\in C \text{ and }\varphi\in B^*,
    \]
    this is equivalent, by virtue of \cref{eq:kerEcomplHaar}, to the fact that
    \[
        R_{\varphi}(x) \in \ker E
            \quad \text{ for all } \varphi\in B^*,
    \]
    which is \cref{eq:kerECoidealWithPhi}.
    Now, let us combine \cref{eq:kerEcomplHaar,eq:kerECoideal}. For all $x\in B$ we can write $x = E(x) + (x-E(x))$, where $E(x)\in C$ and $x-E(x)\in \ker E$, and 
    \[
        (E\otimes \mathrm{Id})(\Delta(x)) =
        (E\otimes \mathrm{Id})(\Delta(E(x))) + (E\otimes \mathrm{Id})(\Delta(x-E(x))) = \Delta(E(x)),
    \]
    since, on the one hand, $\Delta(E(x))\in C\otimes B$, the map $E\otimes \mathrm{Id}$ leaves that term invariant while, on the other hand, $\Delta(x-E(x)) \in \ker E \otimes B$ and hence $E\otimes \mathrm{Id}$ annihilates it. This concludes the proof of \cref{eq:colinearity}.
    \smallbreak\noindent%
    In particular, this allows us to prove that $\phi_C$ is an idempotent since, for all $x\in B$,
    \begin{align*}
        \eval{\phi_C}{x_{(1)}}\eval{\phi_C}{x_{(2)}}
        &=
        \eval{\varepsilon}{E(x_{(1)})}\eval{\varepsilon}{E(x_{(2)})}
        \\&
        =
        \eval{\varepsilon}{E(x)_{(1)}}\eval{\varepsilon}{E(E(x)_{(2)})}
        =
        \eval{\varepsilon}{E(E(x))} = \eval{\varepsilon}{E(x)} = \eval{\phi_C}{x},
    \end{align*}
    where we unfolded the definition of $\phi_C$ in the first equality, applied \cref{eq:colinearity} in the second, used the fact that $E$ is involutive in the third, applied \cref{eq:axiom_eps_id} in the fourth, and recalled the definition of $\phi_C$ in the last. 
    \smallbreak\noindent%
    Finally, to complete the proof, we check that these constructions are mutually inverse.
    Starting with an idempotent state $\phi$, for all $x\in B$ it holds that
    \[ 
        \eval{\phi_{C_\phi}}{x}
        =
        \eval{\varepsilon}{E_\phi(x)}
        =
        \eval{\phi}{x_{(1)}}
            \eval{\varepsilon}{x_{(2)}}
        =
        \eval{\phi}{x}.
    \]
    Conversely, given a coideal $*$-sub\-al\-ge\-bra $C$ with conditional expectation $E$, for all $x\in B$,
    \begin{align*}
        E_{\phi_C}(x)
        &=
        \eval{\phi_C}{x_{(1)}}x_{(2)}
        =
        \eval{\varepsilon}{E(x_{(1)})}  x_{(2)}
        =
        \eval{\varepsilon}{E(x)_{(1)}}  E(x)_{(2)}=E(x),
    \end{align*}
    as we wanted to prove.
\end{proof}

We can also characterize the coideal $*$-subalgebras $C_\phi \subseteq A\otimes A^{\mathrm{op}}$ corresponding to idempotent states of $A\otimes A^{\mathrm{op}}$ with the trace-preserving condition.
\begin{remark}
Let $A$ be a $C^*$-Hopf algebra, let $B=A\otimes A^{\mathrm{op}}$ be the enveloping $C^*$-Hopf algebra of $A$, let $\phi\in B^*$ be an idempotent unital linear functional and let $E: B \to C$ be the corresponding conditional expectation on the coideal $*$-subalgebra $C$. Then, the following conditions are equivalent:
\begin{enumerate}
    \item[(i)] $\eval{\phi}{(\placeholder)\otimes 1}= \hat{h}$;
    \item[(ii)] $(\mathrm{Id}_A\otimes \hat{h})\circ E=\eval{\hat{h} \otimes \hat{h}}{(\placeholder)} 1$.
\end{enumerate}
Note that we then have $(\mathrm{Id}_A\otimes \hat{h})(C) = \mathbb{C}1$.
\end{remark}

\begin{proof}
    Let $x\in A$ and $y\in A^{\mathrm{op}}$, then
    \begin{align*}
        ((\mathrm{Id} \otimes \hat{h})\circ E)(x\otimes y)
        &= \eval{\phi}{x_{(1)}\otimes \eval{\hat{h}}{y_{(2)}}y_{(1)}}x_{(2)} = \eval{\phi}{x_{(1)}\otimes \eval{\hat{h}}{y} 1}x_{(2)}\\
        &= \eval{\hat{h}}{y}\eval{\hat{h}}{x_{(1)}}x_{(2)} = \eval{\hat{h}}{x}\eval{\hat{h}}{y} 1.
    \end{align*}
    where we used the fact that $E = (\phi \otimes \mathrm{Id}) \circ \Delta$ in the first equality, used \cref{eq:DualHaar2} in the second, then the trace-preserving condition of $\phi$ in the third and \cref{eq:DualHaar2} again in the last. Conversely, let $x\in A$, 
    \[
        \eval{\phi}{x \otimes 1}
       = (\varepsilon \otimes  \varepsilon) \circ E(x\otimes 1)
        = (\varepsilon \otimes  \hat{h}) \circ E(x\otimes 1)
        =\eval{\hat{h}}{x}.
    \]
    where we used that $\phi = (\varepsilon \otimes \varepsilon)\circ E$ in the first equality, the fact that $\eval{\hat{h}}{1}=  \eval{\varepsilon}{1} = 1$ in the second, and the trace-preserving condition of $E$ in the last.
\end{proof}

Contrary to what one might expect from the group case, coideal $*$-subalgebras might not, in general, correspond to a ``quantum subgroup'' structure: there are coideal $*$-sub\-al\-ge\-bras that are not $C^*$-Hopf subalgebras. An interesting example of this was given by Pal in \cite{pal_counterexample_1996}.
\begin{example}
    Let $H_8$ be the Kac--Paljutkin $C^*$-Hopf algebra introduced in \cref{example:H8} and consider the state given by
    \[
        \psi
        \coloneqq
        \frac{1}{4} (e^1+e^4) + \frac{1}{2} e^5
        \in H_8^*.
    \]
    One can check that it is an idempotent linear functional on $H_8$ and that the coideal $*$-subalgebra $C_\psi$ cannot be a $C^*$-Hopf subalgebra; in particular, $\Delta(C_\psi) \not\subset C_\psi \otimes C_\psi$.
\end{example}

\subsection{From idempotent states to group-like projections}

We now turn to an equivalent characterization of idempotent states, in terms of elements of the $C^*$-Hopf algebra $B$ called group-like projections. Those objects are more amenable to explicit computation and will be used extensively in the remainder of this paper.

\begin{definition}
\label{def:GrouplikeProjection}
	Let $B$ be a $C^*$-Hopf algebra.
	A non-zero self-adjoint projection $p\in B$ is called a \newterm{group-like projection} if
	\begin{equation}\label{eq:GrouplikeProjection}
		\Delta(p) (1 \otimes p) = p \otimes p.
	\end{equation}
\end{definition}

\begin{remark}\label{rem:basicPropGroupLike}
	Every group-like projection $p$ satisfies $ \eval{\varepsilon}{p} = 1 $ and $S(p) = p$, and the following equalities hold:
	\begin{equation}
    \label{eq:EquivGrouplikeProjDefs}
		p_{(1)} \otimes p_{(2)}p
        =
        p_{(1)} \otimes p p_{(2)}
        =
		p_{(1)}p \otimes p_{(2)}
        =
		p p_{(1)} \otimes p_{(2)}.
	\end{equation}
    Indeed, it is well known that $S^2 = \mathrm{Id}$ for all $C^*$-Hopf algebras, and hence the result follows as a consequence of Lemma~3.3 and Theorem~3.5 in Ref.~\cite{chirvasitu_integrals_2020}.
\end{remark}

\begin{remark}\label{rem:grouplike3}
    Note that if $p$ is a group-like projection, then
    \[
        p \otimes p \otimes p
        =
        p_{(1)} \otimes p_{(2)}p \otimes p_{(3)}p
        =
        p_{(1)}p \otimes p_{(2)} \otimes p_{(3)}p
        =
        p_{(1)}p \otimes p_{(2)}p \otimes p_{(3)}.
    \]
The first equality follows from the following calculation:
    \begin{align*}
        p \otimes p \otimes p
        =
        \Delta(p)(1\otimes p) \otimes p
        =
        \Delta(p_{(1)})(1\otimes p) \otimes p_{(2)} p
        =
        p_{(1)}\otimes p_{(2)} p \otimes p_{(3)} p.
    \end{align*}
    The other identities follow analogously.
\end{remark}

As established in Ref.~\cite{Franz_2009}, idempotent states are in one-to-one correspondence with the following group-like projections.
\begin{proposition}
\label{prop:phip}
    Let $B$ be a $C^*$-Hopf algebra.
	For every group-like projection $p \in B$, the linear functional $\phi_p \in B^*$ defined by
	\begin{equation}\label{eq:pTophi}
        \phi_p
        \coloneqq
        \frac{1}{\eval{\hat{h}}{p}}\eval{\hat{h}}{p (\placeholder) p}
        =
        \frac{1}{\eval{\hat{h}}{p}}\eval{\hat{h}}{(\placeholder) p}
        \in B^*
	\end{equation}
    is an idempotent state in $B^*$.
	Conversely, for every idempotent state $\phi \in B^*$,
	\begin{equation}\label{eq:phiTop}
		p_\phi
        \coloneqq
        \frac{1}{\eval{\phi}{h}}
            \eval{\phi}{h_{(1)}}
                h_{(2)}
        \in B
	\end{equation}
	is a group-like projection. Moreover, these maps establish a one-to-one correspondence between group-like projections and idempotent states.
\end{proposition}

\begin{proof}
	Let $p\in B$ be a group-like projection. First, $\eval{\hat{h}}{p} > 0$ and $\phi_p$ is a positive linear functional since $\hat{h}$ is positive and $p$ is a self-adjoint projection. Note that $\eval{\phi_p}{1} = 1$ by construction. Let us show that $\phi_p$ is idempotent. Indeed, for all $x\in A$,
	\begin{align*}
        \eval{\phi_p}{x_{(1)}}
            \eval{\phi_p}{x_{(2)}}
        &=
        \eval{\hat{h}}{p}^{-2}
            \eval{\hat{h}}{x_{(1)}p}
                \eval{\hat{h}}{x_{(2)}p}
        \\&
        =
        \eval{\hat{h}}{p}^{-2}
            \eval{\hat{h}}{x_{(1)}p_{(1)}}
                \eval{\hat{h}}{x_{(2)} p_{(2)} p}
        =
        \eval{\hat{h}}{p}^{-2}
            \eval{\hat{h}}{xp}\eval{\hat{h}}{p}
        =
        \eval{\phi_p}{x}
	\end{align*}
    where we used the definition of $\phi_p$ in the first equality, the group-like projection property $p \otimes p = p_{(1)} \otimes p_{(2)}p$ in the second and finally the multiplicativity of $\Delta$ and the integral property of \cref{eq:DualHaar2} in the third.
    \smallbreak\noindent%
    Conversely, let us consider an idempotent state $\phi \in B^*$ and prove that $p_\phi$ is a group-like projection. We show that $p_\phi$ is a projection:
    \begin{align*}
        p_\phi^2
        &=
        \eval{\phi}{h}^{-2}\eval{\phi}{h_{(1)}} \eval{\phi}{h_{(1')}} h_{(2)} h_{(2')}
        =
        \eval{\phi}{h}^{-2}\eval{\phi}{h_{(1)}} \eval{\phi}{S(h_{(2)})h_{(1')}} h_{(2')}\\
        &=
        \eval{\phi}{h}^{-2}\eval{\phi}{h_{(1)}} \eval{\phi}{h_{(2)}h_{(1')}} h_{(2')}
        =
        \eval{\phi}{h}^{-2}\eval{\phi}{h} \eval{\phi}{h_{(1')}} h_{(2')}
        =
        p_\phi,
    \end{align*}
    where we used the definition of $p_\phi$ in the first equality, the first pulling-through property of \cref{eq:PT} in the second, the fact that $h_{(1)} \otimes S(h_{(2)})=S(h_{(1)}) \otimes h_{(2)}$ together with the invariance of $\phi$ under the antipode in the third and \cref{prop:IdempProp} with $v\equiv h$ and $w\equiv h_{(1')}$ in the last.
    \smallbreak\noindent%
    We show the group-like property
    \begin{align*}
    	\Delta(p_\phi)(1\otimes p_\phi)
    	&=
    	\eval{\phi}{h}^{-2} \eval{\phi}{h_{(1)}} \eval{\phi}{h_{(1')}} h_{(2)} \otimes h_{(3)} h_{(2')}
    	\\ &=
    	\eval{\phi}{h}^{-2} \eval{\phi}{h_{(1)}} \eval{\phi}{S(h_{(1')})} h_{(2)} \otimes h_{(3)} h_{(2')}
    	\\ &=
    	\eval{\phi}{h}^{-2} \eval{\phi}{h_{(1)}} \eval{\phi}{S(h_{(1')})h_{(3)}} h_{(2)} \otimes h_{(2')}
    	\\ &=
        \eval{\phi}{h}^{-2} \eval{\phi}{h_{(2)}} \eval{\phi}{S(h_{(1')})h_{(1)}} h_{(3)} \otimes h_{(2')}
        \\ &=
        \eval{\phi}{h}^{-2} \eval{\phi}{h_{(1)}} \eval{\phi}{S(h_{(1')})} h_{(2)} \otimes h_{(2')}
        \\ &=
        \eval{\phi}{h}^{-2} \eval{\phi}{h_{(1)}} \eval{\phi}{h_{(1')}} h_{(2)} \otimes h_{(2')}
        \\ &=
    	p_\phi \otimes p_\phi,
    \end{align*}
	where we used the definition of $p_\phi$ in the first equality, the invariance of $\phi$ under the antipode in the second, the first pulling-through of \cref{eq:PTmodified} in the third, the cyclicity of $h_{(1)}\otimes h_{(2)} \otimes h_{(3)}$ in the fourth, \cref{prop:IdempProp} with $v\equiv h_{(1)}$ and $w\equiv S(h_{(1')})$ in the fifth and the invariance of $\phi$ under the antipode again in the sixth.
    \smallbreak\noindent%
    Finally, we check that the two constructions are inverse to each other.
	\begin{equation*}
		p_{\phi_q}
        =
        \frac{\eval{ \phi_q }{ h_{(1)} }}{\eval{ \phi_q }{ h }} h_{(2)}
        =
        \frac{\eval{ \hat{h} }{ q h_{(1)} }}{\eval{ \hat{h} }{ qh }} h_{(2)}
		  =
        \frac{\eval{ \hat{h} }{ h_{(1)} }}{\eval{ \hat{h} }{ qh }} S(q) h_{(2)}
		  =
        \frac{1}{\eval{\varepsilon}{q}} S(q) =
        q,
	\end{equation*}
	where we used the first pulling-through identity of \cref{eq:PT} in the third equality, then simplified the fraction using \cref{eq:DualInt} (applying its first property to the numerator and its second property to the denominator) in the fourth equality and finally concluded using $\varepsilon(q)=1$ and $S(q)=q$ given in \cref{rem:basicPropGroupLike}.
    \smallbreak\noindent%
	Conversely, for all $x\in B$,
	\begin{align*}
		\eval{\phi_{p_\psi}}{x}
        &=
        \frac{\eval{\hat{h}}{p_\psi x}}{\eval{\hat{h}}{p_\psi}}
        =
        \frac{\eval{\hat{h}}{\eval{\psi}{h_{(1)}} h_{(2)} x}}
		{\eval{\hat{h}}{\eval{\psi}{h_{(1')}} h_{(2')}}}%
		=
        \frac{\eval{\psi}{h_{(1)}} 
            \eval{\hat{h}}{h_{(2)} x}}{D^{-1} \eval{\psi}{1}}
        \\
        &=
        D \eval{\psi}{h_{(1)} S(x)}
            \eval{\hat{h}}{h_{(2)}}
        =
        \eval{\psi}{S(x)}
        =
            \eval{\psi}{x},
    \end{align*}
	where we applied \cref{eq:DualInt} to the denominator after bringing the $\hat{h}$ inside in the third equality, used $\eval{\psi}{1}=1$ and applied the second pulling-through property of \cref{eq:PT} in the fourth equality and used \cref{eq:DualInt} again in the fifth equality, and finally concluded using the fact that $\psi$ is invariant under the antipode.
\end{proof}

Moreover, we can see how the trace-preserving condition now translates for these group-like projections.
\begin{remark}
\label{prop:TPgrouplike}
Let $A$ be a $C^*$-Hopf algebra, let $B=A\otimes A^{\mathrm{op}}$ be the enveloping $C^*$-Hopf algebra of $A$, let $\phi\in B^*$ be an idempotent unital linear functional and let $p \in B$ be the corresponding group-like projection. Then, the following conditions are equivalent:
\begin{enumerate}
    \item[(i)] $\eval{\phi}{(\placeholder)\otimes 1}= \hat{h}$;
    \item[(ii)] $\eval{\hat{h}\otimes \hat{h}}{p ((\placeholder) \otimes 1)}=\eval{\hat{h} \otimes \hat{h}}{p} \hat{h}$.
\end{enumerate}
\end{remark}

\subsection{From group-like projections to finite \texorpdfstring{$\boldsymbol{*}$}{*}-quantum hypergroups}

Finally, let us introduce the notion of finite $*$-quantum hypergroup in the sense of Landstad and Van Daele \cite[Definition~3.11]{landstad_finite_2025}.

\begin{definition}
    \label{def:qHG}
    A \newterm{finite $*$-quantum hypergroup} is a finite-dimensional $C^*$-al\-ge\-bra $Q$ over the complex numbers equipped with the structure of a coalgebra, for which the co\-multi\-pli\-ca\-tion is a positive unital $*$-map (not necessarily multiplicative), i.e., for all $x,y\in Q$,
	\begin{gather}
        \label{eq:DeltaHypergroup}
        1_{(1)}\otimes 1_{(2)}
        =
        1\otimes 1,
        \quad
        \text{and}
        \quad
        (x^*)_{(1)}\otimes (x^*)_{(2)}
        =
        (x_{(1)})^* \otimes (x_{(2)})^*,
    \end{gather}
	the counit, $\varepsilon\in Q^*$, is a $*$-algebra homomorphism, i.e., for all $x,y\in Q$,
	\begin{equation}
        \label{eq:axiom_eps_Hypergroup}
        \eval{\varepsilon}{xy} = \eval{\varepsilon}{x}\eval{\varepsilon}{y},
        \quad
        \eval{\varepsilon}{1}= 1,
        \quad 
        \text{and}
        \quad
        \eval{\varepsilon}{x^*}
        =
        \overline{\eval{\varepsilon}{x}},
    \end{equation}
    and there exists a faithful linear functional $\varphi\in Q^*$ satisfying, for all $x\in Q$,
    \begin{equation}
    x_{(1)} \eval{\varphi}{x_{(2)}} = \eval{\varphi}{x} 1,
    \end{equation}
    and defines an antimultiplicative linear map $S_\varphi:Q\to Q$ such that, for all $x,y\in Q$,
    \begin{equation}
        S_\varphi(x_{(1)}) \eval{\varphi}{x_{(2)} y}
        =
        y_{(1)} \eval{\varphi}{x y_{(2)}}.
    \end{equation}
\end{definition}

\begin{remark}
In any finite $*$-quantum hypergroup $Q$, the map $S_\varphi$ satisfies:
\begin{itemize}
    \item[(i)] $S_\varphi: Q \to Q$ is unique.
    \item[(ii)] $S_\varphi: Q \to Q$ is a one-to-one correspondence;
    \item[(iii)] $S_\varphi(1) = 1$ and $\varepsilon \circ S_\varphi = \varepsilon$;
    \item[(iv)] $S_\varphi(x)_{(1)}\otimes S_\varphi(x)_{(2)}= S_\varphi(x_{(2)})\otimes S_\varphi(x_{(1)})$ for all $x\in Q$;
    \item[(v)] $\psi \coloneqq \varphi\circ S_\varphi$ satisfies $\eval{\psi}{x_{(1)}} x_{(2)} = \eval{\psi}{x}1$ for all $x\in Q$;
\end{itemize}
See Propositions~2.6 and~2.8 in Ref.~\cite{landstad_finite_2025}.
\end{remark}

\begin{remark}
    Every $C^*$-Hopf algebra is a finite $*$-quantum hypergroup, where $\varphi \coloneqq \hat{h}$.
\end{remark}

\begin{proposition}
\label{prop:qHG_1}
Let $B$ be a $C^*$-Hopf algebra with comultiplication $\Delta$ and let $p \in B$ be a group-like projection. Then, the corner $*$-subalgebra $c_p(B) = pBp$, equipped with
    \[
        \Delta_p
            \coloneqq
            (p \otimes p)\Delta(\placeholder)(p \otimes p),
    \]
    is a finite $*$-quantum hypergroup.
\end{proposition}

\begin{proof}
     See Ref.~\cite[Theorem~2.7]{landstad_compact_2007}.
\end{proof}

\begin{proposition}
\label{prop:qHG_2}
Let $B$ be a $C^*$-Hopf algebra with comultiplication $\Delta$ and let $\phi \in B^*$ be an idempotent state. Then, the image $*$-subalgebra $c_\phi(B) = (\phi\otimes \mathrm{Id}\otimes \phi)\circ \Delta^{(2)}(B)$, equipped with the comultiplication defined for all $x\in B$ by the expression
\[
    \Delta_\phi(x)
    \coloneqq
    x_{(1)}\otimes \eval{\phi}{x_{(2)}} x_{(3)},
\]
is a finite $*$-quantum hypergroup.
\end{proposition}

\begin{proof}
    See Ref.~\cite[Theorem~3.17]{landstad_compact_2007}.
\end{proof}

In fact, we will need a slightly stronger result in \cref{sec:quantumgoursat}, by combining the previous two statements, we obtain the following result.
\begin{proposition}
\label{prop:qHG_3}
Let $B$ be a $C^*$-Hopf algebra with comultiplication $\Delta$, let $\phi\in B^*$ be an idempotent state and let $q\in B$ be a group-like projection such that, for all $x\in B$,
\begin{equation}\label{eq:prop_q_phi}
\eval{\phi}{xq} = \eval{\phi}{x}
,\quad\text{and}\quad
c_\phi(x) q
=
q c_\phi(x),
\end{equation}
i.e.,
$q\in\operatorname{Center} c_\phi(B)$. 
Then, the $*$-sub\-al\-ge\-bra $Q \coloneqq qc_\phi(B)q$ equipped with
\begin{equation}
    \Delta_Q(x)
        \coloneqq
        x_{(1)} q 
            \otimes \eval{\phi}{x_{(2)}} x_{(3)} q,
\end{equation}
and $1_Q \coloneqq q$, $\varepsilon_Q \coloneqq \varepsilon|_Q = \phi|_Q$ and $\varphi \coloneqq  \hat{h}|_Q / \eval{\hat{h}}{q}$ is a finite $*$-quantum hypergroup.
\smallbreak\noindent%
Moreover, $\Pi_Q:= c_q \circ c_\phi = c_\phi \circ c_q$ is a projection onto $Q$ satisfying
\begin{equation} \label{eq:intertwine_Q}
        (\Pi_Q \otimes \Pi_Q) \circ \Delta = \Delta_Q \circ \Pi_Q.
    \end{equation}
\end{proposition}
\begin{proof}
By construction, $Q$ is a unital $*$-subalgebra of $B$.
With the notation of \cref{def:qHG}, $\varphi$ is the faithful linear functional of $Q$ and $S_{\varphi} = S|_{\smash{Q}}$.
\begin{step-hypergroup}
Let $q\in B$ be a group-like projection. We prove that the following are equivalent:
\begin{enumerate}
    \item[(i)] $\eval{\phi}{qx} = \eval{\phi}{x}$ for all $x\in B$,
    \item[(ii)] $\eval{\phi}{xq} = \eval{\phi}{x}$ for all $x\in B$,
    \item[(iii)] $q\in c_\phi(B)$
    \item[(iv)] $\phi\in c_q(B^*)$, i.e., $\eval{\phi}{qxq} = \eval{\phi}{x}$ for all $x\in B$.
\end{enumerate}
First, note that (i) implies (ii) since, for all $x\in B$,
\[
    \eval{\phi}{xq}
    = \eval{\phi}{S(xq)}
    = \eval{\phi}{S(q)S(x)}
    = \eval{\phi}{qS(x)}
    = \eval{\phi}{S(x)} = \eval{\phi}{x},
\]
where we used that $\phi = \phi \circ S$ by \cref{prop:IdempProp} in the first and last equalities, that $S$ is antimultiplicative in the second, that $S(q) = q$ by \cref{rem:basicPropGroupLike} in the third, and (i) in the fourth. That (ii) implies (i) is proven analogously.
\smallbreak\noindent%
Second, let us prove that (ii) implies (iii). Note that the condition $q\in c_\phi(B)$ is equivalent to the fact that $q = c_\phi(q)$ since $c_\phi\circ c_\phi = c_\phi$. Therefore,
\begin{align*}
    c_\phi(q)
    &=
    \eval{\phi}{q_{(1)}} q_{(2)} \eval{\phi}{q_{(3)}}
    = 
    \eval{\phi}{q_{(1)}q} q_{(2)} \eval{\phi}{q_{(3)}q}
    = 
    \eval{\phi}{q} q \eval{\phi}{q}
    = 
    \eval{\phi}{1}^2 q = q,
\end{align*}
where we used the definition of $c_\phi$ in the first equality, (ii) to insert $q$ in both slots in the second, the group-like identity $q\otimes q\otimes q = q_{(1)}q\otimes q_{(2)}\otimes q_{(3)}q$ of \cref{rem:grouplike3} in the third, that $\eval{\phi}{q} = \eval{\phi}{1}$ in the fourth, and the fact that $\eval{\phi}{1}=1$ in the last.
\smallbreak\noindent%
Third, let us prove that (iii) implies (i). Indeed, for all $x\in B$,
\[
    \eval{\phi}{qx} 
    = \eval{\phi}{q_{(1)}}\eval{\phi}{q_{(2)}x} \eval{\phi}{q_{(3)}}
    = \eval{\phi}{q_{(1)}}\eval{\phi}{x} \eval{\phi}{q_{(2)}}
    = \eval{\varepsilon}{c_\phi(q)}\eval{\phi}{x}
    = \eval{\phi}{x},
\]
where we used that $q = c_\phi(q)$ by hypothesis in the first equality, \cref{prop:IdempProp} with $v\equiv q_{(1)}$ and $w\equiv x$ in the second, that $\eval{\varepsilon}{c_\phi(q)} = \eval{\phi}{q_{(1)}}\eval{\phi}{q_{(2)}}$ in the third, and $\eval{\varepsilon}{c_{\phi}(q)}=\eval{\varepsilon}{q}=1$ by the hypothesis and \cref{rem:basicPropGroupLike} in the last.
\smallbreak\noindent%
Fourth, let us prove that (iv) implies (ii). For all $x\in B$, 
\[
    \eval{\phi}{xq} = \eval{\phi}{q(xq)q} = \eval{\phi}{qxq} = \eval{\phi}{x}.
\]
where we applied the hypothesis (iv) for $xq$ in the first equality, that $q^2 = q$ in the second, and again (iv) for $x$ in the third. Finally, that (i), and thus (ii), implies (iv) is trivial. 
\end{step-hypergroup}
\begin{step-hypergroup}\label{stepHypergroupPi}
    Set
    $\Pi_Q := c_q \circ c_\phi$. Since $q \in c_\phi(B)$, the maps $c_\phi$ and $c_q$ commute:
    for all $x\in B$,
    \begin{align*}
        c_q(c_\phi(x))
        &=
        \eval{\phi}{x_{(1)}}  q x_{(2)} q \eval{\phi}{x_{(3)}}
        \\
        &=
        \eval{\phi}{q x_{(1)} q} q x_{(2)} q \eval{\phi}{q x_{(3)} q}
        \\
        &=
        \eval{\phi}{q (qxq)_{(1)} q} (qxq)_{(2)} \eval{\phi}{q (qxq)_{(3)} q}
        \\
        &=
        \eval{\phi}{(q q_{(1)}) x_{(1)} (q_{(1')} q)} q_{(2)} x_{(2)} q_{(2')} \eval{\phi}{(q q_{(3)})x_{(3)} (q_{(3')} q)}
        \\
        &=
        \eval{\phi}{q x_{(1)} q} q x_{(2)} q \eval{\phi}{q x_{(3)} q}
        \\
        &=
        c_\phi(c_q(x)),
    \end{align*}
    using \cref{eq:prop_q_phi} in the second equality, \cref{rem:grouplike3} in the third,
    and \cref{eq:prop_q_phi} again in the last. In particular $\Pi_Q$ is idempotent.
\end{step-hypergroup}
\begin{step-hypergroup}
    Let us prove that $S(Q) = Q$.
    Note that $S(c_\phi(B))=c_\phi(B)$ since, for all $x\in B$,
    \[
        S ( c_\phi(x) )
        =
        \eval{\phi}{x_{(1)}} \eval{\phi}{x_{(3)}} S(x_{(2)})
        =
        \eval{\phi}{S(x)_{(3)}} \eval{\phi}{S(x)_{(1)}} S(x)_{(2)}
        =
        c_\phi(S(x)),
    \]
using that $\phi\circ S=\phi$ by \cref{prop:IdempProp}, and the anticomultiplicativity of the $S$; moreover,
    \[
        S(q  c_\phi(x) q )
        =
        S(q) S( c_\phi(x) ) S(q)
        =
        q c_\phi( S(x) ) q
        \in Q,
    \]
where we have used that $S(q)=q$ and that $q \in \operatorname{Center}c_\phi(B)$ by assumption.
\end{step-hypergroup}
\begin{step-hypergroup}
    Note that, for all $x\in Q$, one can rewrite
    \begin{equation}\label{eq:DeltaQsimpl}
        \Delta_Q(x)
        =
        x_{(1)} q \otimes \eval{\phi}{x_{(2)}} x_{(3)}
        =
        x_{(1)} \otimes \eval{\phi}{x_{(2)}} x_{(3)} q.
    \end{equation}
    Indeed, the first equality follows from the following calculation, valid for all $x\in Q$,
    \begin{align*}
        \Delta_Q(x) 
        & =
        x_{(1)} q \otimes \eval{ \phi }{ x_{(2)} } x_{(3)} q
        \\ &
        =
        x_{(1)} q \otimes \eval{ \phi }{ x_{(2)}q } x_{(3)} q
        \\ &
        =
        x_{(1)} q_{(1)} q \otimes
            \eval{ \phi }{ x_{(2)} q_{(2)} q } x_{(3)} q_{(3)}
        \\ &
        =
        x_{(1)} q \otimes \eval{ \phi }{ x_{(2)} q } x_{(3)}
        \\ &
        =
        x_{(1)} q \otimes \eval{ \phi }{ x_{(2)} } x_{(3)},
    \end{align*}
    where we recalled the definition of $\Delta_Q$ in the first equality, we used \cref{eq:prop_q_phi} in the second, \cref{rem:grouplike3} in the third, that $xq = x$ since $x\in Q$ and hence $x_{(1)}q_{(1)} \otimes x_{(2)}q_{(2)}\otimes x_{(3)} q_{(3)} = x_{(1)}\otimes x_{(2)}\otimes x_{(3)}$ in the fourth, and \cref{eq:prop_q_phi} again in the last.
    The proof of the second equality is similar. By applying this property twice, for all $x\in Q$,
    \begin{align*}
        (\mathrm{Id} \otimes \Delta_Q) \circ \Delta_Q(x)
        &=
        x_{(1)} q
            \otimes
            \eval{ \phi }{ x_{(2)} } \Delta_Q(x_{(3)})
        \\
        &=
        x_{(1)} q
            \otimes
            \eval{ \phi }{ x_{(2)} } x_{(3)} q
            \otimes
            \eval{ \phi }{ x_{(4)} } x_{(5)},
    \intertext{while, on the other hand,}
        (\Delta_Q\otimes\mathrm{Id})\circ\Delta_Q(x)
        & =
         \Delta_Q( x_{(1)} q )
            \otimes \eval{ \phi }{ x_{(2)} } x_{(3)}
        \\ &
        =
            x_{(1)} q_{(1)} q
            \otimes \eval{ \phi }{ x_{(2)} q_{(2)} } x_{(3)} q_{(3)}
            \otimes \eval{ \phi }{ x_{(4)} } x_{(5)}
            \\
        &=
            x_{(1)} q_{(1)} q
            \otimes  \eval{ \phi }{ x_{(2)} q_{(2)} q } x_{(3)} q_{(3)}
            \otimes \eval{ \phi }{ x_{(4)} } x_{(5)}
            \\
        &=
             x_{(1)}q
            \otimes \eval{\phi}{x_{(2)}q} x_{(3)} q
            \otimes \eval{\phi}{x_{(4)}} x_{(5)}
            \\
        &= 
            x_{(1)} q
            \otimes \eval{ \phi }{ x_{(2)} } x_{(3)} q
            \otimes  \eval{ \phi }{ x_{(4)} }  x_{(5)},
    \end{align*}
    where we used \cref{eq:DeltaQsimpl} in the first and second equality, the first property of \cref{eq:prop_q_phi} in the third, \cref{rem:grouplike3} in the fourth and \cref{eq:prop_q_phi} again in the last.
    The two expressions coincide, so $\Delta_Q$ is coassociative.
\end{step-hypergroup}
\begin{step-hypergroup}
Take $\varepsilon_Q \coloneqq \varepsilon|_Q$,
which is multiplicative.
Note that $\varepsilon_Q = \phi|_Q$ since
\[
    \eval{\phi}{x_{(1)}}x_{(2)}
    =
    x
    ,\quad\text{and hence}\quad
    \eval{\varepsilon_Q}{x}
    =
    \eval{\phi}{x_{(1)}}\eval{\varepsilon}{x_{(2)}}
    =
    \eval{\phi}{x},
\]
as $x\in c_\phi(B)$ by \cref{stepHypergroupPi} and $\phi$ is idempotent. Then, for all $x\in Q$,
\[
    ( \varepsilon_Q \otimes \mathrm{Id} ) \circ \Delta_Q(x)
    =
        \eval{ \varepsilon }{ x_{(1)} } \eval{ \phi }{ x_{(2)} } x_{(3)}q
    =
        \eval{ \phi }{ x_{(1)} } x_{(2)}q
    =
    xq = x,
\]
where we used \cref{eq:DeltaQsimpl} in the first equality, the counit axiom \cref{eq:axiom_eps_id} in the second, that $\eval{\phi}{x_{(1)}}x_{(2)} = x$ since $x\in c_{\phi}(B)$ in the third, and that $x\in Q = q c_\phi(B) q$ in the last.
\end{step-hypergroup}
\begin{step-hypergroup}
Let us prove that the antipode is anticomultiplicative. For all $x\in Q$,
\begin{align*}
    \Delta_Q( S(x) )
    &=
        S(x)_{(1)} q
            \otimes \eval{ \phi }{ S(x)_{(2)} } S(x)_{(3)} q
    \\
    &=
        S(x_{(3)}) q
            \otimes \eval{ \phi }{S(x_{(2)})} S(x_{(1)}) q
    \\
    &=
        S(x_{(3)}) S(q)
            \otimes \eval{ \phi }{ x_{(2)} } S(x_{(1)}) S(q)
    \\
    &=
        S( x_{(3)} q )
            \otimes \eval{ \phi }{ x_{(2)} } S( x_{(1)} q )
    \\
    &=
        (S\otimes S)(\Delta_Q^{\mathrm{op}}(x)),
\end{align*}
where we used the definition of $\Delta_Q$ in the first equality, the anticomultiplicativity of $S$ in the second, the invariance of $\phi$ and $q$ under the antipode in the third and the centrality property of \cref{eq:prop_q_phi} in the fourth.
\end{step-hypergroup}
\begin{step-hypergroup}
    The state in $Q$ is $\varphi \coloneqq \eval{\hat{h}}{q}^{-1} \hat{h}|_Q$.
    Indeed, for all $x\in Q$,
    \begin{align*}
        (\varphi \otimes \mathrm{Id}) \circ \Delta_Q(x)
        &=
            \eval{\hat{h}}{q}^{-1} \eval{\hat{h}}{x_{(1)}q}
                \eval{\phi}{x_{(2)}q} x_{(3)} q
        \\ &
        =
            \eval{\hat{h}}{q}^{-1} \eval{\hat{h}}{x_{(1)} q_{(1)}}
                \eval{\phi}{x_{(2)}q_{(2)} q} x_{(3)} q_{(3)} q
        \\ &
        =
            \eval{\hat{h}}{q}^{-1} \eval{\hat{h}}{x_{(1)}}
                \eval{\phi}{x_{(2)} q} x_{(3)} q
        \\ &
        =
            \eval{\hat{h}}{q}^{-1} \eval{\hat{h}}{x}
                \eval{\phi}{1} q
        \\ &
        =
        \eval{\varphi}{x} 1_Q,
    \end{align*}
    where we used the definitions and the first property of \cref{eq:prop_q_phi} in the first equality, \cref{rem:grouplike3} in the second equality, the fact that $xq=x\in Q$ and the multiplicativity of $\Delta$ in the third, the integral property of \cref{eq:DualHaar2} $\eval{\hat{h}}{x_{(1)}}\Delta(x_{(2)})=\eval{\hat{h}}{x}\Delta(1)=\eval{\hat{h}}{x} 1\otimes 1$ in the fourth and finally the definitions of $1_Q$ and $\varphi$ in the last. The proof that it is also right invariant is analogous. We remark that the factor $\eval{\hat{h}}{q}^{-1}$ is not relevant here but chosen for the subsequent corollary.
\end{step-hypergroup}
\begin{step-hypergroup}
    We now check the compatibility of $\varphi$ with the antipode. 
    For all $x,y\in Q$, 
    \[
    S_Q((\varphi \otimes \mathrm{Id})(\Delta_Q(x)(y\otimes 1)))
    =
    (\varphi \otimes \mathrm{Id})((x\otimes 1)\Delta_Q(y)).
    \]
    On the one hand, the right-hand side takes the form, for all $x,y\in Q$,
    \begin{align*}
        (\varphi \otimes \mathrm{Id})((x\otimes 1)\Delta_Q(y))
        &=
            \eval{\hat{h}}{q}^{-1}\eval{\hat{h}}{xy_{(1)}q}
                \eval{\phi}{y_{(2)}} y_{(3)}
        \\
        &=
            \eval{\hat{h}}{q}^{-1}\eval{\hat{h}}{xy_{(1)}}
                \eval{\phi}{S(y_{(2)})} y_{(3)}
        \\
        &=
            \eval{\hat{h}}{q}^{-1}\eval{\hat{h}}{x_{(1)}y_{(1)}}
                \eval{\phi}{x_{(2)}y_{(2)}S(y_{(3)})} y_{(4)}
        \\
        &=
            \eval{\hat{h}}{q}^{-1}\eval{\hat{h}}{x_{(1)}y_{(1)}}
                \eval{\phi}{x_{(2)}} y_{(2)}
        \\
        &=
            \eval{\hat{h}}{q}^{-1}\eval{\hat{h}}{xy_{(1)}} y_{(2)},
    \end{align*}
    where we used the definition of $\varphi$ together with \cref{eq:DeltaQsimpl} in the first equality, the trace-like property of $\hat{h}$ from \cref{eq:DualHaar} together with the fact that $qx=x\in Q$ and the invariance $\phi \circ S= \phi$ in the second, the integral property \cref{eq:DualHaar2} $\eval{\hat{h}}{z}1=\eval{\hat{h}}{z_{(1)}}z_{(2)}$ with $z\equiv xy_{(1)}$ in the third, \cref{eq:axiomSeps} together with \cref{eq:axiom_eps_id} in the fourth and the fact that $\eval{\phi}{x_{(2)}}x_{(1)}=x$ in the last. Analogously,
    \begin{align*}
        S_Q((\varphi \otimes \mathrm{Id})(\Delta_Q(x)(y\otimes 1)))
        &=
            \eval{\hat{h}}{q}^{-1}
                \eval{\hat{h}}{x_{(1)}qy}
                    \eval{\phi}{x_{(2)}}
                        S(x_{(3)})\\
        &=
            \eval{\hat{h}}{q}^{-1}
                \eval{\hat{h}}{yx_{(1)}}
                    \eval{\phi}{S(x_{(2)})}
                        S(x_{(3)})
        \\
        &=
            \eval{\hat{h}}{q}^{-1}
                \eval{\hat{h}}{y_{(1)}x_{(1)}}
                    \eval{\phi}{y_{(2)}x_{(2)}S(x_{(3)})}
                        S(x_{(4)})
        \\ &
        = 
            \eval{\hat{h}}{q}^{-1}
            \eval{\hat{h}}{x_{(1)}y_{(1)}}
                    \eval{\phi}{y_{(2)}}
                        S(x_{(2)})
        \\ &
        =
            \eval{\hat{h}}{q}^{-1}
            \eval{\hat{h}}{x_{(1)}y}
                S(x_{(2)}).
    \end{align*}
    Then, both expressions are equal due to \cref{eq:propHaarExchXY}.
\end{step-hypergroup}
\noindent%
This concludes the proof of \cref{prop:qHG_3}.
\end{proof}

\begin{corollary} \label{prop:qHG_3_coro}
    If we assume, in addition, that for all $x\in B$,
    \begin{equation} \label{eq:prop_q_phi2}
        L_\phi(x)q=R_\phi(x)q
        \quad\text{or, equivalently,}\quad
        \phi \in \operatorname{Center} c_q(B^*),
    \end{equation}
    then $\Delta_Q$ is multiplicative and $Q$ is a $C^*$-Hopf algebra.
\end{corollary}

\begin{proof}
    The two conditions of \cref{eq:prop_q_phi2} are trivially equivalent.
    \smallbreak\noindent%
    Note that, for all $x\in Q$, we can rewrite
    \begin{equation}\label{eq:DeltaQsimpl2}
        \Delta_Q(x)
        =
        \eval{\phi}{x_{(1)}} x_{(2)}q \otimes x_{(3)} q.
    \end{equation}
    We have, for all $x\in Q$,
    \[
        \Delta_Q(x) 
        =
        R_\phi(x_{(1)})q \otimes x_{(2)}q
        =
        L_\phi(x_{(1)})q \otimes x_{(2)}q
        =
        \eval{\phi}{x_{(1)}} x_{(2)}q \otimes x_{(3)}q.
    \]
    where we used the definition of $\Delta_Q$ and $R_{\phi}$ in the first equality, \cref{eq:prop_q_phi2} in the second and the definition of $L_{\phi}$ in the last.
    We can now prove the multiplicativity of $\Delta_Q$. Indeed, for all $x, y \in Q$,
    \begin{align*}
        \Delta_{Q}(x) \Delta_{Q}(y)
        &=
            \eval{\phi}{x_{(1)}}
                \eval{\phi}{y_{(1)}}
                    x_{(2)}q y_{(2)}q \otimes x_{(3)} q y_{(3)}q
        \\
        &=
            \eval{\phi}{x_{(1)}}
                \eval{\phi}{y_{(1)}}
                    x_{(2)} y_{(2)}q \otimes x_{(3)} y_{(3)}q
        \\
        &=
            \eval{\phi}{x_{(1)}}
                \eval{\phi}{x_{(2)}y_{(1)}}
                    x_{(3)} y_{(2)}q \otimes x_{(4)} y_{(3)}q
        \\
        &=
        \eval{\phi}{L_{\phi}(x)_{(1)}y_{(1)}}
            L_{\phi}(x)_{(2)} y_{(2)}q \otimes L_{\phi}(x)_{(3)} y_{(3)}q
        \\
        &=
        \eval{\phi}{x_{(1)}y_{(1)}}
            x_{(2)} y_{(2)}q \otimes x_{(3)} y_{(3)}q
        \\
        &=
        \Delta_{Q}(xy),
    \end{align*}
    where we used \cref{eq:DeltaQsimpl2} in the first equality, the fact that $(q\otimes q)\Delta_Q(y)=\Delta_Q(y)$ in the second, \cref{prop:IdempProp} with $v\equiv x_{(1)}$ and $w\equiv y_{(1)}$ in the third, the definition of $L_{\phi}$ in the fourth, the fact that $L_{\phi}(x)=x$ in $Q$ in the fifth and the definition of $\Delta_{Q}$ in the last.
\end{proof}

\section{Classification in the finite-group case}
\label{sec:groups}

We now specialize the classification of \cref{sec:fixedpoints} to group algebras and their duals, where idempotent states correspond to subgroups and Goursat's Lemma yields a complete and explicit classification.

\subsection{Group-like projections in (co)commutative {\itshape{}C*}-Hopf algebras}

We now determine the general form of group-like projections in $\mathbb{C}G$ and $\mathbb{C}^G$.

\begin{proposition}
\label{prop:grouplikeProjsInCG}
	Let $G$ be a finite group.
    Then group-like projections of $\mathbb{C}G$ and $\mathbb{C}^G$ are, respectively, of the form
	\[
		p_H \coloneqq
            \frac{1}{|H|} \sum_{g \in H} g
        ,\quad\text{and}\quad
        \hat{p}_H \coloneqq
            \sum_{g \in H} \delta_g,
	\]
	where $H$ is a subgroup of $G$.
    Note that
    $\eval{\hat{h}_{\mathbb{C}G}}{p_H} = |H|^{-1}$ and $\eval{\hat{h}_{\mathbb{C}^G}}{\hat{p}_H} = |H||G|^{-1}$.
\end{proposition}

\begin{proof}
Consider the group algebra $\mathbb{C}G$ and write
\[
    p = \sum_{g\in G} a_g g, \qquad a_g  \coloneqq  \delta_g(p),
\]
for a non-zero self-adjoint projection satisfying $\Delta(p)(1\otimes p) = p\otimes p$. Then,
\[
    \sum_{g,h\in G} a_g a_h (g \otimes gh)
    =
    \sum_{g,h\in G} a_g a_h (g \otimes h),
\]
and by comparing coefficients of the first tensor factor, this proves that
\[
    a_g \big( \sum_{h\in G} a_h gh - \sum_{h\in G} a_h h \big)=0,
\]
for each $g\in G$. Let $H  \coloneqq  \{ g\in G : a_g \neq 0 \} \neq \emptyset$.
Then, for every $g\in H$,
\[
    \sum_{h\in G} a_h gh = \sum_{h\in G} a_h h
    \quad\text{ and hence }\quad
    \sum_{h\in G} a_{g^{-1}h} h = \sum_{h\in G} a_h h
\]
by reindexing the left-hand side by $h \mapsto g^{-1}h$. Therefore,
for all $h\in G$ and $g\in H$,
\begin{equation}
\label{eq:PropReordCoeffGlike}
    a_{g^{-1}h} = a_h.
\end{equation}
Let us prove that $H$ is a subgroup of $G$.
First, taking $h = g$ in the previous expression yields $a_e = a_g \neq 0$ for all $g\in H$, and thus $e\in H$ and all non-zero coefficients coincide.
In other words, there exists $c\neq 0$ such that $a_g = c \delta_{g\in H}$.
Second, for all $g,h\in H$,
\[
    a_{gh}
    =
    a_{g^{-1}(gh)}
    =
    a_h
    \neq 0,
\]
by virtue of \cref{eq:PropReordCoeffGlike}, so $gh \in H$.
Now, since $p$ is self-adjoint by assumption,
\[
    p^*
    =
    \sum_{g\in G} \overline{a_g} g^*
    = 
    \sum_{g\in G} \overline{a_g} g^{-1}
    \overset{!}=
    \sum_{g\in G} a_{\smash{g^{-1}}} g^{-1}
    =
    p,
\]
thus $a_{\smash{g^{-1}}} = \overline{a_g}$ and
hence for all $g\in H$ it trivially holds that $g^{-1}\in H$.
Furthermore,
\[
    p^2
    =
    c^2 \sum_{g,h\in H} gh
    =
    c^2 |H| \sum_{k\in H} k
    =
    c|H| p,
\]
and since $p^2=p\neq 0$, we conclude that $c=|H|^{-1}$, as we wanted to prove.
\smallbreak\noindent%
In the case of $\mathbb{C}^G$,
any element $\hat{p} = \sum_{g \in G} a_g \delta_g $ is a projection, i.e.,
\[
    \hat{p}^2
    =
    \sum_{g,h\in G} a_g a_h \delta_g\delta_h
    =
    \sum_{g,h\in G} \delta_{g,h} a_g a_h \delta_{g}
    =
    \sum_{g\in G} a_g^2 \delta_{g}
    =
    \sum_{g\in G} a_g \delta_{g}
    =
    \hat{p},
\]
if and only if $a_g \in \{0, 1\} $.
Let us define $ H  \coloneqq  \{ g \in G : a_g = 1 \} $.
Then, it follows that
\[
	\hat\Delta(\hat{p})(1 \otimes \hat{p})
    =
    \sum_{g, h \in H} \delta_{g h^{-1}} \otimes \delta_h,
    \quad
    \hat{p} \otimes \hat{p}
    =
    \sum_{g, h \in H} \delta_g \otimes \delta_h.
\]
See \cref{examples:CG}.
Matching terms gives $ gh^{-1} \in H $ for all $ g, h \in H $,
and $\hat{S}(\hat{p}) = \hat{p} $ implies $ g^{-1} \in H $,
so $ H $ is a subgroup.
\smallbreak\noindent%
The converse implication is straightforward.
\end{proof}

\subsection{Back to idempotent states}

Let us determine the corresponding idempotent states $\phi_p$ for each group-like projection found above. 

\begin{proposition}
\label{prop:idempotentStatesInCG}
Let $G$ be a finite group.
Idempotent states of $\mathbb{C}G$ and $\mathbb{C}^G$ are of the form
\[
	\phi_H  \coloneqq  \phi_{p_H} = \sum_{g \in H} \delta_g
    ,\quad\text{and}\quad
	\hat{\phi}_H  \coloneqq  \phi_{\hat{p}_H} = \frac{1}{|H|}\sum_{g \in H} g,
\]
respectively, where $H$ is a subgroup of $G$, and we identified $(\mathbb{C}^G)^*\cong \mathbb{C}G$.
\end{proposition}

\begin{proof}
    It follows easily from \cref{prop:grouplikeProjsInCG,prop:phip}.
\end{proof}

Let us note that, in the case of $\mathbb{C}^G$, this recovers the classical Kawada--It\^o theorem \cite{kawada_probability_1940}, which states that idempotent distributions on a finite group $G$ correspond to uniform measures on the subgroups of $G$. See also Ref.~\cite{franz_idempotent_2013}.
Now, using the identifications
\begin{align*}
    \mathbb{C}G \otimes (\mathbb{C}G)^{\mathrm{op}}
    & \overset{\mathrm{Id}\otimes S}{\cong} \mathbb{C}G \otimes \mathbb{C}G
    \cong \mathbb{C}[G \times G],
    \\
    \mathbb{C}^G \otimes (\mathbb{C}^G)^{\mathrm{op}}
    & \overset{\hphantom{\mathrm{Id}\otimes S}}{\cong} \mathbb{C}^G \otimes \mathbb{C}^G
    \cong \mathbb{C}^{G \times G},
\end{align*}
we can now describe the general form of idempotent states on
$\mathbb{C}[G\times G]$ and
$\mathbb{C}^{G\times G}$.
To obtain a more refined description of these states and of the corresponding subgroups of $G \times G$, we recall Goursat's Lemma \cite{goursat_sur_1889}.

\begin{goursat-lemma-num}
\label{lemma:Goursat}
Let $G_1$ and $G_2$ be two groups. Let $H$ be a subgroup of $G_1 \times G_2$, let $\pi_i:H\to G_i$, $i=1,2$, stand for the canonical projections, define the subgroups
\[
    H_1 \coloneqq \pi_1(H),
    \quad\text{and}\quad
    H_2 \coloneqq \pi_2(H), 
\]
of $G_1$ and $G_2$, respectively, and define the normal subgroups
\begin{equation}
    N_1 \coloneqq  \{ g_1 \in H_1 : (g_1,e) \in H \},
    \quad\text{and}\quad
    N_2 \coloneqq  \{ g_2 \in H_2 : (e,g_2) \in H \},
\end{equation}
of $H_1$ and $H_2$, respectively.
Then, there exists a unique group isomorphism
\[
    \lambda : H_1 / N_1 \to H_2 / N_2
\]
such that
\[
H = \{ (g_1,g_2) \in H_1 \times H_2 : 
       \lambda(g_1 N_1) = g_2 N_2 \}.
\]
Conversely,
for $i=1,2$,
let $H_i$ be a subgroup of $G_i$,
$N_i$ be a normal sub\-group of $H_i$,
and
$\lambda : H_1 / N_1 \to H_2 / N_2$ be a group isomorphism, then
\[
\{ (g_1,g_2) \in H_1 \times H_2 : \lambda(g_1 N_1) = g_2 N_2 \}
\]
is a subgroup of $G_1 \times G_2$.
\end{goursat-lemma-num}

We refer the reader to Ref.~\cite{dung_subgroups_2009} for an elementary proof. Together with \cref{prop:idempotentStatesInCG}, this provides further structure on the new renormalization fixed point tensors.
\begin{theorem}\label{thm:MinftyGroups}
    Let $G$ be a finite group. With the notation of \cref{lemma:Goursat}, the new renormalization fixed points obtained in \cref{sec:fixedpoints}, in the case of the $C^*$-Hopf algebras $\mathbb{C}G$ and $\mathbb{C}^G$, are the following:
    In the case of $\mathbb{C}G$, it holds that $N_1 = \{e\}$ and
	\begin{equation}\label{eq:MinftyCG}
        M_{\infty}
        =
       \frac{1}{|G|}\sum_{g\in H_1}
            \eval{\delta_g}{\placeholder} g \otimes
            \sum_{k\in \lambda(g)} \mathcal{L}(k).
	\end{equation}
    In the case of $\mathbb{C}^G$, it holds that $H_1 = G$ and
	\begin{equation}\label{eq:MinftyCGdual}
        M_{\infty}
        =
            \frac{1}{|H|}\sum_{g\in G}
                \mathcal{L}(g) \otimes 
                \sum_{k\in \lambda(gN_1)}
                    \eval{\delta_k}{\placeholder} k
                    .
	\end{equation}
\end{theorem}

\begin{proof}
    We consider the notation of \cref{examples:CG}.
    First, note that any idempotent state $\phi \in (\mathbb{C}G \otimes (\mathbb{C}G)^{\mathrm{op}})^*$ corresponds to an idempotent state $\phi'\coloneqq \phi \circ (\mathrm{Id} \otimes S) \in (\mathbb{C}G \otimes \mathbb{C}G)^*$ and therefore to a subgroup $H$ of $G\times G$.
    Consider, by \cref{lemma:Goursat}, the associated quintuple $(H_1, H_2, N_1, N_2, \lambda)$, from which it follows by \cref{prop:idempotentStatesInCG} that
    \[
        \phi = \phi' \circ (\mathrm{Id} \otimes S)
        =
        \sum_{C\in H_1/N_1}
            \big( \sum_{g\in C} \delta_g \big)
                \otimes
                    \big( \sum_{k\in \lambda(C)} \delta_{k^{-1}} \big)
        =
            \sum_{g\in H_1} \delta_g \otimes \big(\sum_{k\in \lambda(gN_1)} \delta_{k^{-1}}\big).
    \]
    Since we assume that $\eval{\phi}{g \otimes 1}=\eval{\hat{h}_{\mathbb{C}G}}{g}$ for all $g\in G$, 
    \[
        \eval{\phi}{g \otimes 1}=\delta_{g\in H_1} \delta_{e\in \lambda(gN_1)}=\delta_e,
    \]
    which implies that $N_1=\{e\}$.
    We can now use \cref{th:NewFixedPoints} to obtain that
    \begin{align*}
		M_{\smash{\infty}}
        &=
        \sum_{g\in G} |G| \eval{\phi}{g\otimes S(h_{\mathbb{C}G\,(1)})}  \Psi(\delta_g) \otimes \hat{b}(\hat{h}_{\mathbb{C}G}) \mathcal{L}(h_{\mathbb{C}G\,(2)})
		\\
        & =
            \sum_{g\in G} \sum_{k\in G}  \eval{\phi}{g\otimes k^{-1}}  \Psi(\delta_g) \otimes \frac{1}{|G|} \mathcal{L}(k)
        \\
        &=
        \sum_{g\in H_1} \sum_{k\in \lambda(gN_1)}  \Psi(\delta_g) \otimes \frac{1}{|G|} \mathcal{L}(k),
	\end{align*}
    which proves \cref{eq:MinftyCG}.
    Second, by \cref{prop:idempotentStatesInCG,lemma:Goursat}, any idempotent state $\hat\phi\in (\mathbb{C}^G\otimes \mathbb{C}^G)^*\cong \mathbb{C}G\otimes\mathbb{C}G$ is of the form
    \[
        \hat\phi =
        \frac{1}{|H|} \sum_{C\in H_1/N_1} \big( \sum_{g\in C} g \big) \otimes \big( \sum_{k\in \lambda(C)}  k \big).
    \]
    The condition $\eval{\hat\phi}{\delta_g \otimes 1}=\eval{\hat{h}_{\smash{\mathbb{C}^G}}}{\delta_g}$ now provides
    \[
        \eval{\hat\phi}{\delta_g \otimes 1}
        =
        \frac{1}{|H|}|\lambda(gN_1)|\,\delta_{g\in H_1}
        =
        \frac{|N_2|}{|H|}\delta_{g\in H_1}
        =
        \frac{1}{|H_1|}\delta_{g\in H_1}
        \overset{!}{=}
        \frac{1}{|G|},
    \]
    which holds for all $g\in G$ if and only if $H_1=G$.
    Since $\Delta(h_{\smash{\mathbb{C}^G}}) = \sum_{k\in G} \delta_{\smash{k^{-1}}} \otimes \delta_k$,
    \begin{align*}
		M_{\smash{\infty}}
		    & = \sum_{g\in G} |G| \sum_{k\in G} \eval{\hat\phi}{\delta_g\otimes S(\delta_{k^{-1}})}  \Psi(g) \otimes \frac{1}{|G|} \Phi(\delta_k)\\
            &= \sum_{g\in H_1} \sum_{k\in \lambda(gN_1)}  \Psi(g) \otimes \frac{1}{|H|} \Phi(\delta_k),
	\end{align*}
    as we wanted to prove.
\end{proof}


\begin{remark}\label{rem:RewriteMinftyGroup}
    Note that, in both cases, $M_{\infty}$ can be rewritten as tensors arising from the $C^*$-Hopf algebras $Q\coloneqq \mathbb{C}H_1$ in the group case, and $Q \coloneqq \mathbb{C}^{\smash{G/N_1}}$ in the dual of the group case. In the group case, one can rewrite
    \[
        M_{\infty} = \frac{|N_2|}{|G|} \sum_{g\in H_1} \eval{\delta_g}{(\placeholder)} g \otimes \tilde\Phi(g),
	\]
    where $\tilde\Phi$ is the $*$-representation of $\mathbb{C}H_1$ given, for all $g\in H_1$, by the expression
    \[
        \tilde\Phi(g)
            \coloneqq
                \frac{1}{|N_2|} \sum_{k\in \lambda(g)} \mathcal{L}(k),
    \]
    $\delta_g\mapsto \eval{\delta_g}{(\placeholder)} g$ is a faithful $*$-representation of $(\mathbb{C}H_1)^*$, and note that $\hat{b}(\hat{h}_{Q}) = |N_2|/|G|$.
    Similarly, in the dual case,
    \[
        M_{\infty}
        =
            \frac{1}{|H|} \sum_{C \in G/N_1} \big( \sum_{g \in C} \mathcal{L}(g) \big)
                \otimes  \big( \sum_{\smash{k\in \lambda(C)}} \eval{\delta_k}{(\placeholder)} k \big)
        =
            \frac{|N_1|}{|H|} \sum_{C \in G/N_1} \tilde{\Psi}(C) \otimes  \tilde\Phi(\delta_C),
    \]
    where $\tilde\Phi$ and $\tilde\Psi$ are now the $*$-representations of $\mathbb{C}^{G/N_1}$ and $(\mathbb{C}^{G/N_1})^*\cong \mathbb{C}[G/N_1]$, respectively, defined for all $C\in G/N_1$ by the expressions
    \[
        \tilde\Phi(\delta_C)
            \coloneqq
                \sum_{k\in \lambda(C)} \eval{\delta_k}{(\placeholder)} k
        ,\quad\text{and}\quad
        \tilde\Psi(C)
            \coloneqq
                \frac{1}{|N_1|} \sum_{g\in C} \mathcal{L}(g).
    \]
\end{remark}

\section{A Quantum Goursat Lemma}
\label{sec:quantumgoursat}

\subsection{An extension of Goursat's Lemma}
In the previous section we showed that, for group and dual-group algebras, the fixed points arising from the convergent renormalization flows can be described explicitly through the classical Goursat's Lemma. In this section, we establish a more general structural result for finite-dimensional $C^*$-Hopf algebras.

The guiding idea is that an idempotent state on a tensor product is determined by the data it induces on each tensor factor, together with the residual correlation between them. In this generalized setting, subgroups are replaced by finite $*$-quantum hypergroups, while the quotient data in Goursat's Lemma is encoded by suitable group-like projections and the corner $C^*$-Hopf algebras they determine. The remaining coupling is described by an anti-isomorphism between these corner algebras. Since it is strongly inspired and motivated from a quantum perspective, we dub it a quantum Goursat lemma.

\begin{theorem}
    \label{thm:QuantumGoursat}
    Let $A_1$ and $A_2$ be two $C^*$-Hopf algebras.
    Then, for any idempotent state $\phi\in (A_1 \otimes A_2)^*$, for $i=1,2$, there exists an idempotent state $\phi_i\in A_i^*$, and hence
    \[
        H_i \coloneqq (\phi_i\otimes\mathrm{Id}\otimes\phi_i)(\Delta^{(2)}(A_i))
    \]
    is a finite $*$-quantum hypergroup,
    there exists a group-like projection $q_i\in A_i$ satisfying
    \begin{equation}\label{eq:QGL_qiprop}
        q_i \in \operatorname{Center}H_i, \quad\text{and}\quad
        \phi_i \in \operatorname{Center} A_i^* (q_i(\placeholder)q_i),
    \end{equation}
    which defines a $C^*$-Hopf algebra
    \begin{equation}
        Q_i \coloneqq q_i H_i q_i,
    \end{equation}
    and there exists a linear map $\lambda:A_1 \to A_2$, whose restriction $\lambda|_{\smash{Q_1}}:Q_1 \to Q_2$ is a $C^*$-Hopf al\-ge\-bra anti-isomorphism,
    such that $\phi$ can be rewritten, for all $x\in A_1$ and $y\in A_2$, as 
    \begin{equation}
        \eval{\phi}{x\otimes y}
        =
        \frac{\eval{\hat{h}_{A_2}}{\lambda (x)y}}{\eval{\hat{h}_{A_2}}{\lambda(1)}}.
    \end{equation}
    Conversely, every such quintuple $(\phi_1, \phi_2, q_1, q_2, \lambda)$ determines a unique idempotent state $\phi \in (A_1 \otimes A_2)^*$ where, for $i=1,2$, $\phi_i\in A_i^*$ stands for an idempotent state, $q_i\in A_i$ for a group-like projection of $A_i$ satisfying \cref{eq:QGL_qiprop}, and $\lambda$ is a $C^*$-Hopf algebra anti-isomorphism as above.
\end{theorem}

\begin{proof}
For the sake of clarity, we write $h$ (resp.~$k$) for the Haar integral of $A_2$, (resp.~$A_1$), $\hat{h}$ (resp.~$\hat{k}$) for the Haar integrals of $A_2^*$ (resp.~$A_1^*$), and $D_i \coloneqq \dim A_i$ for $i=1,2$.

\begin{step-goursat}
Assume that $\phi\in (A_1\otimes A_2)^*$ is an idempotent state.
Let $\phi_1\in A_1^*$ and $\phi_2\in A_2^*$ be the linear functionals given, for all $x\in A_1$ and $y\in A_2$, by
\[
    \eval{\phi_1}{x}
        \coloneqq
        \frac{\eval{\phi}{x\otimes h}}{\eval{\phi}{1\otimes h}}
    ,\quad\text{and}\quad
    \eval{ \phi_2 }{ y }
        \coloneqq
        \frac{\eval{\phi}{k\otimes y}}{\eval{\phi}{k\otimes 1}},
\]
and let $\psi_1\in A_1^*$ and $\psi_2\in A_2^*$ be the linear functionals given, for all $x\in A_1$ and $y\in A_2$, by
\[
    \eval{\psi_1}{x}
        \coloneqq
        \eval{\phi}{x\otimes 1}
    ,\quad\text{and}\quad
     \eval{\psi_2}{y}
        \coloneqq
        \eval{\phi}{1\otimes y}.
\]
We first prove that \(\phi_1\) and \(\psi_1\) are idempotent states on \(A_1\). Indeed, for all $x\in A_1$,
\begin{align*}
    \eval{\phi}{1\otimes h}^2
        \eval{\phi_1}{x}
    &=
    \eval{\phi}{x\otimes h}
        \eval{\phi}{1\otimes h}
    \\
    &=
    \eval{\phi}{x_{(1)}\otimes h_{(1)}}
        \eval{\phi}{x_{(2)}\otimes h_{(2)}h}
    \\
    &=
    \eval{\phi}{x_{(1)}\otimes h_{(1)}}
        \eval{\phi}{x_{(2)}\otimes \eval{\varepsilon}{h_{(2)}}h}
    \\
    &=
    \eval{\phi}{x_{(1)}\otimes h}
        \eval{\phi}{x_{(2)}\otimes h}
    \\
    &=
    \eval{\phi}{1\otimes h}^2
        \eval{\phi_1}{x_{(1)}}
            \eval{\phi_1}{x_{(2)}},
\end{align*}
where we used the definition of $\phi_1$ in the first equality, \cref{prop:IdempProp} in the second equality with $v \equiv x\otimes h$ and $w \equiv 1\otimes h$, the integral property of $h$ in the third, \cref{eq:axiom_eps_id} in the fourth and the definition of $\phi_1$ again in the fifth.
Moreover, for all $x\in A_1$,
\begin{equation*}
    \eval{\psi_1}{x}
    =
    \eval{\phi}{x\otimes 1}
    =
    \eval{\phi}{x_{(1)}\otimes 1}
        \eval{\phi}{x_{(2)}\otimes 1}
    =
    \eval{\psi_1}{x_{(1)}}
        \eval{\psi_1}{x_{(2)}}
\end{equation*}
where the second equality follows from the idempotency of $\phi$. Similarly, one can prove that \(\phi_2\) and \(\psi_2\) are idempotent states as well. 
By virtue of \cref{prop:phip}, we let
\begin{equation}
    q_1
    \coloneqq
        \frac{\eval{\phi}{k_{(1)}\otimes 1} k_{(2)}}{\eval{\phi}{k\otimes 1}} \in A_{1}
    ,\quad\text{and}\quad
    q_2
    \coloneqq
        \frac{\eval{\phi}{1\otimes h_{(1)}} h_{(2)}}{\eval{\phi}{1\otimes h}} \in A_{2},
\end{equation}
stand for the group-like projections associated to $\psi_1$ and $\psi_2$, respectively.
\smallbreak\noindent%
    These group-like projections satisfy,
    for all $x\in A_1$, that
    \begin{equation}\label{eq:phi_absorb_q}
        \eval{\phi_1}{q_1x}
        =
        \frac{\eval{\phi}{k_{(1)}\otimes 1}
        \eval{\phi}{k_{(2)} x\otimes h}}{\eval{\phi}{k\otimes 1}\eval{\phi}{1\otimes h}}
        =
        \frac{\eval{\phi}{k\otimes 1}\eval{\phi}{x\otimes h}}{\eval{\phi}{k\otimes 1}\eval{\phi}{1\otimes h}}
        =
        \eval{\phi_1}{x},
    \end{equation}
    where we applied \cref{prop:IdempProp} with $v\equiv k\otimes 1$ and $w\equiv x\otimes h$. Similar arguments can be used to prove $\eval{\phi_1}{xq_1}=\eval{\phi_1}{x}$ and $\eval{\phi_2}{q_2 y}=\eval{\phi_2}{y}=\eval{\phi_2}{y q_2}$ for all $y\in A_2$.
    \end{step-goursat}
    \begin{step-goursat}\label{prop:invariance_q_i_phi_i}
        We prove the following identities for all $x\in A_1$ and $y\in A_2$:
        \begin{equation*}
        \eval{\phi}{x \otimes y}
        =
        \eval{\phi}{L_{\alpha}(x) \otimes y}
        =
        \eval{\phi}{R_{\alpha}(x) \otimes y}
        =
        \eval{\phi}{x \otimes L_{\beta}(y)}
        =
        \eval{\phi}{x \otimes R_{\beta}(y)}
        \end{equation*}
        for $\alpha \equiv q_1,\phi_1$ and $\beta \equiv q_2, \phi_2$.
        For simplicity, we only prove the first equality for both $\alpha\equiv q_1,\phi_1$, the other equalities can be proven the same way.
        First, note that
        \begin{equation*}
            \eval{\phi}{q_1 x\otimes y}= \eval{\phi}{k\otimes 1}^{-1} \eval{\phi}{k_{(1)} \otimes 1} \eval{\phi}{k_{(2)} x\otimes y} =\eval{\phi}{x \otimes y},
        \end{equation*}
        by virtue of \cref{prop:IdempProp} with $v\equiv k\otimes 1$ and $w\equiv x \otimes y$.
        Second, note that
        \begin{align*}
            \eval{\phi}{x \otimes y}
            &=
            \eval{\phi}{1 \otimes h}^{-1}
                \eval{\phi}{1 \otimes h}
                    \eval{\phi}{x \otimes y}
            \\
            &=
            \eval{\phi}{1\otimes h}^{-1}
                \eval{\phi}{x_{(1)}\otimes h y_{(1)}}
                    \eval{\phi}{x_{(2)}\otimes y_{(2)}}
            \\
            &=
            \eval{\phi}{1\otimes h}^{-1}
                \eval{\phi}{x_{(1)}\otimes h}
                    \eval{\phi}{x_{(2)}\otimes y}
            \\
            &=
            \eval{\phi}{L_{\phi_1}(x)\otimes y},
        \end{align*}
        using \cref{cor:IdempPropExch} with $v\equiv x\otimes y$ and $w\equiv 1\otimes h$ in the second equality and the integral property of \cref{eq:Haar2} together with \cref{eq:axiom_eps_id} in the third. 
    \end{step-goursat}
    
    \begin{step-goursat}
    For $i=1,2$, we define 
    \[
        H_i
            \coloneqq
            c_{\phi_i}(A_i),
        \qquad
        \Delta_{\phi_i}
            \coloneqq
            (\mathrm{Id} \otimes \phi_i \otimes \mathrm{Id}) \circ \Delta_{A_i}^{(2)},
    \]
    which, by \cref{prop:qHG_2}, is a finite quantum $*$-hypergroup.
    \smallbreak\noindent%
    Let us prove that $q_1\in \operatorname{Center}H_1$ and $\phi_1 \in \operatorname{Center}c_{q_1}(A_1^*)$. 
    \smallbreak\noindent%
    For that purpose, we show first that, for all $x\in A_1$,
    \begin{align} \label{eq:q_commute_H}
        L_{\phi_1}(x) q_1
        =
        q_1 L_{\phi_1}(x)
        =
        q_1 R_{\phi_1}(x).
    \end{align}
    On the one hand, the first equality follows from the following calculation, for all $x\in A_1$:
    \begin{align*}
        L_{\phi_1}(x) q_1
        &=
            \eval{\phi}{1\otimes h}^{-1}
                \eval{\phi}{k\otimes 1}^{-1}
                    \eval{\phi}{x_{(1)}\otimes h}
                        \eval{\phi}{k_{(1)}\otimes 1}
                            x_{(2)} k_{(2)}
        \\
        &=
            \eval{\phi}{1\otimes h}^{-1}
                \eval{\phi}{k\otimes 1}^{-1}
                    \eval{\phi}{x_{(1)}\otimes h_{(1)}}
                        \eval{\phi}{x_{(2)}
                            k_{(1)} \otimes h_{(2)}} x_{(3)}k_{(2)}
        \\
        &=
            \eval{\phi}{1 \otimes h}^{-1}
                \eval{\phi}{k \otimes 1}^{-1}
                    \eval{\phi}{x \otimes h_{(1)}}
                        \eval{\phi}{k_{(1)} \otimes h_{(2)}}
                            k_{(2)}
        \\
        &= 
            \eval{\phi}{1\otimes h}^{-1}
                \eval{\phi}{k \otimes 1}^{-1}
                    \eval{\phi}{x_{(1)} \otimes h_{(1)}}
                        \eval{\phi}{x_{(2)}S(k_{(1)}) \otimes h_{(2)}S(h_{(3)})}
                            k_{(2)}
        \\
        &= 
            \eval{\phi}{1\otimes h}^{-1}
                \eval{\phi}{k \otimes 1}^{-1}
                    \eval{\phi}{x_{(1)} \otimes h}
                        \eval{\phi}{k_{(1)} \otimes 1}
                            k_{(2)}x_{(2)}
        \\
        &= 
            q_1 L_{\phi_1}(x),
    \end{align*}
    where we used the definition of $\phi_1$ and $q_1$ in the first equality, \cref{prop:IdempProp} with $v\equiv x_{(1)}\otimes h$ and $w\equiv k_{(1)} \otimes 1$ in the second, the fact that $x_{(1)} \otimes \Delta(x_{(2)}k)= x \otimes \Delta(k)$ from the integral property of $k$ in the third, \cref{prop:IdempProp} with $v\equiv x\otimes h_{(1)}$ and $w\equiv k_{(1)} \otimes h_{(2)}$ after using the antipode invariance of $\phi$ in the fourth and \cref{eq:axiomSeps}, and the second pulling-through of \cref{eq:PTmodified} in the fifth.
    On the other hand,
    \begin{align*}
            q_1 R_{\phi_1}(x) 
        &=
            \eval{\phi}{1\otimes h}^{-1}
                \eval{\phi}{k\otimes 1}^{-1}
                    \eval{\phi}{x_{(2)}\otimes h}
                        \eval{\phi}{k_{(1)}\otimes 1}
                            k_{(2)}x_{(1)}
        \\
        &=
            \eval{\phi}{1\otimes h}^{-1}
                \eval{\phi}{k\otimes 1}^{-1}
                    \eval{\phi}{S(x_{(2)})\otimes h}
                        \eval{\phi}{S( k_{(1)})\otimes 1}
                            k_{(2)}x_{(1)}
        \\
        &=
            \eval{\phi}{1\otimes h}^{-1}
                \eval{\phi}{k\otimes 1}^{-1}
                    \eval{\phi}{S(x_{(3)})\otimes h_{(1)}}
                        \eval{\phi}{S(x_{(2)}) S(k_{(1)})\otimes h_{(2)}}
                            k_{(2)} x_{(1)}
        \\
        &=
            \eval{\phi}{1\otimes h}^{-1}
                \eval{\phi}{k\otimes 1}^{-1}
                    \eval{\phi}{S(x_{(3)})\otimes h_{(1)}}
                        \eval{\phi}{S(x_{(2)}) x_{(1)} S(k_{(1)})\otimes h_{(2)}}
                            k_{(2)}
        \\
        &=
            \eval{\phi}{1\otimes h}^{-1}
                \eval{\phi}{k\otimes 1}^{-1}
                    \eval{\phi}{S(x)\otimes h_{(1)}} \eval{\phi}{ S(k_{(1)})\otimes h_{(2)}} k_{(2)}\\
        &=
            \eval{\phi}{1\otimes h}^{-1}
                \eval{\phi}{k\otimes 1}^{-1}
                    \eval{\phi}{x\otimes h_{(1)}}
                        \eval{\phi}{ k_{(1)}\otimes h_{(2)}}
                            k_{(2)}\\
        &= 
        q_1 L_{\phi_1}(x),
    \end{align*}
    using similar steps.
    This concludes the proof of \cref{eq:q_commute_H}.
    Now, for all $x\in A_1$,
    \[
        c_{\phi_1}(x) q_1 = L_{\phi_1}(R_{\phi_1}(x)) q_1 = q_1 L_{\phi_1}(R_{\phi_1}(x))  = q_1  c_{\phi_1}(x),
    \]
    and hence $q_1\in\operatorname{Center} H_1$. Also, for any $f\in c_{q_1}(A_1^*)$, it holds that $f = \eval{\varphi}{q_1(\placeholder)q_1}$ for some $\varphi\in A_1^*$ by definition (in fact, one can choose $\varphi = f$), and thus, for all $x\in A_1$,
    \[
        \eval{\phi_1 f}{x}
        =
        \eval{\varphi}{q_1 L_{\phi_1}(x) q_1}
        =
        \eval{\varphi}{q_1 R_{\phi_1}(x) q_1}
        =
        \eval{f \phi_1}{x},
    \]
    which proves that $\phi_1\in \operatorname{Center}c_{q_1}(A_1^*)$.
    \end{step-goursat}
    \begin{step-goursat}
    Now, since $q_i\in \operatorname{Center}H_i$ and $\phi_i \in c_{q_i}(A_i^*)$ for $i=1,2$, by \cref{prop:qHG_3_coro} there are two $C^*$-Hopf algebras corresponding to the images of the previous projections:
    \begin{equation}
        Q_i
            \coloneqq
            c_{q_i}(H_i)=H_i q_i,
        \qquad
        \Delta_{Q_i}(x) \coloneqq
            \eval{\phi_i}{x_{(2)}} x_{(1)}q_i \otimes x_{(3)} q_i,
    \end{equation}
    and one can check that $1_{Q_i}=q_i$ and for all $x = \eval{\phi_i}{x_{(1)}}\eval{\phi_i}{x_{(3)}}x_{(2)} q_i\in Q_i$
    \begin{equation}
        \varepsilon_{Q_i}(x)
        \coloneqq
        \varepsilon|_{\smash{Q_i}}(x)
        =
        \eval{\phi_i}{x_{(1)}}
            \eval{\phi_i}{x_{(3)}}
                \eval{\varepsilon}{x_{(2)} q_i}
        =
        \eval{\phi_i}{x}.
    \end{equation}
    The Haar measure is then given  by the expression 
    \begin{equation} \label{eq:hathQiexpression}
        \hat{h}_{Q_i}
        =
        \frac{1}{\eval{\hat{h}}{q_i}}\hat{h}|_{\smash{Q_i}}
        =
        \frac{\eval{\hat{h}}{(\placeholder) q_i}}{\eval{\hat{h}}{q_i}},
    \end{equation}
    which satisfies the integral property and $\eval{\hat{h}_{Q_i}}{1_{Q_i}}=1$.
    \smallbreak\noindent%
    Finally, for $i=1, 2$, 
    since $c_{q_i}$ and $c_{\phi_i}$ commute by \cref{prop:qHG_3}, for all $x\in Q_i$,
    \begin{equation} \label{eq:invQunderProj}
        \eval{\phi_i}{x_{(1)}}x_{(2)}
            =
            \eval{\phi_i}{x_{(2)}} x_{(1)}
            =
            x,
        \quad\text{and}\quad
        q_i x = x q_i = x.
    \end{equation}
    \end{step-goursat}
    \begin{step-goursat}
    Let $\lambda:A_1\to A_{2}$ be defined, for all $x\in A_1$, by the expression
    \begin{align*}
        \lambda(x)
        \coloneqq
        \frac{1}{\eval{\phi}{1 \otimes h}}
        \eval{\phi}{x \otimes S(h_{(1)})} h_{(2)}.
    \end{align*}
    First, let us note that 
    \begin{equation}\label{eq:normalizationVarPhi}
        \eval{\hat{h}}{\lambda(1)}
        =
        \frac{\eval{\phi}{1\otimes S(h_{(1)})} \eval{\hat{h}}{h_{(2)}}}{\eval{\phi}{1\otimes h}}
        =
        \frac{1}{D_2\eval{\phi}{1\otimes h}},
    \end{equation}
    by virtue of \cref{eq:DualInt}.
    Moreover, one can recover $\phi$ as follows:
    \begin{equation}
    \eval{\phi}{x\otimes y}
    =
    \frac{\eval{\hat{h}}{\lambda (x)y}}{\eval{\hat{h}}{\lambda(1)}},
    \end{equation}
    as $\lambda$ is constructed similarly to $F_\phi$ from \cref{prop:Fphi} except for a factor $(D_2 \eval{\phi}{1\otimes h})^{-1}$.
    \smallbreak\noindent%
    We now show that $\lambda$ restricts to a map from $Q_1$ to $Q_2$, by proving that
    \begin{equation}\label{eq:lambdaPii}
        \lambda \circ \Pi_1
        =
        \Pi_2 \circ \lambda
        =
        \lambda.
    \end{equation}
    We make use of the invariance properties proven in Step~4.
    It suffices to prove that
    \[
        \lambda \circ L_\alpha
        =
        L_\alpha \circ \lambda
        =
        \lambda
        ,\quad\text{and}\quad
        \lambda\circ R_\alpha
        =
        R_\alpha\circ \lambda
        =
        \lambda
        ,\quad\text{for }
        \alpha = \phi_1,\phi_2, q_1, q_2.
    \]
    For $\alpha = q_1$, note that, or all $x\in A_1$,
    \[
        \lambda \circ L_{q_1}(x) = \frac{\eval{\phi}{q_1x \otimes S(h_{(1)})}h_{(2)}}{\eval{\phi}{1\otimes h}} = \frac{\eval{\phi}{x \otimes S(h_{(1)})} h_{(2)}}{\eval{\phi}{1\otimes h}} = \lambda(x),
    \]
    and analogously $\lambda \circ R_{q_1} = \lambda$.
    For $\alpha = q_2$, one can simply apply the second pulling-through property of \cref{eq:PTmodified} and then using \cref{prop:invariance_q_i_phi_i}:
    \begin{align*}
        L_{q_2} \circ \lambda(x)
        = \frac{\eval{\phi}{x \otimes S(h_{(1)})} q_2 h_{(2)}}{\eval{\phi}{1\otimes h}}
        = \frac{\eval{\phi}{x \otimes q_2S(h_{(1)})} h_{(2)}}{\eval{\phi}{1\otimes h}}
        = \lambda(x),
    \end{align*}
    and analogously $R_{q_2} \circ \lambda = \lambda$.
    For $\alpha = \phi_1$,  note that, using \cref{prop:invariance_q_i_phi_i}, for all $x\in A_1$:
    \[
        \lambda \circ L_{\phi_1}(x) = \frac{\eval{\phi}{L_{\phi_1}(x)\otimes S(h_{(1)})} h_{(2)}}{\eval{\phi}{1\otimes h}} = \frac{\eval{\phi}{x\otimes S(h_{(1)})} h_{(2)}}{\eval{\phi}{1\otimes h}} = \lambda(x),
    \]
    and analogously $\lambda \circ R_{\phi_1} = \lambda$.
    For $\alpha = \phi_2$, note that for all $x\in A_1$: 
    \begin{align*}
        L_{\phi_2} \circ \lambda(x) &= \eval{\phi}{1\otimes h}^{-1} \eval{\phi}{x\otimes S(h_{(1)})} L_{\phi_2}(h_{(2)})\\&= \eval{\phi}{1\otimes h}^{-1} \eval{\phi}{k\otimes 1}^{-1} \eval{\phi}{S(x)\otimes h_{(1)}} \eval{\phi}{k\otimes S(h_{(2)})} h_{(3)}\\
        &= \eval{\phi}{1\otimes h}^{-1} \eval{\phi}{k\otimes 1}^{-1} \eval{\phi}{S(x)_{(1)}\otimes h_{(1)}} \eval{\phi}{S(x)_{(2)}k\otimes h_{(2)}S(h_{(3)})}h_{(4)}\\
        &= \eval{\phi}{1\otimes h}^{-1} \eval{\phi}{k\otimes 1}^{-1} \eval{\phi}{S(x)\otimes h_{(1)}} \eval{\phi}{k\otimes 1}h_{(2)}\\
        &=\lambda(x),
    \end{align*}
    where we used the definition of $\lambda$ in the first equality, the invariance of $\phi$ under the antipode together with the definition of $L_{\phi_2}$ in the second, \cref{prop:IdempProp} with $v\equiv S(x) \otimes h_{(1)}$ and $w\equiv k\otimes S(h_{(2)})$ in the third, \cref{eq:Haar2} for $k$ together with \cref{eq:axiomSeps} and \cref{eq:axiom_eps_id} for $h$ in the fourth, and the invariance of $\phi$ under the antipode together with the definition of $\lambda$ in the last equality.
    \end{step-goursat}
    \begin{step-goursat}
    The candidate inverse map $\lambda' : A_2 \to A_1$ is defined, for all $y\in A_2$, by
    \begin{equation}
        \lambda'(y) \coloneqq
            \frac{1}{\eval{\phi}{k \otimes 1}}\eval{\phi}{S(k_{(1)}) \otimes y} k_{(2)}.
    \end{equation}
    Then, $\lambda$ and $\lambda'$ are mutually inverse restricted to $Q_1$ and $Q_2$.
    For all $x\in Q_1$, 
    \begin{align*}
        (\lambda' \circ \lambda)(x)
        &= \eval{\phi}{1\otimes h}^{-1}\eval{\phi}{k\otimes 1}^{-1}\eval{\phi}{x \otimes S(h_{(1)})} \eval{\phi}{S(k_{(1)}) \otimes h_{(2)}} k_{(2)}\\
        &= \eval{\phi}{1\otimes h}^{-1}\eval{\phi}{k\otimes 1}^{-1}\eval{\phi}{x \otimes h_{(1)}} \eval{\phi}{S(k_{(1)}) \otimes S(h_{(2)})} k_{(2)}\\
        &= \eval{\phi}{1\otimes h}^{-1}\eval{\phi}{k\otimes 1}^{-1}\eval{\phi}{x_{(1)} \otimes h_{(1)}} \eval{\phi}{x_{(2)}S(k_{(1)}) \otimes h_{(2)} S(h_{(3)})} k_{(2)}\\
        &= \eval{\phi}{1\otimes h}^{-1}\eval{\phi}{k\otimes 1}^{-1}\eval{\phi}{x_{(1)} \otimes h_{(1)}} \eval{\phi}{x_{(2)}S(k_{(1)}) \otimes \varepsilon(h_{(2)}) 1} k_{(2)}\\
        &= \eval{\phi}{1\otimes h}^{-1}\eval{\phi}{k\otimes 1}^{-1}\eval{\phi}{x_{(1)} \otimes h} \eval{\phi}{S(k_{(1)}) \otimes 1} k_{(2)}x_{(2)}\\
        &= \eval{\phi}{1\otimes h}^{-1}\eval{\phi}{k\otimes 1}^{-1}\eval{\phi}{x_{(1)} \otimes h} \eval{\phi}{k_{(1)} \otimes 1} k_{(2)}x_{(2)}\\
        &= \eval{\phi_1}{x_{(1)}} q_1 x_{(2)}\\
        &= x,
    \end{align*}
    where we used that $S(h_{(1)})\otimes h_{(2)}= h_{(1)}\otimes S(h_{(2)})$ in the second equality, \cref{prop:IdempProp} with $v\equiv x\otimes h_{(1)}$ and $w\equiv S(k_{(1)}) \otimes S(h_{(2)})$ in the third, \cref{eq:axiomSeps} in the fourth, \cref{eq:axiom_eps_id} and the second pulling-through property of \cref{eq:PTmodified} in the fifth, the invariance of $\phi$ under the antipode in the sixth and \cref{eq:invQunderProj} in the last. Hence $\lambda'\circ\lambda=\mathrm{Id}_{\smash{Q_1}}$. The equality $\lambda\circ\lambda'=\mathrm{Id}_{\smash{Q_2}}$ follows analogously.
    \end{step-goursat}
    \begin{step-goursat}
    Let us now check that $\lambda : Q_1 \to Q_2$ is an anti-homomorphism of $C^*$-Hopf algebras. Indeed, for all $x,y\in Q_1$, it holds that
    \begin{align*}
        \lambda(x)\lambda(y)
        &=\eval{\phi}{1\otimes h}^{-2}\eval{\phi}{x \otimes S(h_{(1)})} \eval{\phi}{y \otimes S(h_{(1')})} h_{(2)}h_{(2')}\\
        &=\eval{\phi}{1\otimes h}^{-2} \eval{\phi}{y \otimes h_{(1')}} \eval{\phi}{x \otimes S(h_{(1)})} h_{(2)}S(h_{(2')})\\
        &=\eval{\phi}{1\otimes h}^{-2} \eval{\phi}{y_{(1)} \otimes h_{(1')}} \eval{\phi}{y_{(2)} x \otimes h_{(2')} S(h_{(1)})} h_{(2)}S(h_{(3')})\\
        &=\eval{\phi}{1\otimes h}^{-2} \eval{\phi}{y_{(1)} \otimes h_{(1')}} \eval{\phi}{y_{(2)} x \otimes S(h_{(1)})} h_{(2)}h_{(2')}S(h_{(3')})\\
        &=\eval{\phi}{1\otimes h}^{-2} \eval{\phi}{y_{(1)} \otimes h_{(1')} \varepsilon(h_{(2')})} \eval{\phi}{y_{(2)} x \otimes S(h_{(1)})} h_{(2)}\\
        &=\eval{\phi}{1\otimes h}^{-2} \eval{\phi}{y_{(1)} \otimes h} \eval{\phi}{y_{(2)} x \otimes S(h_{(1)})} h_{(2)}\\
        &=\eval{\phi}{1\otimes h}^{-1} \eval{\phi_1}{y_{(1)}} \eval{\phi}{y_{(2)} x \otimes S(h_{(1)})} h_{(2)}\\
        &=\eval{\phi}{1\otimes h}^{-1} \eval{\phi}{yx \otimes S(h_{(1)})} h_{(2)}\\
        &=\lambda(yx),
    \end{align*}
    where we used the definition of $\lambda$ in the first equality,  that $S(h_{(1')})\otimes h_{(2')}= h_{(1')}\otimes S(h_{(2')})$ in the second, \cref{prop:IdempProp} with $v\equiv y\otimes h_{(1')}$ and $w\equiv x\otimes S(h_{(1)})$ in the third, the first pulling-through identity of \cref{eq:PTmodified} in the fourth, \cref{eq:axiomSeps} in the fifth, \cref{eq:axiom_eps_id} in the sixth, the definition of $\phi_1$ in the seventh, and \cref{eq:invQunderProj} as $y\in Q_1$ in the last.
    \smallbreak\noindent%
    We now prove that $\lambda$ is a $*$-map. Indeed, for all $x\in Q_1$, 
    \[
        \lambda(x^*)
        =
        \frac{\eval{\phi}{x^*\otimes S(h_{(1)})}h_{(2)}}{\eval{\phi}{1 \otimes h}}
        =
        \frac{\overline{\eval{\phi}{x\otimes S(h_{(1)})^*}}h_{(2)}}{\eval{\phi}{1 \otimes h}}
        =
        \frac{\overline{\eval{\phi}{x\otimes S(h_{(1)})}}h_{(2)}^*}{\eval{\phi}{1 \otimes h}}
        =
        \lambda(x)^*.
    \]
    \smallbreak\noindent%
    Now, recall that $1_{Q_i}=q_i$ for $i=1,2$, and hence $\lambda$ is unital:
    \[
        \lambda(q_1)
        = \frac{\eval{\phi}{k_{(1)}\otimes 1}\eval{\phi}{k_{(2)}\otimes S(h_{(1)})}h_{(2)}}{\eval{\phi}{k\otimes 1} \eval{\phi}{1 \otimes h}}
        = \frac{\eval{\phi}{k\otimes 1}\eval{\phi}{1\otimes S(h_{(1)})}h_{(2)}}{\eval{\phi}{k\otimes 1} \eval{\phi}{1 \otimes h}}
        = q_2,
    \]
    as a consequence of \cref{prop:IdempProp} with $v\equiv k\otimes 1$ and $w\equiv 1\otimes S(h_{(1)})$.
    \smallbreak\noindent%
    For the comultiplicativity of $\lambda$, we show that 
    \begin{equation}\label{eq:ComultVarphi}
        (\lambda\otimes\lambda)\circ\Delta_{Q_1}
        = (\lambda\otimes\lambda)\circ\Delta
        = \Delta_{Q_2} \circ \lambda.
    \end{equation}
    First, note that for all $x\in Q_1$ it holds that
    \begin{align}
        (\lambda \otimes \lambda) \circ \Delta_{Q_1}(x) 
        &=
        \eval{\phi_1}{x_{(2)}} \lambda(x_{(1)}q_1) \otimes \lambda(x_{(3)}q_1)
        \\
        &=
        \eval{\phi_1}{x_{(2)}} \lambda(x_{(1)}) \otimes \lambda(x_{(3)})
        \\
        &=
        (\lambda \circ R_{\phi_1})(x_{(1)}) \otimes \lambda(x_{(2)})
        \\
        &=
        \lambda(x_{(1)})\otimes \lambda(x_{(2)}),
    \end{align}
    where we used the definition of $\Delta_{Q_1}$ in the first equality, the invariance $\lambda \circ R_{q_1}= \lambda$ in the second, and the invariance $\lambda \circ R_{\phi_1}= \lambda$ in the last. This proves the first equality in \cref{eq:ComultVarphi}.
    \smallbreak\noindent%
    Note that, for all $y\in Q_2$, one can rewrite the comultiplication in $Q_2$ as
    \begin{equation}\label{eq:equiv_comult_Q}
        \Delta_{Q_2}(y)
        =  y_{(1)} q_2 \otimes \eval{\phi_2}{y_{(2)}} y_{(3)},
    \end{equation}
    as proven in \cref{eq:DeltaQsimpl}.
    Therefore, for all $x\in Q_1$,
    \begin{align*}
        \Delta_{Q_2} &\circ \lambda(x) 
        =
            \eval{\phi}{1\otimes h}^{-1}
                \eval{\phi}{x \otimes S(h_{(1)})}
                        h_{(2)} q_2 \otimes  \eval{\phi_2}{h_{(3)}} h_{(4)}
        \\
        &=
            \eval{\phi}{1\otimes h}^{-2}
                \eval{\phi}{x \otimes S(h_{(1)})}
                        \eval{\phi}{1\otimes h_{(1')}}
                            h_{(2)} h_{(2')} \otimes \eval{\phi_2}{h_{(3)}} h_{(4)}
        \\
        &=
            \eval{\phi}{1\otimes h}^{-2}
                \eval{\phi}{S(x) \otimes h_{(1)}}
                        \eval{\phi}{1\otimes S(h_{(2)})h_{(1')}}
                            h_{(2')} \otimes \eval{\phi_2}{h_{(3)}} h_{(4)}
        \\
        &=
            \eval{\phi}{1\otimes h}^{-2}
                \eval{\phi}{S(x)_{(1)} \otimes h_{(1)}}
                    \eval{\phi}{S(x)_{(2)}\otimes h_{(2)}S(h_{(3)})h_{(1')}}
                            h_{(2')} \otimes \eval{\phi_2}{h_{(4)}}h_{(5)}
        \\
        &=
            \eval{\phi}{1\otimes h}^{-2}
                \eval{\phi}{S(x)_{(1)} \otimes h_{(1)}}
                    \eval{\phi}{S(x)_{(2)}\otimes h_{(1')}}
                            h_{(2')} \otimes \eval{\phi_2}{h_{(2)}} h_{(3)}
        \\
        &=
            \eval{\phi}{1\otimes h}^{-2}
                \eval{\phi}{x_{(2)} \otimes S(h_{(1)})}
                    \eval{\phi}{x_{(1)}\otimes S(h_{(1')})}
                            h_{(2')} \otimes \eval{\phi_2}{h_{(2)}} h_{(3)}
        \\
        &=
        \lambda(x_{(1)}) \otimes (L_{\phi_2} \circ \lambda)(x_{(2)})
        \\
        &=
        \lambda(x_{(1)}) \otimes \lambda(x_{(2)}),
    \end{align*}
    where we used the definitions of $\lambda$ and \cref{eq:equiv_comult_Q} in the first equality, the definition of $q_2$ in the second, the first pulling-through of \cref{eq:PT} for $h_{(1')} \otimes h_{(2')}$ in the third, \cref{prop:IdempProp} with $v\equiv S(x)\otimes h_{(1)}$ and $w\equiv 1 \otimes S(h_{(2)})h_{(1')}$ in the fourth, \cref{eq:axiomSeps} together with \cref{eq:axiom_eps_id} in the fifth, the invariance property $\phi \circ S=\phi$ together with the anticomultiplicativity of $S$ in the sixth, the definitions of $\lambda$ and $L_{\phi_2}$ in the seventh and the invariance $L_{\phi_2} \circ \lambda= \lambda$ in the last.
    \smallbreak\noindent%
    Moreover, for all $x\in Q_1$,
    \begin{align*}
        \eval{\varepsilon_{Q_2}}{\lambda(x)}
        &=
        \eval{\phi}{1\otimes h}^{-1}
            \eval{\phi}{x \otimes S(h_{(1)})}
                \eval{\phi_2}{h_{(2)}}
        \\
        &=
        \eval{\phi}{1\otimes h}^{-1}
            \eval{\phi}{k\otimes 1}^{-1}
                \eval{\phi}{x \otimes S(h_{(1)})}
                    \eval{\phi}{k\otimes h_{(2)}}
        \\
        &=
        \eval{\phi}{1\otimes h}^{-1}
            \eval{\phi}{k\otimes 1}^{-1}
                \eval{\phi}{x_{(1)} \otimes S(h_{(1)})_{(1)}}
                    \eval{\phi}{x_{(2)}k\otimes S(h_{(1)})_{(2)}h_{(2)}}
        \\
        &=
        \eval{\phi}{1\otimes h}^{-1}
            \eval{\phi}{k\otimes 1}^{-1}
                \eval{\phi}{x \otimes S(h_{(2)})}
                    \eval{\phi}{k\otimes S(h_{(1)})h_{(3)}}
        \\
        &=
        \eval{\phi}{1\otimes h}^{-1}
            \eval{\phi}{k\otimes 1}^{-1}
                \eval{\phi}{x \otimes S(h_{(1)})}
                    \eval{\phi}{k\otimes S(h_{(3)})h_{(2)}}
        \\
        &=
        \eval{\phi}{1\otimes h}^{-1}
            \eval{\phi}{k\otimes 1}^{-1}
                \eval{\phi}{x \otimes S(\varepsilon(h_{(2)})h_{(1)})}
                    \eval{\phi}{k\otimes 1}
        \\
        &=
        \eval{\phi}{1\otimes h}^{-1}
            \eval{\phi}{x \otimes h}
        \\
        &=
        \eval{\varepsilon_{Q_1}}{x},
    \end{align*}
    where we used the definitions of $\lambda$ and $\varepsilon_{Q_2}$ in the first equality, the definition of $\phi_2$ in the second, \cref{prop:IdempProp} with $v\equiv x \otimes S(h_{(1)})$ and $w\equiv k\otimes h_{(2)}$ in the third, the anticomultiplicativity of $S$ in the fourth, the cyclicity of $h_{(1)} \otimes h_{(2)} \otimes h_{(3)}=h_{(3)} \otimes h_{(1)} \otimes h_{(2)}$ in the fifth, \cref{eq:axiomSeps} in the sixth, \cref{eq:axiom_eps_id} in the seventh and the definition of $\varepsilon_{Q_1}$ in the last.
    \smallbreak\noindent%
    Finally, we have the compatibility with the antipode. Let $x\in Q_1$,
    \begin{align*}
        (S_{Q_2} \circ \lambda)(x)
        &=
            \eval{\phi}{1\otimes h}^{-1}
                \eval{\phi}{x \otimes S(h_{(1)})}
                    S(h_{(2)})
        \\
        &=
            \eval{\phi}{1\otimes h}^{-1}
                \eval{\phi}{x \otimes h_{(1)}}
                    h_{(2)}
        \\
        &=
            \eval{\phi}{1\otimes h}^{-1}
                \eval{\phi}{S(x) \otimes S(h_{(1)})}
                    h_{(2)}
        \\
        &=
        (\lambda \circ S_{Q_1})(x),
    \end{align*}
    where we used the definitions of $\lambda$ and $S_{Q_2}$ in the first equality, the facts that $S(h)=h$ and $h_{(1)} \otimes h_{(2)}=h_{(2)} \otimes h_{(1)}$ from \cref{eq:Haar} in the third, the invariance of $\phi$ under the antipode in the fourth and the definitions of $\lambda$ and $S_{Q_1}$ in the last.
    \end{step-goursat}

\begin{step-goursat}
Conversely, for $i=1,2$, let $\phi_i\in A_i^*$ be idempotent states defining finite $*$-quantum hypergroups $H_i \coloneqq c_{\phi_i}(A_i)$, let $q_i \in H_i$ be group-like projections satisfying \cref{eq:QGL_qiprop}, and let $\lambda : H_1 q_1 \to H_2 q_2$ be an anti-isomorphism of $C^*$-Hopf algebras.
\smallbreak\noindent%
One can extend the domain of definition of $\lambda$ to $A_1$ by defining, for all $x\in A_1$,
\[
    \tilde{\lambda}(x)
    \coloneqq
    \lambda(c_{\phi_1}(x)q_1).
\]
For all $x\in A_1$ and $y\in A_2$, we define
\[
    \eval{\phi}{x \otimes y}
    \coloneqq
    \frac{1}{\eval{\hat{h}}{\tilde{\lambda}(1)}}
    \eval{\hat{h}}{\tilde{\lambda}(x) y}
    =
    \frac{1}{\eval{\hat{h}}{q_2}}
    \eval{\hat{h}}{\tilde{\lambda}(x) y},
\]
where $\tilde{\lambda}(1)=\lambda(q_1)=q_2$ because $1_{H_1 q_1}=q_1$, $1_{H_2 q_2}=q_2$, and $\lambda$ is unital.
\smallbreak\noindent%
First, let us prove that $\phi$ is a positive linear functional.
For that purpose, one can adapt the argument used in the proof of \cref{prop:Fphi}, where, for a completely positive map $F : A\to A$, the linear functional defined by $x\otimes y \in A\otimes A^{\mathrm{op}} \mapsto \eval{\hat{h}}{F(x)y}$ is positive.
Because $\tilde{\lambda} : A_1 \to A_2$ is the composition of a projection from $A_1$ to $H_1 q_1$ with a $C^*$-al\-ge\-bra anti-homomorphism from $H_1 q_1$ to $H_2 q_2$, we conclude that $\tilde\lambda$ is completely positive, and hence $\phi$ is a positive linear functional.
Note that $\phi$ is trivially unital.
\smallbreak\noindent%
Finally, $\phi$ is idempotent by the comultiplicativity of $\lambda$. Indeed, for all $x\in A_1$,
\begin{align} \label{eq:lambda_intertwine}
    (\tilde{\lambda} \otimes \tilde{\lambda}) \circ \Delta(x)
    &=
    (\lambda \otimes \lambda) \circ \Delta_{H_1q_1}(c_{\phi_1}(x)q_1)
    =
    \Delta_{H_2q_2} \circ \lambda(c_{\phi_1}(x)q_1)
    \\
    &=
    \eval{\phi_2}{\tilde{\lambda}(x)_{(1)}}
        \tilde{\lambda}(x)_{(2)}
        \otimes q_2 \tilde{\lambda}(x)_{(3)},
\end{align}
where we used \cref{eq:intertwine_Q} in the first equality, the comultiplicativity of $\lambda$ in the second, and the definition of $\tilde{\lambda}$ together with \cref{eq:DeltaQsimpl2} in the last.
Thus, for all $x\in A_1$ and $y \in A_2$,
\begin{align*}
    \eval{\phi}{x_{(1)}\otimes y_{(1)}}
        \eval{\phi}{x_{(2)}\otimes y_{(2)}}
    &=
    \eval{\hat{h}}{q_2}^{-2}
        \eval{\hat{h}}{\tilde{\lambda}(x_{(1)})y_{(1)}}
            \eval{\hat{h}}{\tilde{\lambda}(x_{(2)})y_{(2)}}
    \\
    &=
    \eval{\hat{h}}{q_2}^{-2}
        \eval{\phi_2}{\tilde{\lambda}(x)_{(1)}}
            \eval{\hat{h}}{\tilde{\lambda}(x)_{(2)}y_{(1)}}
                \eval{\hat{h}}{\tilde{\lambda}(x)_{(3)}y_{(2)} q_2}
    \\
    &=
        \eval{\hat{h}}{q_2}^{-2}
            \eval{\phi_2}{\tilde{\lambda}(x)_{(1)}}
                \eval{\hat{h}}{\tilde{\lambda}(x)_{(2)} y}
                    \eval{\hat{h}}{q_2}
    \\
    &=
        \eval{\hat{h}}{q_2}^{-2}
            \eval{\hat{h}}{\tilde\lambda(x) y}
                \eval{\hat{h}}{q_2}
    \\
    &=
    \eval{\phi}{x\otimes y},
\end{align*}
where we used the definition of $\phi$ in the first equality, \cref{eq:lambda_intertwine} together with the trace-like property of $\hat{h}$ to move $q_2$ in the second equality, \cref{prop:IdempProp} with $v\equiv \tilde{\lambda}(x)_{(2)}y$ and $w\equiv q_2$ in the third, the fact that $L_{\phi_2}(\tilde{\lambda}(x)) = \tilde{\lambda}(x)\in H_2 q_2$ in the fourth, and the definition of $\phi$ in the last.
\end{step-goursat}
\noindent%
This concludes the proof of \cref{thm:QuantumGoursat}.
\end{proof}

\subsection{Applications to renormalization trajectories}
Applying \cref{thm:QuantumGoursat} to the limiting idempotent states arising from the renormalization trajectories in \cref{sec:fixedpoints} provides a more explicit form of the corresponding matrix product density operator fixed points, completing the structural description of the renormalized states considered here.

\begin{remark}
    Note that for $\phi\in (A\otimes A^{\mathrm{op}})^*$, the trace-preserving condition \cref{eq:TPphiF},
    $
        \eval{\phi}{(\placeholder)\otimes 1} = \hat{h}
    $,
    gives the constraint $q_1=1_A$ and $Q_1=H_1$. Indeed,
    \[
    q_1=\frac{\eval{\phi}{h_{A(1)}\otimes 1} h_{A(2)}}{\eval{\phi}{h_A\otimes 1}}
    =\frac{\eval{\hat{h}_A}{h_{A(1)}}h_{A(2)}}{\eval{\hat{h}_A}{h_A}}=\frac{D^{-1}1_A}{D^{-1}}
    =1_A.
    \]
\end{remark}

\begin{theorem} \label{th:newGoursatFP}
Let $A$ be a $C^*$-Hopf algebra and let $\Phi$ and $\Psi$ be faithful $*$-representations of $A$ and $A^*$, respectively.
Assume that the renormalization procedure in \cref{sec:flows} converges, and $\phi\in (A\otimes A^{\mathrm{op}})^*$ characterizes its limit.
With the notation of \cref{thm:QuantumGoursat},
\begin{equation*}
    M_{\infty} =
             \sum_{i=1}^{\dim Q_1}   \Psi(e^i) \otimes \hat{b}(\hat{h}_{Q_1}) \tilde{\Phi}(e_i),
\end{equation*}
where $e_1,\ldots,e_{\dim Q_1}$ is a basis for $Q_1$, $e^{1},\ldots,e^{\dim Q_1}\in Q_1^*$ stands for its dual basis, and $\tilde{\Phi} \coloneqq \Phi \circ \lambda$ is a $*$-representation of $Q_1$.
\end{theorem}

\begin{proof}
Recall \cref{thm:QuantumGoursat} with $A_1 = A$ and $A_2 = A^{\mathrm{op}}$. Since 
$\lambda|_{\smash{Q_1}}$ is anti-multiplicative from
$Q_1$ to $Q_2$, and  $\Phi$ is a $*$-algebra anti-homomorphism from $A^{\mathrm{op}}$ to $\operatorname{End}\mathscr{H}$, it holds that $\tilde\Phi$ is a $*$-representation of $Q_1$.
Now, complete a basis $e_{1},\ldots,e_{\dim Q_1}$ of $Q_1$ to a basis $e_1,\ldots,e_{\dim A}$ of $A_1$.
The expression for $M_\infty$ becomes
\begin{align*}
    M_\infty
    &=
    \sum_{i=1}^{\dim A}
        D \eval{\phi}{e_i \otimes S(h_{(1)})} 
        \Psi(e^i)
        \otimes
        \hat{b}(\hat{h}_A) \Phi(h_{(2)})
    \\
    &=
    \sum_{i=1}^{\dim A}
        D \eval{\hat{h}_A}{\lambda(1)}^{-1}
        \eval{\hat{h}_A}{\lambda(e_i)S(h_{(1)})} 
        \Psi(e^i)
        \otimes
        \hat{b}(\hat{h}_A) \Phi(h_{(2)})
    \\
    &=
    \sum_{i=1}^{\dim A}
        D \eval{\hat{h}_A}{\lambda(1)}^{-1}
        \eval{\hat{h}_A}{S(h_{(1)})} 
        \Psi(e^i)
        \otimes
        \hat{b}(\hat{h}_A) \Phi(h_{(2)}\lambda(e_i))
    \\
    &=
    \sum_{i=1}^{\dim A}
        \eval{\hat{h}_A}{\lambda(1)}^{-1}
        \Psi(e^i)
        \otimes
        \hat{b}(\hat{h}_A) \Phi(\lambda(e_i))
    \\
    &=
    \sum_{i=1}^{\dim A}
        \Psi(e^i)
        \otimes
        \hat{b}(\hat{h}_{Q_1}) \Phi(\lambda(e_i))
    \\
    &=
    \sum_{i=1}^{\dim Q_1}
        \Psi(e^i)
        \otimes
        \hat{b}(\hat{h}_{Q_1}) \Phi(\lambda(e_i))
\end{align*}
where we recalled \cref{th:NewFixedPoints} in the first equality, we rewrote $\phi$ in terms of $\lambda$ by virtue of \cref{thm:QuantumGoursat} in the second, we applied the pulling-through identity \cref{eq:PTmodified} in the third, we employed \cref{eq:DualHaar} in the fourth, we rewrote the Haar integral of $Q_1$ as in \cref{eq:hathQiexpression} in the fifth, and used that $\lambda = \lambda \circ \Pi_1$ by \cref{eq:lambdaPii} in the last, where $\Pi_1$ is the projection of $A_1$ onto $Q_1$ introduced in \cref{prop:qHG_3}.
\end{proof}

\begin{acknowledgements}
D.P.-G.~and A.R.-d.-A.~acknowledge support from the Spanish Ministry of Science and Innovation MCIN/AEI/10.13039/501100011033 (CEX2023-001347-S, PID2020-113523GB-I00, PID2023-146758NB-I00), Universidad Complutense de Madrid (FEI-EU-22-06), the Ministry for Digital Transformation and the Civil Service of the Spanish Government through the QUANTUM ENIA project call -- Quantum Spain project, and the European Union through the Recovery, Transformation and Resilience Plan -- NextGenerationEU within the framework of the Digital Spain 2026 Agenda.
D.P.-G.\ also acknowledges support from Comunidad de Madrid (TEC-2024/COM-84 QUITEMAD-CM).
\end{acknowledgements}

\begin{conflictofinterest}
	On behalf of all authors, the corresponding author states that there is no conflict of interest.
\end{conflictofinterest}

\begin{dataavailability}
	Data sharing is not applicable to this article as no new data were created or analyzed in this study.
\end{dataavailability}

\bibliographystyle{abbrv}
\bibliography{renormalization}

@article{fannes_finitely_1992,
	title = {Finitely correlated states on quantum spin chains},
	volume = {144},
	issn = {1432-0916},
	url = {https://doi.org/10.1007/BF02099178},
	doi = {10.1007/BF02099178},
	language = {en},
	number = {3},
	urldate = {2025-12-10},
	journal = {Communications in Mathematical Physics},
	author = {Fannes, M. and Nachtergaele, B. and Werner, R. F.},
	year = {1992},
	pages = {443--490},
    note = {\href{https://doi.org/10.1007/BF02099178}{\nolinkurl{doi:10.1007/BF02099178}}}
}

@article{gell-mann_quantum_1954,
	title = {Quantum {Electrodynamics} at {Small} {Distances}},
	volume = {95},
	copyright = {http://link.aps.org/licenses/aps-default-license},
	issn = {0031-899X},
	url = {https://link.aps.org/doi/10.1103/PhysRev.95.1300},
	doi = {10.1103/PhysRev.95.1300},
	language = {en},
	number = {5},
	urldate = {2026-01-20},
	journal = {Physical Review},
	author = {Gell-Mann, M. and Low, F. E.},
	month = sep,
	year = {1954},
	pages = {1300--1312},
    note = {\href{https://doi.org/10.1103/PhysRev.95.1300}{\nolinkurl{doi:10.1103/PhysRev.95.1300}}}
}

@article{kadanoff_scaling_1966,
	title = {Scaling laws for {Ising} models near ${T}_c$},
	volume = {2},
	copyright = {https://link.aps.org/licenses/aps-default-license},
	issn = {0554-128X},
	url = {https://link.aps.org/doi/10.1103/PhysicsPhysiqueFizika.2.263},
	doi = {10.1103/PhysicsPhysiqueFizika.2.263},
	language = {en},
	number = {6},
	urldate = {2026-01-20},
	journal = {Physics Physique Fizika},
	author = {Kadanoff, Leo P.},
	month = jun,
	year = {1966},
	pages = {263--272},
    note = {\href{https://doi.org/10.1103/PhysicsPhysiqueFizika.2.263}{\nolinkurl{doi:10.1103/PhysicsPhysiqueFizika.2.263}}}
}

@article{wilson_renormalization_1971,
	title = {Renormalization {Group} and {Critical} {Phenomena}. {I}. {Renormalization} {Group} and the {Kadanoff} {Scaling} {Picture}},
	volume = {4},
	copyright = {http://link.aps.org/licenses/aps-default-license},
	issn = {0556-2805},
	url = {https://link.aps.org/doi/10.1103/PhysRevB.4.3174},
	doi = {10.1103/PhysRevB.4.3174},
	language = {en},
	number = {9},
	urldate = {2026-01-20},
	journal = {Physical Review B},
	author = {Wilson, Kenneth G.},
	month = nov,
	year = {1971},
	pages = {3174--3183},
    note = {\href{https://doi.org/10.1103/PhysRevB.4.3174}{\nolinkurl{doi:10.1103/PhysRevB.4.3174}}}
}

@article{wilson_renormalization_1971-1,
	title = {Renormalization {Group} and {Critical} {Phenomena}. {II}. {Phase}-{Space} {Cell} {Analysis} of {Critical} {Behavior}},
	volume = {4},
	copyright = {http://link.aps.org/licenses/aps-default-license},
	issn = {0556-2805},
	url = {https://link.aps.org/doi/10.1103/PhysRevB.4.3184},
	doi = {10.1103/PhysRevB.4.3184},
	language = {en},
	number = {9},
	urldate = {2026-01-20},
	journal = {Physical Review B},
	author = {Wilson, Kenneth G.},
	month = nov,
	year = {1971},
	pages = {3184--3205},
    note = {\href{https://doi.org/10.1103/PhysRevB.4.3184}{\nolinkurl{doi:10.1103/PhysRevB.4.3184}}}
}

@article{wilson_renormalization_1975,
	title = {The renormalization group: {Critical} phenomena and the {Kondo} problem},
	volume = {47},
	copyright = {http://link.aps.org/licenses/aps-default-license},
	issn = {0034-6861},
	shorttitle = {The renormalization group},
	url = {https://link.aps.org/doi/10.1103/RevModPhys.47.773},
	doi = {10.1103/RevModPhys.47.773},
	language = {en},
	number = {4},
	urldate = {2026-01-20},
	journal = {Reviews of Modern Physics},
	author = {Wilson, Kenneth G.},
	month = oct,
	year = {1975},
	pages = {773--840},
    note = {\href{https://doi.org/10.1103/RevModPhys.47.773}{\nolinkurl{doi:10.1103/RevModPhys.47.773}}}
}

@article{vidal_entanglement_2007,
	title = {Entanglement {Renormalization}},
	volume = {99},
	copyright = {http://link.aps.org/licenses/aps-default-license},
	issn = {0031-9007, 1079-7114},
	url = {https://link.aps.org/doi/10.1103/PhysRevLett.99.220405},
	doi = {10.1103/PhysRevLett.99.220405},
	language = {en},
	number = {22},
	urldate = {2026-01-23},
	journal = {Physical Review Letters},
	author = {Vidal, G.},
	month = nov,
	year = {2007},
	pages = {220405},
    note = {\href{https://doi.org/10.1103/PhysRevLett.99.220405}{\nolinkurl{doi:10.1103/PhysRevLett.99.220405}}. \href{https://arxiv.org/abs/cond-mat/0512165}{\nolinkurl{arXiv:cond-mat/0512165}}}
}

@article{vidal_class_2008,
	title = {Class of {Quantum} {Many}-{Body} {States} {That} {Can} {Be} {Efficiently} {Simulated}},
	volume = {101},
	copyright = {http://link.aps.org/licenses/aps-default-license},
	issn = {0031-9007, 1079-7114},
	url = {https://link.aps.org/doi/10.1103/PhysRevLett.101.110501},
	doi = {10.1103/PhysRevLett.101.110501},
	language = {en},
	number = {11},
	urldate = {2026-01-23},
	journal = {Physical Review Letters},
	author = {Vidal, G.},
	month = sep,
	year = {2008},
	pages = {110501},
    note = {\href{https://doi.org/10.1103/PhysRevLett.101.110501}{\nolinkurl{doi:10.1103/PhysRevLett.101.110501}}. \href{https://arxiv.org/abs/quant-ph/0610099}{\nolinkurl{arXiv:quant-ph/0610099}}}
}

@article{levin_tensor_2007,
	title = {Tensor {Renormalization} {Group} {Approach} to {Two}-{Dimensional} {Classical} {Lattice} {Models}},
	volume = {99},
	copyright = {http://link.aps.org/licenses/aps-default-license},
	issn = {0031-9007, 1079-7114},
	url = {https://link.aps.org/doi/10.1103/PhysRevLett.99.120601},
	doi = {10.1103/PhysRevLett.99.120601},
	language = {en},
	number = {12},
	urldate = {2026-01-23},
	journal = {Physical Review Letters},
	author = {Levin, Michael and Nave, Cody P.},
	month = sep,
	year = {2007},
	pages = {120601},
    note = {\href{https://doi.org/10.1103/PhysRevLett.99.120601}{\nolinkurl{doi:10.1103/PhysRevLett.99.120601}}. \href{https://arxiv.org/abs/cond-mat/0611687}{\nolinkurl{arXiv:cond-mat/0611687}}}
}

@article{evenbly_tensor_2015,
	title = {Tensor {Network} {Renormalization}},
	volume = {115},
	copyright = {http://link.aps.org/licenses/aps-default-license},
	issn = {0031-9007, 1079-7114},
	url = {https://link.aps.org/doi/10.1103/PhysRevLett.115.180405},
	doi = {10.1103/PhysRevLett.115.180405},
	language = {en},
	number = {18},
	urldate = {2026-01-23},
	journal = {Physical Review Letters},
	author = {Evenbly, G. and Vidal, G.},
	month = oct,
	year = {2015},
	pages = {180405},
    note = {\href{https://doi.org/10.1103/PhysRevLett.115.180405}{\nolinkurl{doi:10.1103/PhysRevLett.115.180405}}. \href{https://arxiv.org/abs/1412.0732}{\nolinkurl{arXiv:1412.0732}}}
}

@article{evenbly_tensor_2015-1,
	title = {Tensor {Network} {Renormalization} {Yields} the {Multiscale} {Entanglement} {Renormalization} {Ansatz}},
	volume = {115},
	copyright = {http://link.aps.org/licenses/aps-default-license},
	issn = {0031-9007, 1079-7114},
	url = {https://link.aps.org/doi/10.1103/PhysRevLett.115.200401},
	doi = {10.1103/PhysRevLett.115.200401},
	language = {en},
	number = {20},
	urldate = {2026-01-23},
	journal = {Physical Review Letters},
	author = {Evenbly, G. and Vidal, G.},
	month = nov,
	year = {2015},
	pages = {200401},
    note = {\href{https://doi.org/10.1103/PhysRevLett.115.200401}{\nolinkurl{doi:10.1103/PhysRevLett.115.200401}}. \href{https://arxiv.org/abs/1502.05385}{\nolinkurl{arXiv:1502.05385}}}
}

@article{kennedy_tensor_2022,
	title = {Tensor {RG} {Approach} to {High}-{Temperature} {Fixed} {Point}},
	volume = {187},
	issn = {0022-4715, 1572-9613},
	url = {https://link.springer.com/10.1007/s10955-022-02924-4},
	doi = {10.1007/s10955-022-02924-4},
	language = {en},
	number = {3},
	urldate = {2026-01-23},
	journal = {Journal of Statistical Physics},
	author = {Kennedy, Tom and Rychkov, Slava},
	month = jun,
	year = {2022},
	pages = {33},
    note = {\href{https://doi.org/10.1007/s10955-022-02924-4}{\nolinkurl{doi:10.1007/s10955-022-02924-4}}. \href{https://arxiv.org/abs/2107.11464}{\nolinkurl{arXiv:2107.11464}}}
}

@misc{ebel_tensor_2025,
	title = {Tensor {Renormalization} {Group} {Meets} {Computer} {Assistance}},
	url = {http://arxiv.org/abs/2506.03247},
	doi = {10.48550/arXiv.2506.03247},
	urldate = {2026-01-23},
	publisher = {arXiv},
	author = {Ebel, Nikolay and Kennedy, Tom and Rychkov, Slava},
	month = jun,
	year = {2025},
    note = {\href{https://arxiv.org/abs/2506.03247}{\nolinkurl{arXiv:2506.03247}}}
}

@article{kennedy_tensor_2024,
	title = {Tensor {Renormalization} {Group} at {Low} {Temperatures}: {Discontinuity} {Fixed} {Point}},
	volume = {25},
	issn = {1424-0637, 1424-0661},
	shorttitle = {Tensor {Renormalization} {Group} at {Low} {Temperatures}},
	url = {https://link.springer.com/10.1007/s00023-023-01289-y},
	doi = {10.1007/s00023-023-01289-y},
	language = {en},
	number = {1},
	urldate = {2026-01-23},
	journal = {Annales Henri Poincaré},
	author = {Kennedy, Tom and Rychkov, Slava},
	month = jan,
	year = {2024},
	pages = {773--841},
    note = {\href{https://doi.org/10.1007/s00023-023-01289-y}{\nolinkurl{doi:10.1007/s00023-023-01289-y}}. \href{https://arxiv.org/abs/2210.06669}{\nolinkurl{arXiv:2210.06669}}}
}

@article{ebel_rotations_2025,
	title = {Rotations, {Negative} {Eigenvalues}, and {Newton} {Method} in {Tensor} {Network} {Renormalization} {Group}},
	volume = {15},
	issn = {2160-3308},
	url = {https://link.aps.org/doi/10.1103/y3xz-t2w8},
	doi = {10.1103/y3xz-t2w8},
	language = {en},
	number = {3},
	urldate = {2026-01-23},
	journal = {Physical Review X},
	author = {Ebel, Nikolay and Kennedy, Tom and Rychkov, Slava},
	month = aug,
	year = {2025},
	pages = {031047},
    note = {\href{https://doi.org/10.1103/y3xz-t2w8}{\nolinkurl{doi:10.1103/y3xz-t2w8}}. \href{https://arxiv.org/abs/2408.10312}{\nolinkurl{arXiv:2408.10312}}}
}

@article{bohm_weak_1999,
	title = {Weak {Hopf} {Algebras} {I}: {Integral} {Theory} and {C}*-structure},
	author = {B\"{o}hm, Gabriella and Nill, Florian and Szlach\'{a}nyi, Korn\'{e}l},
	volume = {221},
	issn = {00218693},
	doi = {10.1006/jabr.1999.7984},
	number = {2},
	journal = {Journal of Algebra},
	month = nov,
	year = {1999},
	pages = {385--438},
    note = {\href{https://doi.org/10.1006/jabr.1999.7984}{\nolinkurl{doi:10.1006/jabr.1999.7984}}. \href{https://arxiv.org/abs/math/9805116}{\nolinkurl{arXiv:math/9805116}}}
}

@article{white_density_1992,
	title = {Density matrix formulation for quantum renormalization groups},
	volume = {69},
	copyright = {http://link.aps.org/licenses/aps-default-license},
	issn = {0031-9007},
	url = {https://link.aps.org/doi/10.1103/PhysRevLett.69.2863},
	doi = {10.1103/PhysRevLett.69.2863},
	language = {en},
	number = {19},
	urldate = {2025-12-08},
	journal = {Physical Review Letters},
	author = {White, Steven R.},
	month = nov,
	year = {1992},
	pages = {2863--2866},
    note = {\href{https://doi.org/10.1103/PhysRevLett.69.2863}{\nolinkurl{doi:10.1103/PhysRevLett.69.2863}}}
}

@article{ostlund_thermodynamic_1995,
	title = {Thermodynamic {Limit} of {Density} {Matrix} {Renormalization}},
	volume = {75},
	copyright = {http://link.aps.org/licenses/aps-default-license},
	issn = {0031-9007, 1079-7114},
	url = {https://link.aps.org/doi/10.1103/PhysRevLett.75.3537},
	doi = {10.1103/PhysRevLett.75.3537},
	language = {en},
	number = {19},
	urldate = {2025-12-08},
	journal = {Physical Review Letters},
	author = {Östlund, Stellan and Rommer, Stefan},
	month = nov,
	year = {1995},
	pages = {3537--3540},
    note = {\href{https://doi.org/10.1103/PhysRevLett.75.3537}{\nolinkurl{doi:10.1103/PhysRevLett.75.3537}}}
}

@article{schollwock_density-matrix_2005,
	title = {The density-matrix renormalization group},
	volume = {77},
	copyright = {http://link.aps.org/licenses/aps-default-license},
	issn = {0034-6861, 1539-0756},
	url = {https://link.aps.org/doi/10.1103/RevModPhys.77.259},
	doi = {10.1103/RevModPhys.77.259},
	language = {en},
	number = {1},
	urldate = {2025-12-08},
	journal = {Reviews of Modern Physics},
	author = {Schollwöck, U.},
	month = apr,
	year = {2005},
	pages = {259--315},
    note = {\href{https://doi.org/10.1103/RevModPhys.77.259}{\nolinkurl{doi:10.1103/RevModPhys.77.259}}. \href{https://arxiv.org/abs/cond-mat/0409292}{\nolinkurl{arXiv:cond-mat/0409292}}}
}

@article{schollwock_density-matrix_2011,
	title = {The density-matrix renormalization group in the age of matrix product states},
	volume = {326},
	copyright = {https://www.elsevier.com/tdm/userlicense/1.0/},
	issn = {00034916},
	url = {https://linkinghub.elsevier.com/retrieve/pii/S0003491610001752},
	doi = {10.1016/j.aop.2010.09.012},
	language = {en},
	number = {1},
	urldate = {2025-12-08},
	journal = {Annals of Physics},
	author = {Schollwöck, Ulrich},
	month = jan,
	year = {2011},
	pages = {96--192},
    note = {\href{https://doi.org/10.1016/j.aop.2010.09.012}{\nolinkurl{doi:10.1016/j.aop.2010.09.012}}. \href{https://arxiv.org/abs/1008.3477}{\nolinkurl{arXiv:1008.3477}}}
}

@article{hastings_area_2007,
	title = {An area law for one-dimensional quantum systems},
	volume = {2007},
	issn = {1742-5468},
	url = {https://iopscience.iop.org/article/10.1088/1742-5468/2007/08/P08024},
	doi = {10.1088/1742-5468/2007/08/P08024},
	number = {08},
	urldate = {2025-12-08},
	journal = {Journal of Statistical Mechanics: Theory and Experiment},
	author = {Hastings, Matthew B.},
	month = aug,
	year = {2007},
	pages = {P08024--P08024},
    note = {\href{https://doi.org/10.1088/1742-5468/2007/08/P08024}{\nolinkurl{doi:10.1088/1742-5468/2007/08/P08024}}. \href{https://arxiv.org/abs/0705.2024}{\nolinkurl{arXiv:0705.2024}}}
}

@article{verstraete_matrix_2008,
	title = {Matrix product states, projected entangled pair states, and variational renormalization group methods for quantum spin systems},
	volume = {57},
	issn = {0001-8732, 1460-6976},
	url = {http://www.tandfonline.com/doi/abs/10.1080/14789940801912366},
	doi = {10.1080/14789940801912366},
	language = {en},
	number = {2},
	urldate = {2025-12-08},
	journal = {Advances in Physics},
	author = {Verstraete, F. and Murg, V. and Cirac, J.I.},
	month = mar,
	year = {2008},
	pages = {143--224},
    note = {\href{https://doi.org/10.1080/14789940801912366}{\nolinkurl{doi:10.1080/14789940801912366}}. \href{https://arxiv.org/abs/0907.2796}{\nolinkurl{arXiv:0907.2796}}}
}

@article{cirac_unitaries_2017,
	title = {Matrix product unitaries: structure, symmetries, and topological invariants},
	volume = {2017},
	issn = {1742-5468},
	shorttitle = {Matrix product unitaries},
	url = {https://iopscience.iop.org/article/10.1088/1742-5468/aa7e55},
	doi = {10.1088/1742-5468/aa7e55},
	number = {8},
	urldate = {2025-12-08},
	journal = {Journal of Statistical Mechanics: Theory and Experiment},
	author = {Ignacio Cirac, J and Perez-Garcia, David and Schuch, Norbert and Verstraete, Frank},
	month = aug,
	year = {2017},
	pages = {083105},
    note = {\href{https://doi.org/10.1088/1742-5468/aa7e55}{\nolinkurl{doi:10.1088/1742-5468/aa7e55}}. \href{https://arxiv.org/abs/1703.09188}{\nolinkurl{arXiv:1703.09188}}}
}

@article{pirvu_matrix_2010,
	title = {Matrix product operator representations},
	volume = {12},
	issn = {1367-2630},
	url = {https://iopscience.iop.org/article/10.1088/1367-2630/12/2/025012},
	doi = {10.1088/1367-2630/12/2/025012},
	number = {2},
	urldate = {2025-12-08},
	journal = {New Journal of Physics},
	author = {Pirvu, B and Murg, V and Cirac, J I and Verstraete, F},
	year = {2010},
	pages = {025012},
    note = {\href{https://doi.org/10.1088/1367-2630/12/2/025012}{\nolinkurl{doi:10.1088/1367-2630/12/2/025012}}. \href{https://arxiv.org/abs/0804.3976}{\nolinkurl{arXiv:0804.3976}}}
}

@article{schuch_classifying_2011,
	title = {Classifying quantum phases using matrix product states and projected entangled pair states},
	volume = {84},
	copyright = {http://link.aps.org/licenses/aps-default-license},
	issn = {1098-0121, 1550-235X},
	url = {https://link.aps.org/doi/10.1103/PhysRevB.84.165139},
	doi = {10.1103/PhysRevB.84.165139},
	language = {en},
	number = {16},
	urldate = {2026-01-23},
	journal = {Physical Review B},
	author = {Schuch, Norbert and Pérez-García, David and Cirac, Ignacio},
	month = oct,
	year = {2011},
	pages = {165139},
    note = {\href{https://doi.org/10.1103/PhysRevB.84.165139}{\nolinkurl{doi:10.1103/PhysRevB.84.165139}}. \href{https://arxiv.org/abs/1010.3732}{\nolinkurl{arXiv:1010.3732}}}
}

@article{pollmann_symmetry_2012,
	title = {Symmetry protection of topological phases in one-dimensional quantum spin systems},
	volume = {85},
	copyright = {http://link.aps.org/licenses/aps-default-license},
	issn = {1098-0121, 1550-235X},
	url = {https://link.aps.org/doi/10.1103/PhysRevB.85.075125},
	doi = {10.1103/PhysRevB.85.075125},
	language = {en},
	number = {7},
	urldate = {2026-01-23},
	journal = {Physical Review B},
	author = {Pollmann, Frank and Berg, Erez and Turner, Ari M. and Oshikawa, Masaki},
	month = feb,
	year = {2012},
	pages = {075125},
    note = {\href{https://doi.org/10.1103/PhysRevB.85.075125}{\nolinkurl{doi:10.1103/PhysRevB.85.075125}}. \href{https://arxiv.org/abs/0909.4059}{\nolinkurl{arXiv:0909.4059}}}
}

@article{zwolak_mixed-state_2004,
	title = {Mixed-{State} {Dynamics} in {One}-{Dimensional} {Quantum} {Lattice} {Systems}: {A} {Time}-{Dependent} {Superoperator} {Renormalization} {Algorithm}},
	volume = {93},
	copyright = {http://link.aps.org/licenses/aps-default-license},
	issn = {0031-9007, 1079-7114},
	shorttitle = {Mixed-{State} {Dynamics} in {One}-{Dimensional} {Quantum} {Lattice} {Systems}},
	url = {https://link.aps.org/doi/10.1103/PhysRevLett.93.207205},
	doi = {10.1103/PhysRevLett.93.207205},
	language = {en},
	number = {20},
	urldate = {2026-01-23},
	journal = {Physical Review Letters},
	author = {Zwolak, Michael and Vidal, Guifré},
	month = nov,
	year = {2004},
	pages = {207205},
    note = {\href{https://doi.org/10.1103/PhysRevLett.93.207205}{\nolinkurl{doi:10.1103/PhysRevLett.93.207205}}. \href{https://arxiv.org/abs/cond-mat/0406440}{\nolinkurl{arXiv:cond-mat/0406440}}}
}

@article{bohm_weak_2000,
	title = {Weak {Hopf} {Algebras} {II}: {Representation} theory, dimensions and the {Markov} trace},
	volume = {233},
	issn = {00218693},
	url = {https://linkinghub.elsevier.com/retrieve/pii/S0021869300983790},
	doi = {10.1006/jabr.2000.8379},
	number = {1},
	urldate = {2024-01-23},
	journal = {Journal of Algebra},
	author = {B\"{o}hm, Gabriella and Szlach\'{a}nyi, Korn\'{e}l},
	month = nov,
	year = {2000},
	pages = {156--212},
    note = {\href{https://doi.org/10.1006/jabr.2000.8379}{\nolinkurl{doi:10.1006/jabr.2000.8379}}. \href{https://arxiv.org/abs/math/9906045}{\nolinkurl{arXiv:math/9906045}}}
}

@article{etingof_fusion_2005,
	title = {On fusion categories},
	volume = {162},
	issn = {0003-486X},
	doi = {10.4007/annals.2005.162.581},
	number = {2},
	journal = {Annals of Mathematics},
	author = {Etingof, Pavel and Nikshych, Dmitri and Ostrik, Viktor},
	month = sep,
	year = {2005},
	pages = {581--642},
    note = {\href{https://doi.org/10.4007/annals.2005.162.581}{\nolinkurl{doi:10.4007/annals.2005.162.581}}. \href{https://arxiv.org/abs/math/0203060}{\nolinkurl{arXiv:math/0203060}}}
}

@misc{molnar_matrix_2022,
	title = {Matrix product operator algebras {I}: representations of weak {Hopf} algebras and projected entangled pair states},
	shorttitle = {Matrix product operator algebras {I}},
	author = {Moln{\'a}r, Andr{\'a}s and {Ruiz-de-Alarc{\'o}n}, Alberto and Garre-Rubio, Jos{\'e} and Schuch, Norbert and Cirac, J. Ignacio and P{\'e}rez-García, David},
	month = jun,
	year = {2022},
	note = {\href{https://arxiv.org/abs/2204.05940}{\nolinkurl{arXiv:2204.05940}}},
}

@article{sahinoglu_characterizing_2021,
	title = {Characterizing {Topological} {Order} with {Matrix} {Product} {Operators}},
	volume = {22},
	issn = {1424-0637, 1424-0661},
	doi = {10.1007/s00023-020-00992-4},
	number = {2},
	journal = {Annales Henri Poincar\'{e}},
	author = {\c{S}ahino\u{g}lu, {Mehmet Burak} and Williamson, Dominic and Bultinck, Nick and Mari\"{e}n, Micha\"{e}l and Haegeman, Jutho and Schuch, Norbert and Verstraete, Frank},
	month = feb,
	year = {2021},
	pages = {563--592},
    note = {\href{https://doi.org/10.1007/s00023-020-00992-4}{\nolinkurl{doi:10.1007/s00023-020-00992-4}}. \href{https://arxiv.org/abs/1409.2150}{\nolinkurl{arXiv:1409.2150}}}
}

@article{schuch_peps_2010,
	title = {{PEPS} as ground states: {Degeneracy} and topology},
	volume = {325},
	issn = {00034916},
	doi = {10.1016/j.aop.2010.05.008},
	number = {10},
	journal = {Annals of Physics},
	author = {Schuch, Norbert and Cirac, Ignacio and P\'{e}rez-García, David},
	month = oct,
	year = {2010},
	pages = {2153--2192},
    note = {\href{https://doi.org/10.1016/j.aop.2010.05.008}{\nolinkurl{doi:10.1016/j.aop.2010.05.008}}. \href{https://arxiv.org/abs/1001.3807}{\nolinkurl{arXiv:1001.3807}}}
}

@article{bultinck_anyons_2017,
	title = {Anyons and matrix product operator algebras},
	volume = {378},
	issn = {00034916},
	url = {https://linkinghub.elsevier.com/retrieve/pii/S0003491617300040},
	doi = {10.1016/j.aop.2017.01.004},
	urldate = {2024-01-23},
	journal = {Annals of Physics},
	author = {Bultinck, Nick and Mari\"{e}n, Michael and Williamson, Dominic J. and \c{S}ahino\u{g}lu, Mehmet Burak and Haegeman, Jutho and Verstraete, Frank},
	month = mar,
	year = {2017},
	pages = {183--233},
    note = {\href{https://doi.org/10.1016/j.aop.2017.01.004}{\nolinkurl{doi:10.1016/j.aop.2017.01.004}}. \href{https://arxiv.org/abs/1511.08090}{\nolinkurl{arXiv:1511.08090}}}
}

@article {kac_finite_1966,
    author = {Kac, Georgiy Isaakovich and Paljutkin, Vladimir Georgievich},
    title = {Finite ring groups},
    journal = {Trudy Moskov. Mat. Ob\v{s}\v{c}.},
    fjournal = {Trudy Moskovskogo Matemati\v{c}eskogo Ob\v{s}\v{c}estva},
    volume = {15},
    year = {1966},
    pages = {224--261},
    issn = {0134-8663},
    mrclass = {46.65 (22.65)},
    mrnumber = {208401},
}

@incollection {montgomery_representation_2001,
    author = {Montgomery, Susan},
    title = {Representation theory of semisimple {H}opf algebras},
    booktitle = {Algebra - representation theory ({C}onstanta, 2000)},
    series = {NATO Sci. Ser. II Math. Phys. Chem.},
    volume = {28},
    pages = {189--218},
    publisher = {Kluwer Acad. Publ., Dordrecht},
    year = {2001},
    isbn = {0-7923-7113-5},
    mrclass = {16W30},
    mrnumber = {1858037},
    doi = {10.1007/978-94-010-0814-3\_9},
}

@book{montgomery_hopf_1993,
	address = {Providence, Rhode Island},
	series = {Regional {Conference} {Series} in {Mathematics}},
	title = {Hopf algebras and their actions on rings},
	isbn = {9781470424428},
	language = {eng},
	number = {Number 82},
	publisher = {Published for the Conference Board of the Mathematical Sciences by the American Mathematical Society},
	author = {Montgomery, Susan},
	year = {1993},
}

@book{etingof_tensor_2015,
	address = {Providence, Rhode Island},
	editor = {Etingof, P. I. and Gelaki, Shlomo and Nikshych, Dmitri and Ostrik, Victor},
	isbn = {9781470420246},
	number = {volume 205},
	publisher = {American Mathematical Society},
	series = {Mathematical surveys and monographs},
	title = {Tensor categories},
	year = {2015},
}

@article{ruiz_matrix_2024,
    title = {Matrix product operator algebras {II}: phases of matter for {1D} mixed states},
    author = {{Ruiz-de-Alarc{\'o}n}, Alberto and Garre-Rubio, Jos{\'e} and Moln{\'a}r, Andr{\'a}s and P{\'e}rez-Garc{\'i}a, David},
    volume = {114},
    issn = {1573-0530},
    shorttitle = {Matrix product operator algebras {II}},
    doi = {10.1007/s11005-024-01778-z},
    language = {en},
    number = {2},
    journal = {Letters in Mathematical Physics},
    month = mar,
    year = {2024},
    pages = {43},
    archivePrefix = {arXiv},
    eprint = {2204.06295},
    primaryClass  = {quant-ph},
    note = {\href{https://doi.org/10.1007/s11005-024-01778-z}{\nolinkurl{doi:10.1007/s11005-024-01778-z}}. \href{https://arxiv.org/abs/2204.06295}{\nolinkurl{arXiv:2204.06295}}}
}

@article{cirac_matrix_2021,
	title = {Matrix product states and projected entangled pair states: {Concepts}, symmetries, theorems},
	volume = {93},
	shorttitle = {Matrix product states and projected entangled pair states},
	doi = {10.1103/RevModPhys.93.045003},
	number = {4},
	journal = {Reviews of Modern Physics},
	author = {Cirac, J. Ignacio and P{\'e}rez-Garc{\'i}a, David and Schuch, Norbert and Verstraete, Frank},
	month = dec,
	year = {2021},
	pages = {045003},
    note = {\href{https://doi.org/10.1103/RevModPhys.93.045003}{\nolinkurl{doi:10.1103/RevModPhys.93.045003}}. \href{https://arxiv.org/abs/2011.12127}{\nolinkurl{arXiv:2011.12127}}}
}

@misc{kato_exact_2024,
	title = {Exact renormalization group flow for matrix product density operators},
	url = {http://arxiv.org/abs/2410.22696},
	doi = {10.48550/arXiv.2410.22696},
	publisher = {arXiv},
	author = {Kato, Kohtaro},
	month = oct,
	year = {2024},
    note = {\href{https://arxiv.org/abs/2410.22696}{\nolinkurl{arXiv:2410.22696}}}
}

@article{chirvasitu_integrals_2020,
	title = {Integrals in {Left} {Coideal} {Subalgebras} and {Group}-{Like} {Projections}},
	volume = {23},
	issn = {1572-9079},
	url = {https://doi.org/10.1007/s10468-019-09889-1},
	doi = {10.1007/s10468-019-09889-1},
	language = {en},
	number = {4},
	urldate = {2025-01-31},
	journal = {Algebras and Representation Theory},
	author = {Chirvasitu, Alexandru and Kasprzak, Paweł and Szulim, Piotr},
	month = aug,
	year = {2020},
	pages = {1499--1522},
    note = {\href{https://doi.org/10.1007/s10468-019-09889-1}{\nolinkurl{doi:10.1007/s10468-019-09889-1}}. \href{https://arxiv.org/abs/1808.01922}{\nolinkurl{arXiv:1808.01922}}}
}

@misc{landstad_compact_2007,
	title = {Compact and discrete subgroups of algebraic quantum groups {I}},
	url = {http://arxiv.org/abs/math/0702458},
	doi = {10.48550/arXiv.math/0702458},
	urldate = {2025-01-31},
	publisher = {arXiv},
	author = {Landstad, M. B. and Van Daele, A.},
	month = apr,
	year = {2007},
    note = {\href{https://arxiv.org/abs/math/0702458}{\nolinkurl{arXiv:math/0702458}}}
}

@misc{vandaele_algebraic_2025,
	title = {Algebraic quantum groups and duality {I}},
	url = {http://arxiv.org/abs/2304.13448},
	doi = {10.48550/arXiv.2304.13448},
	urldate = {2026-06-20},
	publisher = {arXiv},
	author = {Van Daele, Alfons},
	month = nov,
	year = {2025},
    note = {\href{https://arxiv.org/abs/2304.13448}{\nolinkurl{arXiv:2304.13448}}}
}

@article{goursat_sur_1889,
	title = {Sur les substitutions orthogonales et les divisions régulières de l'espace},
	volume = {6},
	issn = {0012-9593, 1873-2151},
	url = {http://www.numdam.org/item?id=ASENS_1889_3_6__9_0},
	doi = {10.24033/asens.317},
	urldate = {2025-02-01},
	journal = {Annales scientifiques de l'École normale supérieure},
	author = {Goursat, Edouard},
	year = {1889},
	pages = {9--102},
    note = {\href{https://doi.org/10.24033/asens.317}{\nolinkurl{doi:10.24033/asens.317}}}
}

@incollection{dung_subgroups_2009,
	address = {Providence, Rhode Island},
	title = {Subgroups of direct products of groups, ideals and subrings of direct products of rings, and {Goursat}’s lemma},
	volume = {480},
	isbn = {9780821843703 9780821881590},
	url = {http://www.ams.org/conm/480/},
	language = {en},
	urldate = {2025-02-01},
	booktitle = {Contemporary {Mathematics}},
	publisher = {American Mathematical Society},
	author = {Anderson, D. D. and Camillo, V.},
	editor = {Dung, Nguyen Viet and Guerriero, Franco and Hammoudi, Lakhdar and Kanwar, Pramod},
	year = {2009},
	doi = {10.1090/conm/480/09364},
	pages = {1--12},
    note = {\href{https://doi.org/10.1090/conm/480/09364}{\nolinkurl{doi:10.1090/conm/480/09364}}}
}

@article{Franz_2009,
    title={On idempotent states on quantum groups},
    volume={322},
    ISSN={0021-8693},
    url={http://dx.doi.org/10.1016/j.jalgebra.2009.05.037},
    DOI={10.1016/j.jalgebra.2009.05.037},
    number={5},
    journal={Journal of Algebra},
    publisher={Elsevier BV},
    author={Franz, Uwe and Skalski, Adam},
    year={2009},
    month=sep,
    pages={1774–1802},
    note = {\href{https://doi.org/10.1016/j.jalgebra.2009.05.037}{\nolinkurl{doi:10.1016/j.jalgebra.2009.05.037}}. \href{https://arxiv.org/abs/0808.1683}{\nolinkurl{arXiv:0808.1683}}}
}

@misc{sun_anomalous_2025,
	title = {Anomalous matrix product operator symmetries and {1D} mixed-state phases},
	url = {http://arxiv.org/abs/2504.16985},
	doi = {10.48550/arXiv.2504.16985},
	urldate = {2025-10-19},
	publisher = {arXiv},
	author = {Sun, Xiao-Qi},
	month = sep,
	year = {2025},
    note = {\href{https://arxiv.org/abs/2504.16985}{\nolinkurl{arXiv:2504.16985}}}
}

@misc{liu_trading_2025,
	title = {Trading {Mathematical} for {Physical} {Simplicity}: {Bialgebraic} {Structures} in {Matrix} {Product} {Operator} {Symmetries}},
	shorttitle = {Trading {Mathematical} for {Physical} {Simplicity}},
	url = {http://arxiv.org/abs/2509.03600},
	doi = {10.48550/arXiv.2509.03600},
	urldate = {2025-10-19},
	publisher = {arXiv},
	author = {Liu, Yuhan and Molnar, Andras and Sun, Xiao-Qi and Verstraete, Frank and Kato, Kohtaro and Lootens, Laurens},
	month = sep,
	year = {2025},
    note = {\href{https://arxiv.org/abs/2509.03600}{\nolinkurl{arXiv:2509.03600}}}
}

@article{cirac_matrix_2017,
	title = {Matrix product density operators: {Renormalization} fixed points and boundary theories},
	volume = {378},
	issn = {00034916},
	shorttitle = {Matrix product density operators},
	url = {https://linkinghub.elsevier.com/retrieve/pii/S0003491616303013},
	doi = {10.1016/j.aop.2016.12.030},
	language = {en},
	urldate = {2025-11-27},
	journal = {Annals of Physics},
	author = {Cirac, J.I. and Pérez-García, D. and Schuch, N. and Verstraete, F.},
	month = mar,
	year = {2017},
	pages = {100--149},
    note = {\href{https://doi.org/10.1016/j.aop.2016.12.030}{\nolinkurl{doi:10.1016/j.aop.2016.12.030}}. \href{https://arxiv.org/abs/1606.00608}{\nolinkurl{arXiv:1606.00608}}}
}

@article{coser_classification_2019,
	title = {Classification of phases for mixed states via fast dissipative evolution},
	volume = {3},
	copyright = {https://creativecommons.org/licenses/by/4.0/},
	issn = {2521-327X},
	url = {https://quantum-journal.org/papers/q-2019-08-12-174/},
	doi = {10.22331/q-2019-08-12-174},
	language = {en},
	urldate = {2025-11-27},
	journal = {Quantum},
	author = {Coser, Andrea and Pérez-García, David},
	month = aug,
	year = {2019},
	pages = {174},
    note = {\href{https://doi.org/10.22331/q-2019-08-12-174}{\nolinkurl{doi:10.22331/q-2019-08-12-174}}. \href{https://arxiv.org/abs/1810.05092}{\nolinkurl{arXiv:1810.05092}}}
}

@article{chen_classification_2011,
	author = {Chen, Xie and Gu, Zheng-Cheng and Wen, Xiao-Gang},
	doi = {10.1103/PhysRevB.83.035107},
	issn = {1098-0121, 1550-235X},
	journal = {Physical Review B},
	language = {en},
	number = {3},
	pages = {035107},
	title = {Classification of gapped symmetric phases in one-dimensional spin systems},
	url = {https://link.aps.org/doi/10.1103/PhysRevB.83.035107},
	urldate = {2023-08-02},
	volume = {83},
	year = {2011},
    note = {\href{https://doi.org/10.1103/PhysRevB.83.035107}{\nolinkurl{doi:10.1103/PhysRevB.83.035107}}. \href{https://arxiv.org/abs/1008.3745}{\nolinkurl{arXiv:1008.3745}}}
}

@article{levin_string-net_2005,
	title = {String-net condensation: {A} physical mechanism for topological phases},
	volume = {71},
	copyright = {http://link.aps.org/licenses/aps-default-license},
	issn = {1098-0121, 1550-235X},
	shorttitle = {String-net condensation},
	url = {https://link.aps.org/doi/10.1103/PhysRevB.71.045110},
	doi = {10.1103/PhysRevB.71.045110},
	language = {en},
	number = {4},
	urldate = {2025-11-28},
	journal = {Physical Review B},
	author = {Levin, Michael A. and Wen, Xiao-Gang},
	year = {2005},
	pages = {045110},
    note = {\href{https://doi.org/10.1103/PhysRevB.71.045110}{\nolinkurl{doi:10.1103/PhysRevB.71.045110}}. \href{https://arxiv.org/abs/cond-mat/0404617}{\nolinkurl{arXiv:cond-mat/0404617}}}
}

@article{franz_new_2009,
	title = {A new characterisation of idempotent states on finite and compact quantum groups},
	volume = {347},
	issn = {1631-073X, 1778-3569},
	url = {https://comptes-rendus.academie-sciences.fr/mathematique/articles/10.1016/j.crma.2009.06.015/},
	doi = {10.1016/j.crma.2009.06.015},
	language = {en},
	number = {17-18},
	urldate = {2025-11-30},
	journal = {Comptes Rendus. Mathématique},
	author = {Franz, Uwe and Skalski, Adam},
	year = {2009},
	pages = {991--996},
    note = {\href{https://doi.org/10.1016/j.crma.2009.06.015}{\nolinkurl{doi:10.1016/j.crma.2009.06.015}}. \href{https://arxiv.org/abs/0906.2362}{\nolinkurl{arXiv:0906.2362}}}
}

@article{pal_counterexample_1996,
	title = {A counterexample on idempotent states on a compact quantum group},
	volume = {37},
	issn = {1573-0530},
	url = {https://doi.org/10.1007/BF00400140},
	doi = {10.1007/BF00400140},
	language = {en},
	number = {1},
	urldate = {2025-11-30},
	journal = {Letters in Mathematical Physics},
	author = {Pal, Arupkumar},
	month = may,
	year = {1996},
	pages = {75--77},
    note = {\href{https://doi.org/10.1007/BF00400140}{\nolinkurl{doi:10.1007/BF00400140}}}
}

@article{salmi_idempotent_2012,
	title = {{I}dempotent states on locally compact quantum groups},
	volume = {63},
	issn = {0033-5606, 1464-3847},
	url = {https://academic.oup.com/qjmath/article-lookup/doi/10.1093/qmath/har023},
	doi = {10.1093/qmath/har023},
	language = {en},
	number = {4},
	urldate = {2026-01-26},
	journal = {The Quarterly Journal of Mathematics},
	author = {Salmi, P. and Skalski, A.},
	year = {2012},
	pages = {1009--1032},
    note = {\href{https://doi.org/10.1093/qmath/har023}{\nolinkurl{doi:10.1093/qmath/har023}}. \href{https://arxiv.org/abs/1102.2051}{\nolinkurl{arXiv:1102.2051}}}
}

@incollection{vandaele_fourier_2007,
    author = {Van Daele, Alfons},
    title = {The {F}ourier transform in quantum group theory},
    booktitle = {New Techniques in {H}opf Algebras and {G}raded {R}ing {T}heory},
    publisher = {Koninklijke Vlaamse Academie van Belgi{\"e} voor Wetenschappen en Kunsten},
    address = {Brussels},
    year = {2007},
    pages = {187--196},
    note = {\href{https://arxiv.org/abs/math/0609502}{\nolinkurl{arXiv:math/0609502}}}
}

@article{landstad_finite_2025,
	title = {Finite quantum hypergroups},
	volume = {55},
	issn = {0035-7596},
	url = {https://projecteuclid.org/journals/rocky-mountain-journal-of-mathematics/volume-55/issue-6/FINITE-QUANTUM-HYPERGROUPS/10.1216/rmj.2025.55.1651.full},
	doi = {10.1216/rmj.2025.55.1651},
	number = {6},
	urldate = {2026-02-11},
	journal = {Rocky Mountain Journal of Mathematics},
	author = {Landstad, Magnus B. and Van Daele, Alfons},
	year = {2025},
    note = {\href{https://doi.org/10.1216/rmj.2025.55.1651}{\nolinkurl{doi:10.1216/rmj.2025.55.1651}}. \href{https://arxiv.org/abs/2209.13282}{\nolinkurl{arXiv:2209.13282}}}
}

@misc{kawada_probability_1940,
	title = {On the {Probability} {Distribution} on a {Compact} {Group}. {I}},
	author = {Kawada, Yukiyosi and It{\^o}, Kiyosi},
	url = {https://doi.org/10.11429/ppmsj1919.22.12_977},
	doi = {10.11429/ppmsj1919.22.12_977},
	language = {en},
    journal={Proceedings of the Physico-Mathematical Society of Japan. 3rd Series},
    volume={22},
    number={12},
    pages={977-998},
    year={1940},
    note = {\href{https://doi.org/10.11429/ppmsj1919.22.12_977}{\nolinkurl{doi:10.11429/ppmsj1919.22.12_977}}}
}

@article{franz_idempotent_2013,
	title = {Idempotent states on compact quantum groups and their classification on $\mathrm{U}_q(2)$, $\mathrm{SU}_q(2)$, and $\mathrm{SO}_q(3)$},
	volume = {7},
	issn = {1661-6952, 1661-6960},
	url = {https://ems.press/doi/10.4171/jncg/115},
	doi = {10.4171/jncg/115},
	number = {1},
	urldate = {2026-02-12},
	journal = {Journal of Noncommutative Geometry},
	author = {Franz, Uwe and Skalski, Adam and Tomatsu, Reiji},
	year = {2013},
	pages = {221--254},
    note = {\href{https://doi.org/10.4171/jncg/115}{\nolinkurl{doi:10.4171/jncg/115}}. \href{https://arxiv.org/abs/0903.2363}{\nolinkurl{arXiv:0903.2363}}}
}

@article{kitaev_fault-tolerant_2003,
	title = {Fault-tolerant quantum computation by anyons},
	volume = {303},
	number = {1},
	url = {https://doi.org/10.1016/S0003-4916(02)00018-0},
	doi = {10.1016/S0003-4916(02)00018-0},
	journal = {Annals of Physics},
	author = {Kitaev, A. Yu.},
	year = {2003},
	pages = {2--30},
	note = {\href{https://doi.org/10.1016/S0003-4916(02)00018-0}{\nolinkurl{doi:10.1016/S0003-4916(02)00018-0}}. \href{https://arxiv.org/abs/quant-ph/9707021}{\nolinkurl{arXiv:quant-ph/9707021}}}
}

@misc{liu_establishing_2026,
	title = {Establishing {Mixed}-{State} {Phase} {Equivalence} beyond {Renormalization} {Fixed} {Points}},
	url = {http://arxiv.org/abs/2605.04159},
	doi = {10.48550/arXiv.2605.04159},
	urldate = {2026-06-10},
	publisher = {arXiv},
	author = {Liu, Yuhan},
	year = {2026},
	note = {\href{https://arxiv.org/abs/2605.04159}{\nolinkurl{arXiv:2605.04159}}},
}

@article{cirac2011entanglement,
  title={Entanglement spectrum and boundary theories with projected entangled-pair states},
  author={Cirac, J. Ignacio and Poilblanc, Didier and Schuch, Norbert and Verstraete, Frank},
  journal={Physical Review B},
  volume={83},
  number={24},
  pages={245134},
  year={2011},
  publisher={APS}
}

@article{kastoryano2019locality,
  title={Locality at the boundary implies gap in the bulk for 2D PEPS},
  author={Kastoryano, Michael J. and Lucia, Angelo and Perez-Garcia, David},
  journal={Communications in Mathematical Physics},
  volume={366},
  number={3},
  pages={895--926},
  year={2019},
  publisher={Springer}
}

@article{ogata2025haag,
  title={{H}aag duality for 2D quantum spin systems},
  author={Ogata, Yoshiko and P{\'e}rez-Garc{\'\i}a, David and Ruiz-de-Alarc{\'o}n, Alberto},
  journal={arXiv preprint arXiv:2509.23734},
  year={2025}
}

\appendix

\section{Proofs of Section~\ref{sec:preliminaries}}
\label{appendix:proofs2}

\begin{proof}[Proof of \cref{rem:DaggerLR}]
    That $(L_x)^* = L_{x^*}$ is trivial since $\eval{\hat{h}}{(x^* a)^* b} = \eval{\hat{h}}{a^* x b} $.
    Now,
    \begin{align*}
        \innerprod{a}{L_{\varphi}(b)} = \innerprod{L_{\varphi^*}(a)}{b}
        &\iff
        \eval{\hat{h}}{a^* L_{\varphi}(b)} = \eval{\hat{h}}{(L_{\varphi^*}(a))^* b}
        \\
        &\iff
        \eval{\hat{h}}{a^* b_{(2)}}\eval{\varphi}{b_{(1)}} = \overline{\eval{\varphi^*}{a_{(1)}}}\eval{\hat{h}}{(a_{(2)})^* b}
        \\
        &\iff
        \eval{\hat{h}}{a^* b_{(2)}}\eval{\varphi}{b_{(1)}} = \eval{\varphi}{S(a_{(1)})^*}\eval{\hat{h}}{(a_{(2)})^* b}
        \\
        &\iff
        \eval{\hat{h}}{a^* b_{(2)}} \eval{\varphi}{b_{(1)}} = \eval{\varphi}{S((a^*)_{(1)})}\eval{\hat{h}}{(a^*)_{(2)} b}
    \end{align*}
    and the last equality holds by \cref{eq:propHaarExchXY} with $v\equiv a^*$ and $w\equiv b$. The identities for $R_x$ and $R_\varphi$ follow similarly.
\end{proof}

\begin{proof}[Proof of \cref{lemma:fourier}]
    Let us first note that, for all $x\in A$:
	\begin{align*}
		\jmath^{-1}(\jmath(x))
		&=
		D \eval{\jmath(x)}{h_{(1)}} S(h_{(2)})
		\\
        &=
		D \eval{\hat{h}}{h_{(1)}x} S(h_{(2)})
		=
		D \eval{\hat{h}}{h_{(1)}} S(h_{(2)}) x
		=
		x,
	\end{align*}
	where the first two equalities unfold the definitions of $\jmath^{-1}$ and $\jmath$, the third follows by the pulling-through identity in \cref{eq:PTmodified} and the last equality is due to the first identity in \cref{eq:DualHaar} and \cref{eq:DualInt}.
    Moreover, for all $\varphi\in A^*$:
	\begin{align*}
		\jmath(\jmath^{-1}(\varphi))
		  & =
		\eval{\hat{h}}{(\placeholder)\jmath^{-1}(\varphi)}
		\\
		  & =
		D  \eval{\varphi}{h_{(1)}} \eval{\hat{h}}{(\placeholder)S(h_{(2)})}
		=
		D \eval{\varphi}{h_{(1)}(\placeholder)}
		\eval{\hat{h}}{S(h_{(2)})}
		=
		\varphi,
	\end{align*}
	where in the third equality we have employed \cref{eq:PT}, while the last equality follows by \cref{eq:DualInt} and the last identity in \cref{eq:Haar}.
\end{proof}

\begin{proof}[Proof of \cref{lemma:sigma}]
We can rewrite the definition of $\sigma$ to obtain, for all $x,y\in A$:
\begin{align*}
    \sigma(x \otimes y)
      & =
    \jmath^{-1}( (\jmath(x)\otimes\jmath(y)) \circ \Delta )
    \\ &
    =
    D \eval{ (\jmath(x)\otimes \jmath(y)) \circ \Delta }{ h_{(1)} } S(h_{(2)})
    \\ &
    =
    D \eval{\jmath(x)}{ h_{(1)}}
        \eval{ \jmath(y) }{ h_{(2)}} S(h_{(3)})
    \\ &
    =
    D \eval{\hat{h}}{h_{(1)} x}
        \eval{\hat{h}}{h_{(2)} y} S(h_{(3)}),
\end{align*}
where the first two equalities unfold the definitions of $\sigma$ and $\jmath^{-1}$; after applying the comultiplication, the last equality follows from the definition of $\jmath$.
\smallbreak\noindent%
The previous representation facilitates the proof of the complete positivity of $\sigma$. Indeed, for all $x,y\in A$:
\begin{align*}
    \sigma(x \otimes y)
      & =
    D \eval{\hat{h}}{x h_{(1)} h_{(1')}^*}
        \eval{\hat{h}}{y h_{(2)} h_{(2')}^*}
            S(h_{(3)} h_{(3')}^*)
    \\ &
    =
    D \eval{\hat{h}}{h_{(1)}^* x h_{(1')}}
        \eval{\hat{h}}{h_{(2)}^* y h_{(2')}}
            S(h_{(3)}^*) S(h_{(3')})
    \\ &
    =
    D (\hat{h}\otimes \hat{h} \otimes \mathrm{Id}_A)
            (w^* (x\otimes y\otimes 1) w),
\end{align*}
where $w \coloneqq h_{(1)}\otimes h_{(2)}\otimes S(h_{(3)})$. In the first equality we used $h=hh^*$, see \cref{eq:Haar}; in the second equality we used the antimultiplicativity of $S$ and the trace-like property of $\hat{h}$, see \cref{eq:PropsAntipode,eq:DualHaar}.
\smallbreak\noindent%
Finally, \cref{eq:prophsigma} holds by virtue of the following calculation, valid for all $x,y,z\in A$:
\begin{align*}
    \eval{\hat{h}}{\sigma(x\otimes y)z}
      & =
    D \eval{\hat{h}}{h_{(1)} x}
    \eval{\hat{h}}{h_{(2)} y}
    \eval{\hat{h}}{S(h_{(3)})z}
    \\ &
    =
    D \eval{\hat{h}}{z_{(1)} h_{(1)} x}
    \eval{\hat{h}}{z_{(2)} h_{(2)} y}
    \eval{\hat{h}}{S(h_{(3)})}
    \\ &
    =
    D \eval{\hat{h}}{h_{(1)} x z_{(1)}}
    \eval{\hat{h}}{h_{(2)} y z_{(2)}}
    \eval{\hat{h}}{S(h_{(3)})}
    \\ &
    =
    \eval{\hat{h}}{x z_{(1)}}
    \eval{\hat{h}}{y z_{(2)}},
\end{align*}
where the first equality is due to the previous simplification of $\sigma(x\otimes y)$, in the second we have used the pulling-through identity \cref{eq:PT} (after comultiplying and applying the antipode), the third follows from the trace-like property in \cref{eq:DualHaar}, and the fourth follows from \cref{eq:DualInt}.
\smallbreak\noindent%
Finally, let us prove that $\sigma$ is an associative map. For that purpose, note that applying $\jmath$ on \cref{eq:def_sigma} yields
$
    \jmath(\sigma(x\otimes y)) = \jmath(x)\jmath(y)
$,
and hence, for all $x,y,z\in A$,
\begin{align*}
    \jmath (\sigma(\sigma(x\otimes y)\otimes z))
    &= \jmath(\sigma(x\otimes y))\jmath(z)
     = (\jmath(x)\jmath(y))\jmath(z) \\
    &= \jmath(x)(\jmath(y)\jmath(z))
     = \jmath(x)\jmath(\sigma(y\otimes z))
     = \jmath (\sigma(x\otimes \sigma(y\otimes z))).
\end{align*}
As $\jmath$ is injective, it follows that
$\sigma(\sigma(x\otimes y)\otimes z)=\sigma(x\otimes \sigma(y\otimes z))$.
\end{proof}

\begin{proof}[Proof of \cref{prop:sigma_is_S}]
The first identity is known; see Ref.~\cite[Theorem~3.2]{ruiz_matrix_2024}. For the second identity, note that, for all $x,y\in A$:
\begin{align*}
    \Coarse(\Phi_b(x)\otimes \Phi_b(y)) 
      &
    =
    D
    \operatorname{Tr}((\Phi\circ S)(h_{(3)}) \Phi_b(x) )
    \operatorname{Tr}((\Phi\circ S)(h_{(2)}) \Phi_b(y) ) 
    \Phi_b(h_{(1)})
    \\ &
    =
    D
    \operatorname{Tr}( \Phi_b(x S(h_{(3)})) )
    \operatorname{Tr}( \Phi_b(y S(h_{(2)})) ) 
    \Phi_b(h_{(1)})
    \\ &
    =
    D
    \eval{\hat{h}}{x S(h_{(3)})}
    \eval{\hat{h}}{y S(h_{(2)})}
    \Phi_b(h_{(1)})
    \\ &
    =
    D
    \eval{\hat{h}}{ x h_{(1)} }
    \eval{\hat{h}}{ y h_{(2)} }
    \Phi_b(S(h_{(3)}))
    \\ &
    =
    \Phi_b(\sigma( x \otimes y) ),
\end{align*}
where we unfolded the definition of $\Coarse$ in the first equality, we used cyclicity of the trace in the second, we employed \cref{eq:DefBdry-Dual} in the third, that
$h_{(1)}\otimes h_{(2)}\otimes h_{(3)} = S(h_{(3)})\otimes S(h_{(2)})\otimes S(h_{(1)})$,
using $h=S(h)$ and \cref{eq:Santicomult} in the fourth, and the trace-like property in \cref{eq:DualHaar} and the definition of $\sigma$ in the last.
\end{proof}

\end{document}